\documentclass[pra,aps,reprint,superscriptaddress,preprintnumbers,floatfix,nofootinbib,longbibliography]{revtex4-2}
\usepackage{wrapfig, blindtext}
\usepackage[most]{tcolorbox}
\tcbset{colback=yellow!10!white, colframe=red!50!black, 
  highlight math style= {enhanced, 
    colframe=red,colback=red!10!white,boxsep=0pt}
}
\usepackage{amsmath}
\usepackage{tikz-cd}
\usepackage{empheq}
\usepackage{graphicx}

\usepackage{amsthm,bbm}
\usepackage{thm-restate}
\usepackage{anyfontsize}

\usepackage[
colorlinks=true,
linkcolor=blue,
citecolor=blue,
filecolor=blue,
urlcolor=blue,
]{hyperref}
\usepackage[capitalise]{cleveref}
\usepackage{nameref}
\usepackage{orcidlink}
\usepackage{xpatch}
\makeatletter
\xpatchcmd{\@ssect@ltx}{\@xsect}{\protected@edef\@currentlabelname{#8}\@xsect}{}{}
\xpatchcmd{\@sect@ltx}{\@xsect}{\protected@edef\@currentlabelname{#8}\@xsect}{}{}
\makeatother
\usepackage{youngtab}

\usepackage{tikz}
\usetikzlibrary{decorations.pathreplacing}
\usepackage{pgfplots}
\pgfplotsset{compat=1.18} 

\newcommand{\U}{\operatorname{U}}
\newcommand{\SU}{\operatorname{SU}}
\newcommand{\Sp}{\operatorname{Sp}}
\newcommand{\su}{\mathfrak{su}}
\newcommand{\GL}{\operatorname{GL}}

\newcommand{\e}{\operatorname{e}}

\renewcommand{\i}{\mathrm{i}}

\renewcommand{\P}{\mathbf{P}}

\newcommand{\Tr}{\operatorname{tr}}
\newcommand{\sgn}{\operatorname{sgn}}
\newcommand{\V}{\mathcal{V}}

\renewcommand{\S}{\mathbb{S}}

\makeatletter 
\renewcommand\onecolumngrid{%
  \do@columngrid{one}{\@ne}%
  \def\set@footnotewidth{\onecolumngrid}%
  \def\footnoterule{\kern-6pt\hrule width 1.5in\kern6pt}%
}
\makeatother 

\newlength\dlf 

\usepackage{amsmath,amssymb}
\usepackage[shortlabels]{enumitem}
\usepackage{graphicx}
\usepackage{graphics}
\usepackage{amsmath}

\usepackage{color}
\usepackage{dsfont}
\usepackage{bbm}
\usepackage{stmaryrd}

\usepackage{lipsum}
\usepackage{lmodern}
\usepackage{tcolorbox}

\usepackage{tabularray}
\usepackage{adjustbox}
\UseTblrLibrary{amsmath}
\UseTblrLibrary{booktabs}
\UseTblrLibrary{varwidth}

\usepackage{amsfonts}
\usepackage{graphicx,graphics,epsfig,bm,bbm,amssymb,amsmath,amsfonts,mathrsfs}
\usepackage[normalem]{ulem}
\usepackage{subfigure}
\usepackage{dsfont}
\usepackage{upgreek}
\usepackage{tikz}
\usepackage{natbib}
\usepackage{chngcntr}

\usepackage{pifont}

\theoremstyle{definition}
\newtheorem{definition}{Definition}
\newtheorem*{definition*}{Definition}

\theoremstyle{plain}
\newtheorem{theorem}{Theorem}
\newtheorem*{theorem*}{Theorem}

\newtheorem{corollary}{Corollary}
\newtheorem*{corollary*}{Corollary}

\newtheorem{lemma}{Lemma}
\newtheorem*{lemma*}{Lemma}

\newtheorem{proposition}{Proposition}
\newtheorem*{proposition*}{Proposition}

\theoremstyle{remark}
\newtheorem{remark}{Remark}
\newtheorem{remark*}{Remark}

\newtheorem*{fact*}{Fact}

\newcommand{\bes} {\begin{subequations}}
  \newcommand{\ees} {\end{subequations}}
\newcommand{\bea} {\begin{eqnarray}}
  \newcommand{\eea} {\end{eqnarray}}
\newcommand{\be} {\begin{equation}}
  \newcommand{\ee} {\end{equation}}

\def\>{\rangle}
\def\<{\langle}
\def\Tr{\operatorname{Tr}}

\newcommand{\ident}{\mathbb{I}}
\newcommand{\ignore}[1]{}

\usepackage{amsmath,amssymb}
\usepackage{qcircuit}

\newcommand{\real}{\mathbb{R}}
\newcommand{\complex}{\mathbb{C}}
\newcommand{\hilbert}[1][H]{\mathcal{#1}}

\usepackage{mathtools}
\usepackage{microtype}

\DeclarePairedDelimiter{\sets}{\{}{\}}
\newcommand{\cset}[2]{\sets{#1 \, : \, #2}}

\DeclarePairedDelimiter{\nket}{\lvert}{\rangle}
\DeclarePairedDelimiter{\nbra}{\langle}{\rvert}

\newcommand{\diff}{\mathop{}\!\mathrm{d}}
\newcommand{\vct}[1]{\bm{{#1}}}

\providecommand\given{}
\newcommand\givensymbol[1][]{%
  \nonscript#1\vert
  \allowbreak
  \nonscript
  \mathopen{}
}
\newcommand\ngivensymbol[1][]{%
  \nonscript:
  \allowbreak
  \nonscript
  \mathopen{}
}

\DeclarePairedDelimiterX{\parens}[1]{\lparen}{\rparen}{%
  \renewcommand\given{\ngivensymbol[\delimsize]}
  #1
}
\DeclarePairedDelimiterX{\bracks}[1]{\lbrack}{\rbrack}{%
  \renewcommand\given{\givensymbol[\delimsize]}
  #1
}
\DeclarePairedDelimiterX{\braces}[1]{\lbrace}{\rbrace}{%
  \renewcommand\given{\ngivensymbol[\delimsize]}
  #1
}
\DeclarePairedDelimiterX{\angles}[1]{\langle}{\rangle}{%
  \renewcommand\given{\ngivensymbol[\delimsize]}
  #1
}

\DeclarePairedDelimiter{\verts}{\lvert}{\rvert}

\newcommand{\p}{\parens}

\newcommand{\set}{\braces}
\newcommand{\abs}{\verts}

\DeclarePairedDelimiter{\floor}{\lfloor}{\rfloor}

\DeclarePairedDelimiterX{\inn}[2]{\langle}{\rangle}{#1,#2}

\DeclarePairedDelimiterX{\comm}[2]{\lbrack}{\rbrack}{#1,#2}
\DeclarePairedDelimiterX{\anti}[2]{\lbrace}{\rbrace}{#1,#2}

\DeclarePairedDelimiter{\ket}{\lvert}{\rangle}
\DeclarePairedDelimiterX{\qout}[2]{\lvert}{\rvert}{%
  #1\delimsize\rangle\delimsize\langle\mathopen{}#2
}
\DeclarePairedDelimiterX{\qproj}[1]{\lvert}{\rvert}{%
  #1\delimsize\rangle\delimsize\langle\mathopen{}#1
}

\DeclarePairedDelimiterX{\qinn}[2]{\langle}{\rangle}{%
  #1\givensymbol[\delimsize]#2
}
\DeclarePairedDelimiterX{\qamp}[3]{\langle}{\rangle}{%
  #1\givensymbol[\delimsize]#2\givensymbol[\delimsize]#3
}
\DeclarePairedDelimiterX{\qavg}[2]{\langle}{\rangle}{%
  #2\givensymbol[\delimsize]#1\givensymbol[\delimsize]#2
}

\newcommand{\SO}{\operatorname{SO}}

\newcommand{\lie}[1]{\mathfrak{#1}}

\renewcommand\P[1][]{%
  \ifstrempty{#1}{%
      \mathbf{P}
    }{%
      \mathbf{p}_{#1}
    }%
  }

\usepackage{ytableau}
\ytableausetup{smalltableaux,centertableaux,nobaseline}
\usetikzlibrary{tikzmark}

\newcommand{\ydiag}[1]{\vcenter{\hbox{\tiny \ydiagram{#1}}}}
\newcommand{\sydiag}[1]{\vcenter{\hbox{\scalebox{0.85}{\tiny \ydiagram{#1}}}}}
\newcommand{\ysub}[1]{\vcenter{\hbox{\scalebox{0.5}{\scriptsize \ydiagram{#1}}}}}
\newcommand{\sysub}[1]{\vcenter{\hbox{\scalebox{0.37}{\scriptsize \ydiagram{#1}}}}}
\NewDocumentCommand{\yket}{om}{%
  \IfValueTF{#1}{%
    \ket[#1]{\vcenter{\hbox{\small \begin{ytableau} #2 \end{ytableau}}}}
  }{%
    \ket*{\vcenter{\hbox{\small \begin{ytableau} #2 \end{ytableau}}}}
  }
}
\NewDocumentCommand{\tyket}{om}{%
  \IfValueTF{#1}{%
    \ket[#1]{\vcenter{\hbox{\tiny \begin{ytableau} #2 \end{ytableau}}}}
  }{%
    \ket*{\vcenter{\hbox{\tiny \begin{ytableau} #2 \end{ytableau}}}}
  }
}

\crefname{section}{Sec.}{Secs.}
\crefname{claim}{Claim}{Claims}

\begin{document}


\title{Qudit circuits with \texorpdfstring{$\bm{{\SU(d)}}$}{SU(d)} symmetry II}

\author{Austin Hulse}
\email{austin.hulse@duke.edu}
\affiliation{Duke Quantum Center, Duke University, Durham, NC 27708, USA}
\affiliation{Department of Physics, Duke University, Durham, NC 27708, USA}
\author{Hanqing Liu\,\orcidlink{0000-0003-3544-6048}}
\email{hanqing.liu@lanl.gov}
\affiliation{Theoretical Division, Los Alamos National Laboratory, Los Alamos, New Mexico 87545, USA}
\author{Iman Marvian}
\email{iman.marvian@duke.edu}
\affiliation{Duke Quantum Center, Duke University, Durham, NC 27708, USA}
\affiliation{Department of Physics, Duke University, Durham, NC 27708, USA}
\affiliation{Department of Electrical and Computer Engineering, Duke University, Durham, NC 27708, USA}

\title{A framework for semi-universality: Semi-universality of \texorpdfstring{$3$}{3}-qudit  \texorpdfstring{$\SU(d)$}{SU(d)}-invariant gates}

\begin{abstract}
  Quantum circuits with symmetry-respecting gates have attracted broad interest in quantum information science. 
  While recent work has developed a theory for circuits with Abelian symmetries, revealing important distinctions between Abelian and non-Abelian cases, a comprehensive framework for non-Abelian symmetries has been lacking. In this work, we develop novel techniques and a powerful framework that is particularly useful for understanding circuits with non-Abelian symmetries. 
  Using this framework we settle an open question on quantum circuits with $\SU(d)$ symmetry. We show that 3-qudit $\SU(d)$-invariant gates are semi-universal, i.e.,  generate all $\SU(d)$-invariant unitaries, up to certain constraints on the relative phases between sectors with inequivalent representation of symmetry. Furthermore, we prove that these gates achieve full universality when supplemented with 3 ancilla qudits.   
 Interestingly, we find that studying circuits with 3-qudit gates is also useful for a better understanding of circuits with 2-qudit gates.  
In particular, we establish that even though 2-qudit $\SU(d)$-invariant gates are not themselves semi-universal, they become universal with at most 11 ancilla qudits.
   Additionally, we investigate the statistical properties of circuits composed of random $\SU(d)$-invariant gates. Our findings reveal that while circuits with 2-qudit gates do not form a 2-design for the Haar measure over $\SU(d)$-invariant unitaries, circuits with 3-qudit gates generate a $t$-design, with $t$ that is quadratic in the number of qudits.
\end{abstract}

\maketitle

\section{Introduction}

The universality of 2-qudit quantum gates is a celebrated result in the fields of quantum computing and control theory \cite{divincenzo1995two, lloyd1995almost, deutsch1995universality,brylinski2002universal}. According to this result, any unitary transformation on any finite number of qudits can be realized with a finite sequence of two qudit gates. However, 
in the presence of a global symmetry, this universality fails: generic unitaries that respect a global continuous symmetry cannot be realized, even approximately, using $k$-qudit gates that respect the same symmetry, with any fixed $k$ \cite{Marvian2022Restrict,marvian2022quditcircuit,Marvian2024Rotationally,marvian2024theoryabelian}. In general, the locality of gates imposes various types of restrictions on the set of realizable unitaries. For instance, it restricts the possible relative phases between sectors with inequivalent irreducible representations (charges) of the symmetry. To distinguish these more common types of restrictions from other types, Ref.~\cite{marvian2024theoryabelian} proposes the notion of semi-universality, a weakening of the notion of universality. This concept, which is the main focus of the present paper, is defined below (see \cref{sec:semiundef} for the formal definition).

Recall that under the action of a symmetry group $G$, the total Hilbert space $\hilbert$ decomposes into subspaces (charge sectors), $\hilbert = \bigoplus_\lambda \hilbert_\lambda$, corresponding to inequivalent irreducible representations of $G$. 
A set of gates respecting this symmetry is called \emph{semi-universal}, if for any unitary $V$ that respects the symmetry, there exists a set of phases $\set{\theta_\lambda}$ such that $V \sum_\lambda \e^{\i \theta_\lambda} \Pi_\lambda$ can be realized as a sequence of gates in that set, where $\Pi_\lambda$ is the orthogonal projector to the charge sector $\hilbert_\lambda$.

One may hope that even though 2-qudit gates are not generally universal, at least they might be semi-universal. Previous works have shown that this is indeed the case for qubits with $\U(1)$ and $\SU(2)$ symmetry \cite{Marvian2024Rotationally,marvian2024theoryabelian}. In particular, in the case of SU(2) symmetry, 2-qubit gates realized by the Heisenberg exchange interaction, which is $\SU(2)$-invariant, are semi-universal. On the other hand, surprisingly, it turns out that for $d\ge 3$, 2-qudit $\SU(d)$-invariant gates on qudit systems are not semi-universal. Indeed, Ref. \cite{marvian2022quditcircuit} identifies new conservation laws that restrict the time evolution of the system,\footnote{In particular, in \cite{marvian2022quditcircuit}, it is shown that in certain subspaces, the dynamics of qudits under 2-qudit $\SU(d)$-invariant unitaries can be mapped to the dynamics of a fermionic system evolving under a free (non-interacting) Hamiltonian.} even when the state is restricted to one $\SU(d)$ charge sector. Since 2-qudit $\SU(d)$-invariant gates are not semi-universal in general, it is natural to ask whether $k$-local gates are; and, in particular, what is the minimum locality that achieves semi-universality. 

\begin{table*}
  \begin{tblr}{c|c|c|c}
    \toprule
    Ancillae & Local dimension $d$ & Semi-universality & Universality \\ \midrule
    \SetCell[r=2]{c} \ding{55} & 2 & $k = 2$ & $k = 2 \floor{\frac{n}{2}}$ \\ \midrule
    & $d \geq 3$ & $k = 3$ & $k = n$ \\ \midrule
    \SetCell[r=2]{c} \ding{51} & 2 & $k = 2$ {with no ancillae} & $k = 2$ and $2$ ancillae \\ \midrule
    & $d \geq 3$ & $k = 2$ and $\leq 8$ ancillae & $k = 2$ and $\leq 11$ ancillae, or $k = 3$ and $\leq 3$ ancillae \\ \bottomrule
  \end{tblr}
  \caption{{\textbf{Locality of gates needed to achieve (semi)-universality.}} This table lists the minimum $k$, such that $k$-qudit $\SU(d)$-invariant unitaries achieve semi-universality and universality on $n\ge 3$ qudits, with or without ancillae. The results for $ \SU(2)$ symmetry were shown in \cite{Marvian2024Rotationally}, and the results on $\SU(d)$ symmetry with $d\ge 3$ are established in this work. Note that for $d = 2$, 2 ancilla qudits are needed to achieve universality, even if one is allowed to use $k$-qubit gates, with $k < 2 \floor{\frac{n}{2}}$. }
  \label{tab:summary}
\end{table*}

\subsection{Summary of results}

In this paper, we develop new powerful tools and a framework for understanding semi-universality in quantum circuits with arbitrary symmetries. While recent work \cite{marvian2024theoryabelian} has found a simple characterization of circuits with Abelian symmetries, 
it is known that circuits with non-Abelian symmetries can show significantly more complicated behaviors. For instance, in the presence of non-Abelian symmetries, the unitaries realized in one charge sector $\mathcal{H}_\lambda$, may dictate the unitaries in other (possibly multiple) sectors, whereas this cannot happen in the case of Abelian symmetries \cite{marvian2024theoryabelian, marvian2022quditcircuit}. 

Applying these tools to the important example of $\SU(d)$ symmetry, we settle an open question about circuits with this symmetry.
  It was recently shown that 2-qudit $\SU(d)$-invariant gates are not semi-universal when $d > 2$, while they are for $d = 2$ \cite{marvian2022quditcircuit, Marvian2024Rotationally}. Furthermore, using properties of the Young-Jucys-Murphy elements and Okounkov-Vershik's approach to the representation theory of the symmetric group \cite{okounkov2005}, Ref. \cite{Zheng_2023} argues that 4-local $\SU(d)$-invariant unitaries are semi-universal. However, prior to the present work, it was not known if semi-universality can be achieved with 3-qudit unitaries or not. 

Here, we settle this open question and prove that 3-qudit unitaries that respect a global $\SU(d)$ symmetry on $d$-dimensional qudits are indeed semi-universal. It is worth emphasizing that our proof of the semi-universality of 3-qudit gates is elementary and based on the tools developed in this paper, which are applicable to other symmetries (see \cref{sec:semiuni}, and in particular, \cref{MainLemma,thm:mulblocks}). 

In our construction, the generating gate set includes all 2-qudit $\SU(d)$-invariant gates
\begin{align}\label{eq:gateset}
  \exp({\i \theta \P_{ij}})\ : \quad \theta \in [0, 2 \pi)\ , 
\end{align}
where $\P_{ij}$ is the SWAP operator on qudits\footnote{We label the qudits as if they were in a chain for convenience: this geometry plays no role in the proof of semi-universality.} $i$ and $j$, and a single 3-qudit gate, e.g., one of the unitaries 
\begin{align}\label{reflection}
  R_+=\exp({\i\pi\Pi_{\ysub{3}}})\ ,\ \ \text{ or }\ \ \ 
  R_-=\exp({\i\pi\Pi_{\ysub{1,1,1}}})\ .
\end{align}
Here, $\Pi_{\ysub{3}}$ and $\Pi_{\ysub{1,1,1}}$ are, respectively, the Hermitian projectors to the symmetric and anti-symmetric subspaces of the three-qudit Hilbert space $(\mathbb{C}^d)^{\otimes 3}$, which means  $R_+=\mathbb{I}-2\Pi_{\ysub{1,1,1}}$
and $R_-=\mathbb{I}-2\Pi_{\ysub{1,1,1}}$ are reflection unitaries, and are also permutationally invariant gates (see \cref{fig:circuit}). Here we emphasize that $R_\pm$ are not special: almost any single 3-qudit $\SU(d)$-invariant unitary together with 2-qudit unitaries in \cref{eq:gateset} are semi-universal.   \Cref{prop:2local-criteria} provides a simple characterization of 3-qudit gates that are capable of achieving semi-universality. In particular, we find that any 3-qudit gate that cannot be realized with 2-qudit gates is sufficient to make them universal.

It is worth noting that when restricted to a 3-qudit system, 2-qudit $\SU(d)$-invariant gates already achieve semi-universality, and the gate $R_+$ is just a relative phase between the $\SU(d)$ charge sectors. However, interestingly, when acting on 3 qudits in a system with $n>3$ qudits, such gates can drastically change the set of realizable unitaries and make them semi-universal.

\begin{figure}[t]
  \centering
  \includegraphics[width=0.35\textwidth]{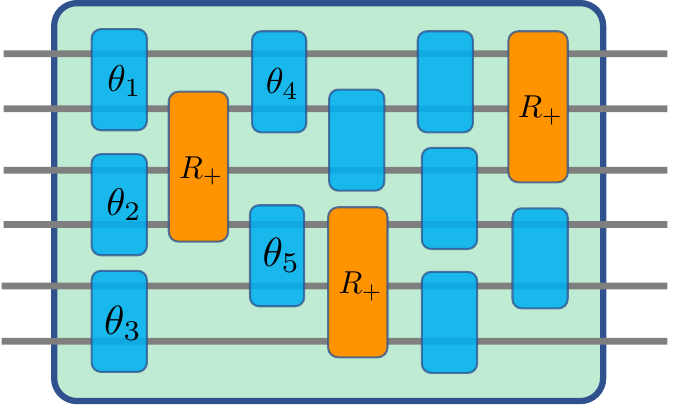}
  \caption{\textbf{Semi-universality with 3-qudit gates.} 
    Any set of $\SU(d)$-invariant gates is called \emph{semi-universal}, if they generate all $\SU(d)$-invariant unitaries, up to possible constraints on the relative phases between sectors with inequivalent irreps of $\SU(d)$. While 2-qudit $\SU(d)$-invariant gates are not semi-universal for $d \geq 3$, we show that amending them with any single generic 3-qudit unitary makes them semi-universal. In this schematic circuit, 2-qudit gates are in the form of \cref{eq:gateset} for arbitrary $\theta$, and $R_+$ is the 3-qudit reflection unitary defined in \cref{reflection}. }
  \label{fig:circuit}
\end{figure}

\Cref{tab:summary} summarizes the results on (semi-)universality from the current paper and \cite{Marvian2024Rotationally, marvian2022quditcircuit}. 
As listed in this table, in addition to the semi-universality of 3-qudit gates, which is shown in \cref{sec:4ex,sec:over}, in this work, we also prove that 
\begin{enumerate}
\item Without ancilla qudits, for $d>2$,  $(n - 1)$-qudit $\SU(d)$-invariant gates are still not sufficient to achieve universality on $n$ qudits (see \cref{sec:imposs}).
\item With at most 8 ancilla qudits, semi-universality can be achieved with 2-qudit $\SU(d)$-invariant gates when $d>2$ (see \cref{lem:ancillauni}). 
\item With at most 11 ancilla qudits, universality can be achieved with 2-qudit $\SU(d)$-invariant gates when $d>2$ (see \cref{cor:anc}).
\end{enumerate}

Therefore, a remarkable (and perhaps, unexpected) corollary of our study of 3-qudit $\SU(d)$-invariant gates is a significantly simpler understanding of the computational universality of 2-qudit gates using ancilla qudits!   This has been recently established in \cite{vanmeter2021universality} by applying an advanced result in the mathematical Lie-algebraic literature by Marin \cite{marin2007algebre}.  (This work characterizes the Lie algebra generated by transpositions as a subalgebra of the group algebra of permutations.)

The semi-universality of 3-qudit gates allows us to characterize the group 
generated by $k$-qudit $\SU(d)$-invariant unitaries on $n$ qudits, denoted by $\mathcal{V}_k^{(n)}$ for $k\ge 3$ (Note that we have suppressed the $d$ dependence to simplify the notation). Recall that according to the general results of \cite{Marvian2022Restrict}, $\mathcal{V}_k^{(n)}$ is a compact connected Lie group. Then, as we discuss in \cref{sec:imposs}, in the regime $n\ge k\gg d\ge 2$, the difference between the dimensions of this group and the subgroup generated by 3-qudit $\SU(d)$-invariant gates is approximately
\begin{align}\label{eq:asym-center}
  \dim\mathcal{V}_k^{(n)}- \dim\mathcal{V}^{(n)}_3\approx \frac{k^{d-1}}{d!(d-1)!} + \mathcal{O}(k^{d-2})\ .
\end{align}
In \cref{fig:asym-rpn} we plot 
\begin{align}\label{eq:define-ratio}
  \rho_{k,d}^{(n)} := \frac{\dim\mathcal{V}_k^{(n)}-\dim \mathcal{V}_3^{(n)}}{\dim\mathcal{V}^{(n)}-\dim \mathcal{V}_3^{(n)}} \approx \Big(\frac{k}{n}\Big)^{d-1} \ , 
\end{align}
for various $d$ together with their asymptotic behavior, where $\mathcal{V}^{(n)}=\mathcal{V}_n^{(n)}$ denotes the group of all $\SU(d)$-invariant unitaries on $n$ qudits. This ratio determines how the dimension of the Lie group of realizable unitaries grows with $k$. In particular, $\rho_{k,d}^{(n)}=1$ means that universality is achieved. Interestingly, when $d\ge 3$ this happens only if $k=n$, i.e., gates act on all qudits in the system.\footnote{On the other hand, in the case of $d=2$, i.e., qubits with SU(2) symmetry, when $n$ is odd, the universality can be achieved with gates acting on $n-1$ qubits.} 
Note that the right-hand side of \cref{eq:asym-center} is independent of $n$ the number of qudits (indeed, according to the general results of \cite{Marvian2022Restrict}, this is a consequence of the fact that the symmetry group $G=\SU(d)$ is connected). 

\begin{figure}[t]
  \centering
  \includegraphics[width=0.48\textwidth]{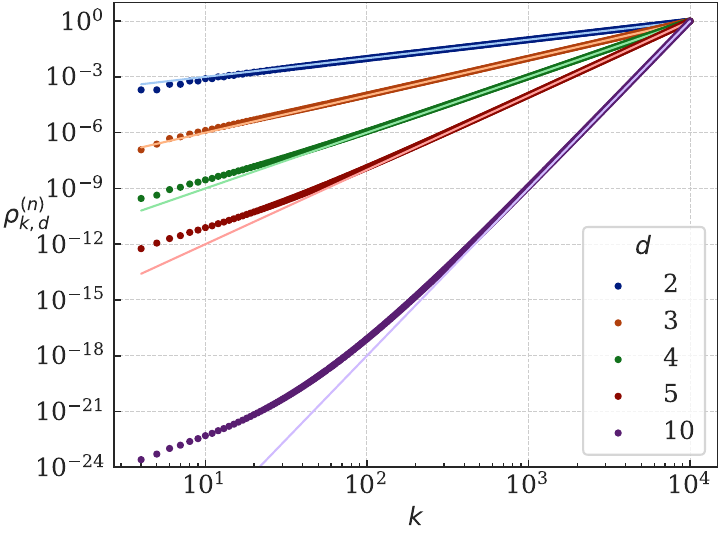}
  \caption{\textbf{Convergence to universality as a function of the locality of gates.} According to the general no-go theorem of \cite{Marvian2022Restrict}, in symmetric quantum circuits with continuous symmetries, without ancilla qudits, universality cannot be achieved with $k$-qudit gates with a fixed $k$. On the other hand, we show that in the case of $\SU(d)$ symmetry, semi-universality is achieved with 3-qudit gates, which means for $k\ge 3$, the only constraints on the realizable unitaries are on the relative phases between sectors with inequivalent irreducible representations of $\SU(d)$. 
    Roughly speaking, the ratio $\rho_{k,d}^{(n)}$ defined in \cref{eq:define-ratio}, describes the fraction of these constraints that vanish with $k$-qudit gates with $k\ge 3$. Universality is achieved when $\rho_{k,d}^{(n)}=1$. Here, we plot this ratio for a system with $n=10^4$ qudits for different values of $d$, as specified in the plot. The dots are $\rho_{k,d}^{(n)}$ and the lines are its asymptotic expression $(k/n)^{d-1}$. The denominators of the ratio in \cref{eq:define-ratio} are $4998$, $8.3\times 10^6$, $7.0\times 10^9$, $3.5\times 10^{12}$ and $7.9\times 10^{23}$ for $d=2,3,4,5,10$, respectively.}
  \label{fig:asym-rpn}
\end{figure}

In \cref{Sec:design} we also discuss the implications of this result on 
the statistical properties of random quantum circuits with $\SU(d)$-invariant gates. As was noted in \cite{marvian2022quditcircuit}, the additional conservation laws, that restrict unitaries realized by 2-qudit circuits, imply that the distribution of unitaries generated by such random circuits is not a 2-design for the Haar distribution over the group of $\SU(d)$-invariant unitaries $\mathcal{V}^{(n)}$ (see \cref{Sec:design} for the definition of $t$-designs). As we show in \cref{Sec:design}, a corollary of the semi-universality of 3-qudit gates is that, assuming the number of qudits is $n> d$, the distribution of unitaries generated by random circuits formed from such gates is a $t$-design up to $t\approx n^2/2$. 

Finally, we note that while the main emphasis of this paper is on semi-universality in the context of symmetric quantum circuits, many of the ideas and techniques developed here are broadly useful in the context of quantum computing and control theory. In particular, \cref{MainLemma}, which provides the necessary and sufficient conditions for semi-universality, and \cref{thm:mulblocks}, which is used for extending controllability from a subspace to the full space, are of independent interest.

\subsection{Outline}

In \cref{sec:semiuni} we formally define semi-universality for an arbitrary unitary symmetry group, and present a number of generally applicable tools. In particular, we provide a necessary and sufficient condition for semi-universality to hold in \cref{MainLemma}, and we describe how using ancillae, semi-universality can be promoted to universality. 

In \cref{sec:semiSU} we specify to the case of $\SU(d)$ symmetry on $d$-dimensional qudits and state \cref{thm:semi-universality}, the semi-universality of 3-qudit $\SU(d)$-invariant unitaries. We also describe one of the main tools used in this paper, namely Schur-Weyl duality. 

To explain the applications of the tools developed in \cref{sec:semiuni}, in \cref{sec:4ex} we present a detailed discussion of the examples of $n = 3$ and $n = 4$ qudits. In particular, we describe in detail how semi-universality fails for 2-qudit gates and prove that it holds for 3-qudit gates using \cref{MainLemma}. This serves as the base case of the induction argument for \cref{thm:semi-universality}, which is proven in \cref{sec:over}.

In \cref{sec:imposs}, we show that, when $d \geq 3$ and without using ancilla, universality on $n$ qudits cannot be achieved without full nonlocal control, i.e. $n$-local gates are required. We consider the use of ancilla qudits for achieving (semi-)universality in \cref{sec:ancilla}. The statistical properties of circuits generated from $\SU(d)$-invariant $3$-local gates are studied in \cref{Sec:design}.

Finally, in \cref{sec:tools}, we prove \cref{MainLemma}. We also describe more general scenarios in which semi-universality does not hold, even when there is subsystem universality on all charge sectors. In particular, in \cref{Mainlemma_gen} we present a characterization of $G$-invariant groups that are subsystem universal.

\begin{table}[ht]
  \centering
  \caption{\textbf{Table of notations.}}
  \begin{tblr}{
      colspec = {Q[c,m] Q[l,m]}, 
      rowsep = 1pt, column{2}={6.3cm}
    }
    \toprule
    Notation & Definition \\
    \midrule
    $\V^G$ & $G$-invariant unitaries \\
    $\mathcal{SV}^G$ & commutator subgroup of $\mathcal{V}^G$ ($G$-invariant unitaries w/o relative phases) \\
    $\V_k^{(n)}$ & group generated by $k$-qudit $\SU(d)$-invariant unitaries on $n$ qudits \\
    $\V^{(n)}=\V_n^{(n)}$ & all $\SU(d)$-invariant unitaries on $n$ qudits \\
    $\mathcal{SV}^{(n)}$ & 
    commutator subgroup of $\mathcal{V}^{(n)}$ ($\SU(d)$- invariant unitaries {w/o} relative phases) \\
    $\mathbb{S}_n$ & symmetric group on $n$ objects \\
    $\mathcal{M}_\lambda$ & multiplicity subsystem (irrep space of $\mathbb{S}_n$ in the case of $\SU(d)$ symmetry) \\
    $\Lambda_{n,d}$ & irreps of $\mathbb{S}_n$ on $n$ qudits  \\
    \bottomrule
  \end{tblr}
  \label{tab:notations}
\end{table}

\section{Semi-universality}\label{sec:semiuni}

In this section, we define semi-universality rigorously for arbitrary symmetry groups, and introduce three powerful lemmas that are later used in \cref{sec:over} to prove the semi-universality of 3-qudit gates. We anticipate that these lemmas will find other applications beyond this result.  

\subsection{Definition}\label{sec:semiundef}

Let $\mathcal{V}^G$ be the set of all $G$-invariant unitaries; that is, $V \in \mathcal{V}^G$ if and only if $[V, U(g)] = 0$ for all $g \in G$, where $U(g)$ is the unitary representation of a finite or compact Lie group $G$. Then, $\mathcal{V}^G$ itself is a compact, connected Lie group \cite{Marvian2022Restrict}. 

In the introduction we (informally) defined the notion of semi-universality. As is explained below, an equivalent and useful definition of semi-universality is the following: a set of $G$-invariant unitaries is called \emph{semi-universal} for $\mathcal{V}^G$, if the subgroup $\mathcal{T}$ generated by them contains the commutator subgroup of $\mathcal{V}^G$, denoted by $\mathcal{SV}^G$. That is,
\be\label{def0}
 \mathcal{T} \supseteq  \mathcal{SV}^G:=[\mathcal{V}^G, \mathcal{V}^G] \ ,
\ee
where for any group $\mathcal{S}$ the commutator subgroup $[\mathcal{S}, \mathcal{S}]$ is generated by the group commutators $W V W^{-1} V^{-1}: W, V \in \mathcal{S}$. As it is more apparent from \cref{rq} below, from a Lie-algebraic perspective, this condition implies that the semi-simple parts of the Lie algebras associated with $\mathcal{V}^G$ and $\mathcal{T}$ are identical, whereas their center can be different (this motivates the name ``semi-universality'').

A useful characterization of (semi-)universality can be obtained by considering the isotypic decomposition of representation $\{U(g): g\in G\}$, namely 
\begin{equation}\label{eq:charge}
  \hilbert = \bigoplus_{\lambda \in \Lambda} \hilbert_\lambda = \bigoplus_{\lambda \in \Lambda} \hilbert[Q]_\lambda \otimes \hilbert[M]_\lambda\ ,
\end{equation}
where the sum is over a set $\Lambda$ of inequivalent irreducible representations (irreps) of $G$, $\hilbert[Q]_\lambda$ is a space carrying the irrep $\lambda$, and $\hilbert[M]_\lambda$ corresponds to the multiplicity of $\lambda$, i.e., $\dim \hilbert[M]_\lambda > 0$ is the multiplicity of $\lambda$ in $\hilbert$. Applying Schur's lemma, we find that, with respect to the decomposition in \cref{eq:charge}, any symmetric unitary $V \in \mathcal{V}^G$ is block-diagonal and takes the form
\begin{equation}
  V = \bigoplus_{\lambda \in \Lambda} V_\lambda = \bigoplus_{\lambda \in \Lambda} (\ident_{\hilbert[Q]_\lambda} \otimes v_\lambda)\ ,
\end{equation}
where $\ident_{\hilbert[Q]_\lambda}$ is the identity operator on $\hilbert[Q]_\lambda$ and $v_\lambda \in \U(\hilbert[M]_\lambda)$ is a unitary acting only on the multiplicity degrees of freedom. Here, $\U(\hilbert[M]_\lambda)$ is the group of unitaries on the Hilbert space $\hilbert[M]_\lambda$, and $\SU(\hilbert[M]_\lambda)$ is its subgroup with determinant one.

It is useful to consider the homomorphisms $\pi_\lambda: \lambda\in \Lambda$ from $\mathcal{V}^G$ to $\U(\mathcal{M}_\lambda)$, defined by $\pi_\lambda(V)=v_\lambda$. Then, we denote the collection of these unitaries $v_\lambda: \lambda\in\Lambda$, as $\pi_{\Lambda} (V) = (v_\lambda)_{\lambda \in \Lambda}$. This defines an isomorphism 
\begin{equation}\label{eq:prodproj}
  \pi_\Lambda : \mathcal{V}^G \to \prod_{\lambda \in \Lambda} \U(\hilbert[M]_\lambda)\ ,
\end{equation}
from the group of all $G$-invariant unitaries to the Cartesian product of unitary groups on the multiplicity spaces. This means the unitary $V$ is uniquely determined by the set $\pi_{\Lambda} (V)$, 
via $V= {\bigoplus}_{\lambda \in \Lambda} \ident_{\hilbert[Q]_\lambda} \otimes \pi_\lambda(V)$.

Recall that $\SU(\hilbert[M]_\lambda)$ is a perfect group, i.e., it is equal to its commutator subgroup, which means $\SU(\hilbert[M]_\lambda) = [\U(\hilbert[M]_\lambda), \U(\hilbert[M]_\lambda)]$. This, in turn, implies 
\begin{equation}
  \mathcal{SV}^G \cong \pi_\Lambda([\mathcal{V}^G, \mathcal{V}^G]) = \prod_{\lambda \in \Lambda} \SU(\hilbert[M]_\lambda)\ .
\end{equation}
In particular, this means for any $G$-invariant unitary $V\in \mathcal{V}^G$, there exists a set of phases $\theta_\lambda\in[0,2\pi)$ and $\widetilde{V}\in \mathcal{SV}^G$, such that $V=\widetilde{V}\sum_\lambda \e^{\i \theta_\lambda} \Pi_\lambda$, demonstrating the equivalence of the above definition with the definition presented in 
the introduction.

Therefore, the definition in \cref{def0} means that group $\mathcal{T}$ is semi-universal if, and only if,  its commutator subgroup is equal to $\mathcal{SV}^G$, i.e.,
\be\label{rq}
[\mathcal{T}, \mathcal{T}] = [\mathcal{V}^G, \mathcal{V}^G] =  \mathcal{SV}^G\ .
\ee
If the group $\mathcal{T} \subseteq \mathcal{V}^G$ is connected and each $\pi_\lambda(\mathcal{T}) = \set{\pi_\lambda(V) \given V \in \mathcal{T}}$ acts irreducibly on $\hilbert[M]_\lambda$, then it follows that the connected component of the identity of $\mathcal{T} \cap \mathcal{SV}^G$ is equal to $[\mathcal{T}, \mathcal{T}]$. See \cref{sec:commconn} for a proof of this statement.

Finally, it is worth noting that a subgroup of $G$-invariant unitaries $\mathcal{T}\subseteq \mathcal{V}^G$ contains $\mathcal{SV}^G$, if and only if it contains the one-parameter family of unitaries $\exp(\i H t): t\in \mathbb{R}$ for all centerless $G$-invariant Hamiltonians, where we say $G$-invariant Hamiltonian $H$ is \emph{centerless} if $\Tr(\Pi_\lambda H)=0$ for all $\lambda\in\Lambda$, or equivalently, if $\Tr(U(g) H)=0$ for $g\in G$.

\subsection{Lemma 1: A simple characterization of semi-universality}\label{sec:nec-suf}

The following lemma is one of our main new tools for studying semi-universality and can potentially have broad applications beyond the context of quantum circuits.

\begin{lemma}\label{MainLemma}
  A subgroup $\mathcal{T}\subseteq \mathcal{V}^G$ of $G$-invariant unitaries contains the commutator subgroup of all $G$-invariant unitaries $\mathcal{SV}^G=[\mathcal{V}^G,\mathcal{V}^G]$ if, and only if, the following two conditions hold:

  \noindent\textbf{\textup{A ({Subsystem universality} in all sectors):}} For any irrep $\lambda\in\Lambda$ the action of $\mathcal{T}$ on the corresponding multiplicity subsystem $\mathcal{M}_{\lambda}$ contains $\SU(\mathcal{M}_{\lambda})$, i.e.,
  $\SU(\mathcal{M}_{\lambda})\subseteq \pi_{\lambda}(\mathcal{T})=\{\pi_{\lambda}(V): V\in \mathcal{T}\}$.

  \noindent\textbf{\textup{B (Pairwise independence):}} For any pair of distinct irreps $\lambda_1, \lambda_2 \in\Lambda$, if $\dim(\mathcal{M}_{\lambda_1}) = \dim(\mathcal{M}_{\lambda_2}) \ge 2$, then there exists a unitary $V\in\mathcal{T}$ such 
  that
  \be\label{cond4}
  |\Tr(\pi_{\lambda_1}(V))|\neq |\Tr(\pi_{\lambda_2}(V))|\ .
  \ee
\end{lemma}
As we show in \cref{sec:mainproof}, this lemma can be established using Goursat's and Serre's lemmas (see \cref{lem:SUgoursat,lem:serre}, respectively).  Additionally, in \cref{Mainlemma_gen} we present a variant of \cref{MainLemma}, which does not assume condition $\textbf{B}$ holds. In particular, we find the most general form of the subgroups of $G$-invariant unitaries that respect condition $\textbf{A}$, subsystem universality in all sectors. 

In words, condition \textbf{A} means that inside the subspace associated with any irrep $\lambda$, all $G$-invariant unitaries are realizable up to a phase. When irrep $\lambda$ of group $G$ is itself 1D, this means all unitaries inside the subspace $\mathcal{H}_\lambda$ are realizable, a condition that is sometimes called ``subspace controllability" in control theory. However, in general, $\lambda$ is not a 1D irrep, and therefore following the standard terminology in quantum information, we refer to this condition as ``subsystem universality". Condition \textbf{B}, on the other hand, guarantees that for any pair of irreps, the realized unitaries in the corresponding subspaces are independent of each other. 

When condition \textbf{A} holds, condition  \textbf{B} is equivalent to the following, which can therefore replace  it:\\ 

\noindent \textbf{B}': For any pair of distinct irreps $\lambda_1,\lambda_2$, 
if $\dim(\mathcal{M}_{\lambda_1}) = \dim(\mathcal{M}_{\lambda_2}) \ge 2$, then 
there exists a unitary $V\in\mathcal{T}$ 
that acts as the identity operator on one of $\mathcal{M}_{\lambda_1}$ or $\mathcal{M}_{\lambda_2}$, and is not proportional to the identity on the other, {e.g.,}
\begin{align}\label{cond3}
  \pi_{\lambda_1}(V)= \mathbb{I}_{\hilbert[M]_{\lambda_1}} \ , \ \ \ \pi_{\lambda_2}(V)\neq \e^{\i \theta} \mathbb{I}_{\hilbert[M]_{\lambda_2}}\ ,
\end{align}
for any $\e^{\i \theta}$, where $\mathbb{I}_{\hilbert[M]_{\lambda_1}}$ and $\mathbb{I}_{\hilbert[M]_{\lambda_2}}$ are the identity operators on $\mathcal{M}_{\lambda_1}$ and $\mathcal{M}_{\lambda_2}$, respectively.\\

Note that if a unitary $V$ satisfies the condition in \cref{cond3}, then it also satisfies the condition in \cref{cond4}. This can be seen by noting that $|\Tr(\pi_\lambda(V))|=\Tr(\pi_\lambda(\mathbb{I}_{\hilbert[M]_{\lambda}}))$ if and only if all eigenvalues of $\pi_\lambda(V)$ have the same phase, which means in the complex plane they are aligned in the same direction.

If both conditions \textbf{A} and \textbf{B} (or, equivalently, conditions \textbf{A} and \textbf{B}') are satisfied for a pair of irreps $\lambda_1$ and $\lambda_2$, then the joint projection of $\mathcal{T}$ to $\mathcal{M}_{\lambda_1}$ and $\mathcal{M}_{\lambda_2}$ contains $\SU(\mathcal{M}_{\lambda_1})\times \SU(\mathcal{M}_{\lambda_2})$. That is,
\be\label{both}
\pi_{\lambda_1, \lambda_2}(\mathcal{T})=\{(v_{\lambda_1},v_{\lambda_2}): V\in \mathcal{T}\} \supseteq \SU(\mathcal{M}_{\lambda_1})\times \SU(\mathcal{M}_{\lambda_2}) .
\ee
This, in particular, means that the unitary realized in one sector is not dictated by the other, up to possible constraints on the global phases. In \cref{sec:mainproof} we argue that together with Serre's \cref{lem:serre}, \cref{both} implies independence in all sectors, as claimed in \cref{MainLemma}.

\subsection*{Pairwise independence implies full independence}

While the necessity of both conditions \textbf{A} and \textbf{B} is trivial, their sufficiency is far from obvious. Indeed, it is remarkable that according to this lemma, to demonstrate semi-universality one needs to check the independence of the realized unitaries only among pairs of sectors and not, e.g., among 3-tuples of sectors. To see an example of such dependencies, consider the subgroup of U$(1)^{3}$ corresponding to a 3-tuple of phases 
\be\label{ex:sub}
(\e^{\i \theta_1}, \e^{\i \theta_2}, \e^{\i( \theta_1+\theta_2)})\ \ \ : \theta_1, \theta_2\in [0,2\pi) \ .
\ee
Then, any pair of these 3 phases are fully independent of each other. That is, for any 2-tuples the projection of this group is $\U(1)^2$, and yet the overall group is isomorphic to $\U(1)^2$, rather than U$(1)^{3}$. 
However, \cref{MainLemma} implies that this situation cannot happen in the context of semi-universality. As we further explain in \cref{sec:tools}, this is a consequence of the fact that the group $\SU(m)$ is perfect, while $\U(1)$ in the above example is not.

\subsection*{Failure of semi-universality}

As mentioned before, when semi-universality holds, the only constraints on realizable unitaries are constraints on the relative phases between sectors with different charges (irreps) of symmetry. On the other hand, failure of semi-universality can be due to various kinds of constraints. It is useful to recall the numbering system of \cite{marvian2024theoryabelian} for all possible constraints on the universality of $\mathcal{T}$. Namely, 
\begin{enumerate}[align=left]
\item[Type \textbf{I}:] {Constraints on the relative phases between charge sectors. Note that this is a failure of universality rather than semi-universality.} 
\item[Types \textbf{II} and \textbf{III}:] Charge sectors which are not subsystem universal, either because (\textbf{II}) the action of $\mathcal{T}$ on this sector is not irreducible, or (\textbf{III}) the action is irreducible, but only a proper subset of determinant-one unitaries can be achieved.
\item[Type \textbf{IV}:] ``Correlations'' between the determinant-one unitaries in distinct charge sectors. That is, the unitary realized in one sector determines the unitary realized in the other sector (up to a possible global phase). Such correlations can arise independent of type \textbf{II} and \textbf{III} constraints.
\end{enumerate}

\cref{Mainlemma_gen} characterizes type \textbf{IV} constraints assuming type  \textbf{II} and \textbf{III} do not exist. It is worth noting that constraints of type $\textbf{I}$ on a smaller system can result in constraints of other types in a system with more qudits. Indeed, this is exactly what happens in the case of 2-qudit $\SU(d)$-invariant gates. As we show in \cref{sec:2on3}, when restricted to $n=3$ qudits, the 
2-qudit $\SU(d)$-invariant gates are semi-universal, and therefore the only constraints on the realizable unitaries are type $\textbf{I}$ constraints. However, as soon as we go to $n=4$ and more qudits, these constraints on the relative phases will also cause type $\textbf{III}$ and $\textbf{IV}$ constraints. See \cite{marvian2022quditcircuit} for further discussions on restrictions on circuits with 2-qudit $\SU(d)$-invariant gates. 

It is also worth noting that according to the result of \cite{marvian2024theoryabelian}, when the symmetry group $G$ is Abelian and has on-site representation (see \cref{sec:schur-weyl}), semi-universality holds if and only if there are no type \textbf{II} restrictions. That is, types \textbf{III} and \textbf{IV} restrictions do not exist for Abelian symmetries.

\subsection{Lemma 2: Extending controllability from a 3D subspace to the full space}

In \cref{sec:over} we show how \cref{MainLemma} can be applied recursively via an induction argument to establish the semi-universality of 3-qudit $\SU(d)$-invariant gates. To apply such recursion, we need one more tool, which is discussed in this subsection. 

In \cref{sec:genSUprf} we present a series of results, which under different assumptions allow us to extend controllability from a subspace to the full space (see also \cite{marvian2022quditcircuit} for previous examples of such results). Using these results, in \cref{sec:lem2} we prove the following lemma, which is one of the tools needed to establish the semi-universality of $\SU(d)$-invariant 3-qudit gates. 

\begin{lemma}\label{thm:mulblocks}
  Let $\hilbert$ be a finite-dimensional Hilbert space with a subspace $\hilbert_1 \subset \hilbert$ with dimension $\dim(\hilbert_1) \geq 3$. Let $A_i$, $i = 1, \dots, k$, be traceless anti-Hermitian operators on $\hilbert$ and consider the one-parameter groups $\mathcal{A}_i = \set{\e^{t A_i} \given t \in \real}$. If the group
  \begin{equation}
    \mathcal{W} = \angles{\mathcal{A}_i, \SU(\hilbert_1) \, \given \, i = 1, \dots, k}
  \end{equation}
  acts irreducibly on $\hilbert$, then $\mathcal{W} = \SU(\hilbert)$.
\end{lemma}

\begin{remark}\label{rem1}
  Indeed, a similar result holds when $\dim(\hilbert_1) =2$ and $\dim(\hilbert)$ is odd. We present the proof of this for the special case of $\dim(\hilbert)=3$ in \cref{app:twoblocks}, and postpone the proof of the general case to \cite{liu2024control}. 
\end{remark}
Furthermore, in 
\cite{liu2024control}, we also show that a variant of this theorem applies to the case when $\dim(\hilbert_1)=2$ and $\dim \hilbert$ is even. However, in that case, the realized group $\mathcal{W}$ can be either the full $\SU(\hilbert)$ or a subgroup isomorphic to the symplectic group $\Sp(\hilbert)$ (see \cite{liu2024control} for further discussions).

\subsection{Lemma 3: From semi-universality to universality with 3 ancillae}

So far, we have not made any assumptions about the structure of the Hilbert space $\mathcal{H}$ and the unitary representation of symmetry $G$ on this space. For many applications in physics and quantum computing, we are interested in the scenarios where $\mathcal{H}$ is the Hilbert space of $n$ identical qudits, i.e., $\mathcal{H}=(\mathbb{C}^d)^{\otimes n}$. We are also often interested in the case where $U(g)$ is an ``on-site" representation of the symmetry, such that $U(g)=u(g)^{\otimes n}$, where $u(g): g\in G$ acts on a single qudit $\mathbb{C}^d$ (this is often called ``global rotations'' on the system). However, it is worth noting that the above lemmas apply to other cases, e.g., when $G$ is the permutation group $\mathbb{S}_n$.

Then, under the assumption that the representation of symmetry is ``on-site'', it turns out that one can use ancilla qudits to elevate semi-universality to universality, as defined in the following. 
Consider a subset of $G$-invariant realizable unitaries $\mathcal{T}$ on $n+c$ qudits and a fixed  state $\ket{\eta} \in (\complex^d)^{\otimes c}$ of $c$ ancilla qudits. Then, we say unitary $V \in \U((\complex^d)^{\otimes n})$ on $n$ qudits is realizable with $c$ ancilla qudits, if there exists a unitary $\widetilde{V}\in \mathcal{T}$ such that \begin{equation}\label{eq:cat}
  \widetilde{V}(\ket{\psi} \otimes \ket{\eta}) = (V \ket{\psi}) \otimes \ket{\eta} 
\end{equation}
for all $\ket{\psi} \in (\complex^d)^{\otimes n}$.

Recent work in \cite{marvian2024theoryabelian} shows that in the case of Abelian groups, such as $\U(1)$, universality can be achieved with a single ancilla qudit. Furthermore, in \cite{Marvian2024Rotationally} we showed that in the case of qubits with $\SU(2)$ symmetry, universality can be achieved using 2 ancilla qubits, whereas 1 ancilla qubit is not sufficient in that case. Based on these previous observations, one may expect that the number of required ancillae may grow with the size of the group (e.g., with $d$ in the case of $\SU(d)$ symmetry). However, as we show in \cref{sec:ancilla}, this is not the case: 3 ancilla qudits are sufficient to achieve universality, provided that semi-universality holds.
\begin{lemma}\label{prop:ufromsemi}
  Let $G$ be an arbitrary symmetry group with on-site representation on a system with $n+3$ qudits (i.e., it acts via $g \mapsto u(g)^{\otimes (n+3)}$). Suppose a group $\mathcal{T}$ is semi-universal, i.e., it is a subgroup of $G$-invariant unitaries that contains $\mathcal{SV}^G$. Then, using 3 ancilla qudits, we can realize any $G$-invariant unitary on $n$ qudits using unitaries in $\mathcal{T}$, as defined in \cref{eq:cat}. 
\end{lemma}

For instance, as we further explain in \cref{sec:uni3}, the state $|\eta\rangle$ of ancilla can be chosen to be either of
\begin{align}
  \ket{\eta_1} &= \frac{1}{2 \sqrt{3}}(2 \ket{100} + (\sqrt{3} - 1) \ket{010} - (1 + \sqrt{3}) \ket{001})\ ,\nonumber\\ \ket{\eta_2} &= \frac{1}{\sqrt{2}}(\ket{01} - \ket{10}) \otimes \ket{0}\ ,
\end{align}
where $|0\rangle$ and $|1\rangle$ are orthonormal states.


\section{Semi-universality for \texorpdfstring{$\bm{{\SU(d)}}$}{SU(d)} symmetry}\label{sec:semiSU}

Next, we apply the tools discussed in the previous section to the important case $G = \SU(d)$, with the on-site representation of this symmetry on qudits.

\subsection{Global \texorpdfstring{$\SU(d)$}{SU(d)} symmetry and \texorpdfstring{$\SU(d)$}{SU(d)}-invariant unitaries }\label{sec:schur-weyl}

Consider a quantum system composed of $n$ qudits $\complex^d$, with local dimension $d \geq 2$. There is a natural representation of $\SU(d)$ on $(\complex^d)^{\otimes n}$, where the single-qudit unitary $u\in \SU(d)$ acts simultaneously on each qudit, $u \mapsto u^{\otimes n}$, corresponding to a global ``rotation'' of the system. In this case, the isotypic decomposition in \cref{eq:charge} takes the form
\begin{equation}\label{eq:schur-weyl}
  (\complex^d)^{\otimes n} \cong \bigoplus_{\lambda \in \Lambda_{n, d}} \hilbert_\lambda = \bigoplus_{\lambda \in \Lambda_{n, d}} \hilbert[Q]_\lambda \otimes \hilbert[M]_\lambda,
\end{equation}
where $\Lambda_{n, d}$ labels the inequivalent irreps of $\SU(d)$ which show up in $(\complex^d)^{\otimes n}$.

The elements of the set $\Lambda_{n, d}$ can be labeled by Young diagrams. Namely, elements of $\Lambda_{n, d}$ are in one-to-one correspondence with Young diagrams with $n$ boxes and $\leq d$ rows, satisfying the property that the number of boxes in each row is non-increasing. However, 
our proof of semi-universality of symmetric 3-qudit gates does not require the manipulation of Young diagrams as long as one accepts certain basic facts about the representation theory of the symmetric group, which are reviewed in \cref{sec:symfacts}. Any reader unfamiliar with Young diagrams may consider them an elaborate labeling scheme, and nothing more, for the purposes of this proof (we provide some details for the interested reader in \cref{fig:young}). 

Let $\mathcal{V}_k^{(n)} \subset \U((\complex^d)^{\otimes n})$ be the group generated by $k$-local $\SU(d)$-invariant unitaries, i.e., $\SU(d)$-invariant unitaries that can be written as ${V}'\otimes \mathbb{I}^{\otimes (n-k)}$ up to a permutation of qudits, where ${V}'$ acts on $k$ qudits. For each $k$, this is a compact, connected Lie group \cite{Marvian2022Restrict}. With the notation of \cref{sec:semiundef}, we have the equality 
\begin{equation}
  \mathcal{V}_n^{(n)} = \mathcal{V}^{\SU(d)} \cong \prod_{\lambda \in \Lambda_{n,d}} \U(\hilbert[M]_\lambda).
\end{equation}
Then, the commutator subgroup of $\mathcal{V}^{(n)}_n$, denoted as $\mathcal{SV}^{(n)}_n$, is isomorphic to $\prod_{\lambda \in \Lambda_{n,d}} \SU(\hilbert[M]_\lambda)$. We will sometimes use the notation $\mathcal{V}^{(n)} = \mathcal{V}_n^{(n)}$, and $\mathcal{SV}^{(n)} = \mathcal{SV}_n^{(n)}$. Similarly, for $k$-qudit gates we define
  \be
  \mathcal{SV}^{(n)}_k=[\mathcal{V}^{(n)}_k,\mathcal{V}^{(n)}_k]\ .
  \ee
Using these definitions, we can now present the formal statement of our result on the semi-universality of 3-qudit gates. 

\begin{theorem}\label{thm:semi-universality}
  3-qudit $\SU(d)$-invariant gates are semi-universal on arbitrary $n$ qudits. That is, 
  \begin{equation}\label{eq:semi-universality}
    [\mathcal{V}_3^{(n)}, \mathcal{V}_3^{(n)}] = \mathcal{SV}^{(n)} \cong \prod_{\lambda \in \Lambda_{n, d}} \SU(\hilbert[M]_\lambda).
  \end{equation}
  Specifically, 2-qudit $\SU(d)$-invariant gates together with any single 3-qudit gate that is not realizable with 2-qudit gates, such as gates $R_+$ or $R_-$, defined in \cref{reflection}, 
 are semi-universal.
\end{theorem}
See \cref{prop:2local-criteria} for a characterization of the 3-qudit gates that cannot be realized with 2-qudit $\SU(d)$-invariant unitaries.

In \cref{sec:2on3}, we show that in the special case of $d=2$, i.e., qubits with $\SU(2)$ symmetry, all 3-qudit SU(2)-invariant unitaries can be realized with 2-qudit SU(2)-invariant unitaries (which is not true for $d\ge 3$). Therefore, we have the following immediate corollary of this result, which was also previously established in \cite{marvian2022quditcircuit}.
\begin{corollary}
  For qubit systems with $\SU(2)$ symmetry, 2-qubit $\SU(2)$-invariant unitaries are semi-universal.
\end{corollary}

Before presenting the proof of this theorem, we briefly review a powerful tool for understanding the properties of $\SU(d)$-invariant unitaries, which plays a crucial role in our arguments: Schur-Weyl duality.

\subsection{Symmetric group \texorpdfstring{$\S_n$}{Sn} and Schur-Weyl duality}

An important class of $\SU(d)$-invariant unitaries are permutations. 
The symmetric group on $n$ objects $\mathbb{S}_n$ has a natural representation on $(\complex^d)^{\otimes n}$. In particular, for any permutation $\sigma \in \mathbb{S}_n$, let $\P(\sigma)$ denote the unitary operator on $(\complex^d)^{\otimes n}$ which permutes the qudits according to $\sigma$. Occasionally, we will also use the notation $\P_\sigma$.

Since permutations are $\SU(d)$-invariant, they are block-diagonal with respect to the decomposition in \cref{eq:schur-weyl}. Furthermore, according to the Schur-Weyl duality, $\Lambda_{n,d}$ also labels inequivalent irreps of $\mathbb{S}_n$, i.e., for all $\sigma\in \mathbb{S}_n$
\be
\P(\sigma)=\bigoplus_{\lambda \in \Lambda_{n, d}} \mathbb{I}_{\hilbert[Q]_\lambda} \otimes \P[\lambda](\sigma)\ ,
\ee
where $\mathbb{I}_{\hilbert[Q]_\lambda}$ is the identity operator on $\hilbert[Q]_\lambda$, and $\P[\lambda](\sigma)\in\U(\mathcal{M}_\lambda)$, defines an irrep of $\mathbb{S}_n$ on $\mathcal{M}_\lambda$ (see, e.g., \cite{goodman2009symmetry,harrow2005applications}). This group of unitaries is generated by the transposition (swap) of qubits $i$ and $j$, denoted by $\P(ij)$, or $\P_{ij}$ for simplicity. It follows that all permutations $\P(\sigma)$ belong to $\mathcal{V}_2^{(n)}$. We conclude that for $k \geq 2$, the group $\mathcal{V}_k^{(n)}$, as well as its commutator subgroup $\mathcal{SV}_k^{(n)}$, acts irreducibly on $\hilbert[M]_\lambda$ for all $\lambda \in \Lambda_{n,d}$.

Schur-Weyl duality implies that on any pair of distinct qudits $i$ and $j$, any 2-qudit $\SU(d)$-invariant unitary can be written as a global phase $\e^{\i \phi} \ident$ times 
\begin{equation}\label{eq:exch}
  \exp({\i \theta \P_{ij}}) = \cos \theta \, \ident + \i \sin \theta \, \P_{ij}\ ,
\end{equation}
for $\phi, \theta \in [0, 2 \pi)$, where $\ident$ denotes the identity operator on $\hilbert$. Therefore, $\mathcal{V}_2^{(n)}$ is the group generated by
\be
\mathcal{V}_2^{(n)}=\langle\exp({\i \theta \P_{ij}}), \exp(\i \phi) \ident: i\neq j , \theta,\phi\in [0, 2 \pi)\rangle\ .
\ee
See \cite{marvian2022quditcircuit} for further discussions on properties of group $\mathcal{V}_2^{(n)}$.

\section{Illustrative Example: \texorpdfstring{$\bm{n = 4}$}{n = 4} qudits}\label{sec:4ex}

As an illustrative example, in this section, we discuss the case of $n = 4$ qudits in depth and show that, 3-qudit $\SU(d)$-invariant unitaries are semi-universal in this case. This also provides the base case of an inductive proof that 3-qudit $\SU(d)$-invariant unitaries are semi-universal on arbitrary $n > 4$ qudits, discussed in \cref{sec:over}. 
Also, it demonstrates all the main techniques that are applied for the general proof. It is helpful to first consider 2-qudit gates---in doing so, the failure of  {2-qudit $\SU(d)$-invariant} {gates} to be semi-universal when $d > 2$ is illustrated.

 It is worth noting that while 3-qudit $\SU(d)$-invaraint gates are semi-universal on $n=4$ qudits, they are not universal. In particular, applying the results of \cite{Marvian2022Restrict}, reviewed in \cref{eq:centerdiff}, we find that the difference between the dimensions of the Lie group of all $\SU(d)$-invariant unitaries on $n=4$ qudits, and the subgroup generated by 3-qudit $\SU(d)$-invariant gates is equal to
\begin{align}
  \dim\V^{(4)}-\dim\V^{(4)}_3 = \abs{\Lambda_{4, d}} - \abs{\Lambda_{3, d}}=
  \begin{cases}
    1 &\!\! :d=3 \\
    2 &\!\! :d\ge 4
  \end{cases}
\end{align}
where $\abs{\Lambda_{n, d}}$ is the number of inequivalent irreps of $\SU(d)$, (or equivalently, $\mathbb{S}_n$) on $n$ qudits.

\subsection{2-qudit gates on \texorpdfstring{$n=3$}{n=3} qudits}\label{sec:2on3}

First, we consider $n = 3$ qudits. In this case, for general $d\geq 3$, the decomposition in \cref{eq:schur-weyl} contains 3 inequivalent irreps of $\mathbb{S}_3$, which are labeled by the Young diagrams 
\begin{align}\label{Eq:Lambda_3}
  \Lambda_3 & = \set*{\ydiag{3} \, , \; \ydiag{2, 1} \, , \; \ydiag{1, 1, 1}}\ , 
\end{align}
where $\sydiag{3}$ and $\sydiag{1, 1, 1}$ correspond to 1D irreps of $\mathbb{S}_3$, namely the trivial representation and the sign representation\footnote{Recall that $\sgn(\sigma)=\pm 1$ depending on whether the number of transpositions needed to realize $\sigma$ is odd or even.} $\P[\ysub{1,1,1}](\sigma)=\sgn(\sigma): \sigma\in \mathbb{S}_3$ (Note that the irrep $\sydiag{1,1,1}$ appears only for qudits with $d>2$.). Diagram $\sydiag{2, 1}$ corresponds to the 2D irrep of $\mathbb{S}_3$. 
  In particular, in Young's orthonormal basis for 
$\hilbert[M]_{\ysub{2, 1}}$, we obtain
\begin{align}
  [\P[\ysub{2, 1}](12)] = \begin{pmatrix} 1 & 0 \\ 0 & -1 \end{pmatrix} \label{eq:stair12} \ , \ \ 
  [\P[\ysub{2, 1}](23)] & = \frac{1}{2} \begin{pmatrix} -1 & \sqrt{3} \\ \sqrt{3} & 1 \end{pmatrix} \ , 
\end{align}
where we use the notation $[A]$ to denote the matrix representation of an operator $A$. Therefore, when projected to $\sydiag{2,1}$, unitaries $\exp({\i \theta \P(12)})$ and $\exp({\i \theta \P(23)})$ correspond to $\SU(2)$ rotations in two non-parallel directions, which means together they generate the full $\SU(2)$ group inside $\hilbert[M]_{\ysub{2, 1}}$.
In general, the projection of such unitaries inside $\hilbert[M]_{\ysub{3}}$ and 
$\hilbert[M]_{\ysub{1,1,1}}$ are non-trivial phases. However, such phases, which are allowed in $\mathcal{V}^{(3)}$ but not $\mathcal{SV}^{(3)}$, can be canceled by 
noting that $\SU(2)$ is equal to its commutator subgroup (In particular, any $U \in \SU(2)$ can be decomposed as $U = U_2^\dag U_1^\dag U_2U_1$ for some $U_1, U_2\in \SU(2)$ \cite{nielsen2000quantum}). 

We conclude that 2-qudit $\SU(d)$-invariant gates are semi-universal on $n=3$ qudits. That is
\be\label{SU1}
\mathcal{V}^{(3)}_2 \supseteq \mathcal{SV}^{(3)} \cong \SU(2) \ .
\ee
However, whether they are universal or not depends on $d$. In particular, the irrep $\sydiag{1,1,1}$, whose charge sector $\hilbert_{\ysub{1,1,1}}$ is the subspace of states of three qudits which are totally antisymmetric, shows up in the decomposition in \cref{eq:schur-weyl} only if $d \geq 3$. In this case, there is a one-parameter family of relative phases between charge sectors that cannot be generated by 2-local $\SU(d)$-invariant gates.  In particular, as we show in   \cref{app:V23char}, 
\begin{proposition}[Characterization of $\mathcal{V}^{(3)}_2$]\label{prop:2local-criteria} 
  For a system with $n=3$ qudits, the family of unitary evolutions $\exp(-\i H t): t\in\mathbb{R}$ is realizable with 2-qudit $\SU(d)$-invariant unitaries, i.e., $\exp(-\i H t)\in\mathcal{V}^{(3)}_2$ if, and only if $\Tr(H C)=0$, where 
  \begin{align}\label{eq:C3}
    C & = 2 (d - 1) (d - 2) \Pi_{\ysub{3}} - (d + 2) (d - 2) \Pi_{\ysub{2,1}} \nonumber\\
      & \mathrel{\phantom{=}} {} + 2 (d + 2) (d + 1) \Pi_{\ysub{1,1,1}} \\
      & = d^2(\P_{(123)} + \P_{(132)}) - 2 d (\P_{12} + \P_{13} + \P_{23}) + 4 \ident. \nonumber
  \end{align}
Furthermore, when $d \geq 3$, the unitary $V \in \mathcal{V}^{(3)}$ is realizable by 2-qudit $\SU(d)$-invariant unitaries, i.e. $V \in \mathcal{V}_2^{(3)}$, if and only if
\begin{equation}\label{eq:Vdet}
    \det v_{\ysub{2,1}} = (\det v_{\ysub{3}}) (\det v_{\ysub{1,1,1}}).
\end{equation}
\end{proposition}
\begin{remark}
  In the special case of $d=2$, operator $C=0$, and therefore the condition holds trivially, which means 2-qudit $\SU(2)$-invariant unitaries are indeed \emph{universal} on $n=3$ qubits, i.e., $\mathcal{V}^{(3)}_2=\mathcal{V}^{(3)}$. Furthermore, this implies that in this case, $\mathcal{V}^{(n)}_2=\mathcal{V}^{(n)}_3$ for all $n\ge 3$.
\end{remark}
The constraint imposed by \cref{eq:Vdet} implies that when $d \geq 3$, the difference between the dimensions of the Lie group of all $\SU(d)$-invariant unitaries on 3 qudits, and its subgroup generated by 2-qudits gates is
\begin{equation}
  \dim\mathcal{V}^{(3)} - \dim\mathcal{V}_2^{(3)} = 1\ ,
\end{equation}
saturating a general lower bound previously established in \cite{Marvian2022Restrict}. Therefore, a generic 3-qudit $\SU(d)$-invariant unitary cannot be generated by 2-qudit ones, since $\mathcal{V}_2^{(3)}$ forms a set of measure zero.

\subsection{2-qudit gates on \texorpdfstring{$n=4$}{n=4} qudits}\label{sec:2on4}

Next, we consider $n=4$ qudits. The group $\mathbb{S}_4$ has 5 inequivalent irreps, which are labeled by the Young diagrams 
\begin{align}\label{eq:S4mat}
  \Lambda_4 & = \set*{\ydiag{4} \, , \; \ydiag{3, 1} \, , \; \ydiag{2, 2} \, , \; \ydiag{2, 1, 1} \, , \; \ydiag{1, 1, 1, 1} \,}. 
\end{align}
For $d\ge 4$, all 5 diagrams appear in the decomposition in \cref{eq:schur-weyl}, whereas for $d=3$ the last diagram does not appear. Again, the diagrams $\sydiag{4}$ and $\sydiag{1, 1, 1, 1}$ correspond to the 1D irreps, namely the trivial and 
sign representations of $\mathbb{S}_4$, respectively. 

To determine which unitary transformations can be realized with 2-qudit $\SU(d)$-invariant gates, we consider the unitaries in the form $V\otimes \mathbb{I}: V\in \mathcal{SV}^{(3)}$, which according to \cref{SU1} form a group isomorphic to SU(2), and they can all be realized by 2-qudit $\SU(d)$-invariant gates that act trivially on qudit 4. Fortunately, to understand how these unitaries act on the multiplicity spaces $\hilbert[M]_{\ysub{2, 2}}$, 
$\hilbert[M]_{\ysub{3, 1}}$, and  
$\hilbert[M]_{\ysub{2,1, 1}}$, it suffices to understand the action of the $\mathbb{S}_3$ subgroup of $\mathcal{V}_2^{(3)}\otimes \mathbb{I}$ corresponding to permutations of the first 3 qudits. This follows from the fact that
\be
\{\textbf{P}(\sigma): \sigma\in \mathbb{S}_3\} \subset \mathcal{V}^{(3)}_2 \subset \text{span}_\mathbb{C}\{\textbf{P}(\sigma): \sigma\in \mathbb{S}_3\} \ .
\ee
Using the standard facts about the representation theory of $\mathbb{S}_4$ (see \cref{sec:symfacts}), it can be easily shown that a copy of the 2D irrep $\sydiag{2, 1}$ of $\mathbb{S}_3$ 
appear inside each of the multiplicity spaces $\hilbert[M]_{\ysub{2, 2}}$, $\hilbert[M]_{\ysub{3, 1}}$, and $\hilbert[M]_{\ysub{2, 1, 1}}$. In particular, under the restriction to $\mathbb{S}_3$ subgroup that acts on the first three qudits, the irreps of $\mathbb{S}_4$ decompose as 
\begin{equation}\label{dec24}
  \begin{split}
    \hilbert[M]_{\ysub{2, 2}}&\cong \hilbert[M]_{\ysub{2, 1}} \cong \mathbb{C}^2\\
    \hilbert[M]_{\ysub{3, 1}}&\cong\hilbert[M]_{\ysub{2, 1}}\oplus \hilbert[M]_{\ysub{3}} \cong \mathbb{C}^2\oplus \mathbb{C}\\ 
    \hilbert[M]_{\ysub{2,1, 1}}&\cong\hilbert[M]_{\ysub{2, 1}}\oplus \hilbert[M]_{\ysub{1,1,1}} \cong \mathbb{C}^2\oplus \mathbb{C}\ .
  \end{split}
\end{equation}
The first two rows of \cref{tab:n=4} give the explicit matrix representations of this $\mathbb{S}_3$ subgroup in irreps $\sydiag{3, 1}$ and $\sydiag{2,1, 1}$.

\begin{table}
  \begin{tblr}{|c|}
    \toprule
    \scalebox{0.8}[0.8]{\(\displaystyle [\P[\ysub{3,1}](\sigma)] = \begin{+pmatrix}
        \SetCell[r=2,c=2]{c} [\P[\ysub{2,1}](\sigma)] & & \vline \\
        & & \\ \hline
        & & \P[\ysub{3}](\sigma)
      \end{+pmatrix} = \begin{+pmatrix}
        \SetCell[r=2,c=2]{c} [\P[\ysub{2,1}](\sigma)] & & \vline \\
        & & \\ \hline
        & & 1
      \end{+pmatrix} \)} \\ \midrule
    \quad~~\scalebox{0.8}[0.8]{\(\displaystyle [\P[\ysub{2,1,1}](\sigma)] = \begin{+pmatrix}
        \SetCell[r=2,c=2]{c} [\P[\ysub{2,1}](\sigma)] & & \vline \\
        & & \\ \hline
        & & \P[\ysub{1,1,1}](\sigma)
      \end{+pmatrix} = \begin{+pmatrix}
        \SetCell[r=2,c=2]{c} [\P[\ysub{2,1}](\sigma)] & & \vline \\
        & & \\ \hline
        & & \sgn(\sigma)
      \end{+pmatrix} \)} \qquad \\ \midrule
    \quad
    \scalebox{0.8}[0.8]{\(\displaystyle [J] = \begin{+pmatrix}
         & 1 & \vline  \\
        -1 &  &  \\ \hline
         &  & 1
      \end{+pmatrix}\)} \qquad
    \scalebox{0.8}[0.8]{\(\displaystyle [\P[\ysub{3, 1}](34)] = \begin{+pmatrix}
        \frac{1}{3} & & \vline \frac{2 \sqrt{2}}{3} \\
        & 1 & \\ \hline
        \frac{2 \sqrt{2}}{3} & & -\frac{1}{3}
      \end{+pmatrix}\)}
    \\ \bottomrule
  \end{tblr}
  \caption{In the first two rows $\sigma$ is in $\mathbb{S}_3$ subgroup of $\mathbb{S}_4$ that acts trivially on qudit 4, and $[\P[\ysub{2,1}](\sigma)]$ is given by \cref{eq:stair12}. Matrix $[J]$ satisfies \cref{eq:J4}. }
  \label{tab:n=4}
\end{table}

\cref{dec24} immediately implies that the projection of $\mathcal{SV}^{(3)}_2 \otimes \mathbb{I}$ to $\hilbert[M]_{\ysub{2, 2}}$ is isomorphic to $\SU(2)$, i.e.,
\be
\pi_{\ysub{2, 2}}(\mathcal{SV}^{(3)}_2\otimes \mathbb{I})=\SU(\hilbert[M]_{\ysub{2, 2}})\cong \SU(2)\ .
\ee
In other words, condition \textbf{A} in \cref{MainLemma}, namely, subsystem universality, is satisfied in the sector with irrep 
$\sydiag{2, 2}$.

Next, we consider the projection of this 
group 
to $\mathcal{M}_{\ysub{3,1}}$ and $\mathcal{M}_{\ysub{2,1,1}}$. Relative to the decomposition in \cref{dec24}, this group of unitaries will be in the form
\begin{align}\label{dec25}
  \begin{+pmatrix}
    \SetCell[r=2,c=2]{c} \SU(2) & & \vline \\
    & & \\ \hline
    & & 1
  \end{+pmatrix}\ .
\end{align}

Note that this block-diagonal form is a consequence of the fact that under the operators that act on the first 3 qudits, the $\SU(d)$ charge of qudit 4 is conserved. Now suppose we include 2-qudit $\SU(d)$-invariant unitaries $\exp(\i\theta \textbf{P}(34)): \theta\in[0,2\pi)$ that act on qudits 3 and 4, which allows qudit 4 to interact with the rest of the system. Since the transposition $\textbf{P}(34)$ together with the aforementioned $\mathbb{S}_3$ subgroup generate $\mathbb{S}_4$ and $\hilbert[M]_{\ysub{2,1, 1}}$ and $\hilbert[M]_{\ysub{3,1}}$ are both irreps of $\mathbb{S}_4$, the projection of these unitaries inside 
$\hilbert[M]_{\ysub{2,1, 1}}$ and $\hilbert[M]_{\ysub{3,1}}$ cannot be block-diagonal with respect to the decompositions in \cref{dec24}.

In summary, inside both subspaces $\hilbert[M]_{\ysub{2,1, 1}}$ and $\hilbert[M]_{\ysub{3,1}}$ 2-qudit unitaries acting on qudits 1, 2 and 3 generate the block-diagonal unitaries in the form \cref{dec25}, and 
$\exp(\i\theta \textbf{P}(34)): \theta\in[0,2\pi)$ 
is a one-parameter family of unitaries that are not block diagonal with respect to this decomposition. According to \cref{rem1}, adding any such one-parameter family of unitaries to unitaries in \cref{dec25}, generates a group that contains $\SU(3)$ (see also Lemma 5 of \cite{marvian2022quditcircuit}). 

Therefore, similar to irrep $\sydiag{2, 2}$, condition \textbf{A} in \cref{MainLemma} is also satisfied in irreps $\sydiag{3, 1}$ and $\sydiag{2, 1, 1}$. Since $\mathcal{M}_{\ysub{4}}$ and $\mathcal{M}_{\ysub{1, 1, 1, 1}}$ are 1D, this condition is trivially satisfied in those irreps. We conclude that condition \textbf{A} is fully satisfied, i.e., for all $\lambda\in \Lambda_4$, the projection of $\mathcal{V}^{(4)}_2$ to $\mathcal{M}_\lambda$ 
is equal to the projection of $\mathcal{V}^{(4)}$ to $\mathcal{M}_\lambda$:
\begin{align}
  \forall\lambda\in\Lambda_4:\ \ \ \ \pi_\lambda(\mathcal{SV}^{(4)}_2)=\pi_\lambda(\mathcal{SV}^{(4)})=\SU(\mathcal{M}_\lambda)\ .
\end{align}

However, simply realizing all possible unitaries in the multiplicity spaces does not preclude the possibility of correlation between sectors. That is, condition \textbf{B} may not be satisfied, and in fact, it turns out that this is the case! The unitaries realized in $\hilbert[M]_{\ysub{3, 1}}$ and $\hilbert[M]_{\ysub{2,1,1}}$ uniquely determine each other, up to a global phase. This is a consequence of a standard representation isomorphism between pairs of irreps of $\mathbb{S}_4$ that are related by the sign representation, namely 
\begin{align}
  \ydiag{3, 1}\cong \sgn \otimes ~\ydiag{2, 1, 1}\ \ \ , \quad \quad \ydiag{2, 1, 1}\cong \sgn \otimes ~\ydiag{3, 1} \ .
\end{align}
That is, there exists a unitary operator $J : \hilbert[M]_{\ysub{3,1}} \to \hilbert[M]_{\ysub{2,1,1}}$ such that
\begin{equation}\label{eq:J4}
  \sigma\in\mathbb{S}_4:\ \ J \P[\ysub{3, 1}](\sigma) J^\dag = \sgn(\sigma) \P[\ysub{2, 1, 1}](\sigma)\ .
\end{equation}
As presented in \cref{tab:n=4}, we can pick Young's basis where permutations have real matrix representation, denoted as $[\P[\ysub{3, 1}](\sigma)]$ and $[\P[\ysub{2, 1, 1}](\sigma)]$. Then, relative to this basis, $J$ also becomes a real orthogonal matrix $[J]$. 
Since for any transposition (SWAP) $(ij)$ the parity $\sgn((ij))=-1$, this in turn implies
\begin{equation}
  [J] [\exp (\i\theta \P[\ysub{3, 1}](ij))] [J]^{\text{T}} = [\exp (\i\theta \P[\ysub{2, 1, 1}](ij))]^\ast\ ,
\end{equation}
where $[J]^{\text{T}}$ is the transpose of matrix $[J]$ in the aforementioned basis. 
Recall that any element of $\mathcal{V}^{(4)}_2$ can be decomposed to a sequence of 2-qudit gates $\exp(\i \theta \P_{ij})$ and a global phase. 
 Because products of unitaries that satisfy the above constraint also satisfy this constraint, it follows that for any $V\in \mathcal{V}^{(4)}_2$
\begin{equation}\label{const}
  \forall V\in \mathcal{V}^{(4)}_2 :\ \ [J] [v_{\ysub{3, 1}}] [J]^{\text{T}} = \e^{\i\phi}\ [v_{\ysub{2, 1, 1}}]^\ast\ ,
\end{equation}
where $v_{\ysub{3, 1}}=\pi_{\ysub{3, 1}}(V)$, $v_{\ysub{2, 1, 1}}=\pi_{\ysub{2, 1, 1}}(V)$ 
are the components of $V$ in $\mathcal{M}_{\ysub{3, 1}}$ and $\mathcal{M}_{\ysub{2, 1,1 }}$ respectively, $[v_{\ysub{2, 1, 1}}]^\ast$ is the complex conjugate of $[v_{\ysub{2, 1, 1}}]$, and $\e^{\i\phi}$ is an unspecified phase that depends on $V$.

The relation in \cref{const} means that the joint projection of $\mathcal{SV}_2^{(4)}$ to multiplicity spaces $\mathcal{M}_{\ysub{3, 1}}$ and $\mathcal{M}_{\ysub{2, 1, 1}}$ does not contain $\SU(3)\times \SU(3)$, as required by semi-universality. Rather, it satisfies
\begin{equation}\label{re1}
  \begin{split}
    \pi_{\ysub{3, 1},\, \ysub{2, 1, 1}}(\mathcal{SV}_2^{(4)}) &:=\big\{ (v_{\ysub{3, 1}}, v_{\ysub{2, 1, 1}}): V\in\mathcal{SV}_2^{(4)} \big\}\\ &\cong \SU(\mathcal{M}_{\ysub{2, 1, 1}}) \cong \SU(\mathcal{M}_{\ysub{3, 1}}) \cong \SU(3)\ ,
  \end{split}
\end{equation}
and the explicit form of the isomorphism is given by \cref{const}. In this situation we say the unitaries realized in $\sydiag{2, 1, 1}$ and $\sydiag{3, 1}$ are ``correlated". 

On the other hand, since $\mathcal{M}_{\ysub{2, 2}}$ has a different dimension from $\mathcal{M}_{\ysub{3, 1}}$ and $\mathcal{M}_{\ysub{2, 1, 1}}$, condition \textbf{B} is automatically satisfied for the pair $\sydiag{2, 2}$ and $\sydiag{3, 1}$, and the pair $\sydiag{2, 2}$ and $\sydiag{2,1, 1}$. In particular, according to \cref{both}, for both $\lambda=\sydiag{3, 1}$ and $\lambda=\sydiag{2,1, 1}$ the joint projection of $\mathcal{SV}_2^{(4)}$ to $\lambda$ and $\sydiag{2, 2}$ is isomorphic to
\begin{equation}\label{ind}
  \begin{split}
    \pi_{\lambda ,\, \ysub{2, 2}}(&\mathcal{SV}_2^{(4)}) := \{ (v_{\lambda}, v_{\ysub{2, 2}}): V\in\mathcal{SV}_2^{(4)}\} \\ &\cong \SU(\mathcal{M}_\lambda)\times \SU(\mathcal{M}_{\ysub{2, 2}}) \cong \SU(3)\times \SU(2)\ .
  \end{split}
\end{equation}
In summary, the commutator subgroup of $\mathcal{V}^{(4)}_2$ is
$$\mathcal{SV}^{(4)}_2=[\mathcal{V}^{(4)}_2,\mathcal{V}^{(4)}_2]\cong \SU(2)\times \SU(3)\ ,$$
whereas the commutator subgroup of the group of all $\SU(d)$-invariant unitaries is 
$$\mathcal{SV}^{(4)}=[\mathcal{V}^{(4)},\mathcal{V}^{(4)}] \cong \SU(2)\times \SU(3)\times \SU(3)\ ,$$
which means 2-qudit $\SU(d)$-invariant gates are not semi-universal. We revisit this example in \cref{sec:fail}, and show that while $\mathcal{V}_2^{(4)}$ does not satisfy condition \textbf{B} of \cref{MainLemma}, because it still satisfies condition \textbf{A} (subsystem universality), it can be characterized via  \cref{Mainlemma_gen}. (This lemma is an extension of \cref{MainLemma} that does not assume condition \textbf{B}.)

Here, we saw the failure of semi-universality and constraints imposed by the locality of interactions in the case of a system with $n=4$ qudits, as a consequence of \cref{const}. Indeed, this relation can be understood in terms of an anti-unitary transformation. See \cite{marvian2022quditcircuit} for further details, where these facts can be interpreted both in terms of 
 conservation of an observable $K$ defined on two copies of the system, and in the language of a freely evolving fermionic system. This phenomenon, i.e., that the dynamics of 2-qudit symmetric interactions in distinct charge sectors determine each other, exists for general $n$ when $d>2$.\footnote{The reason that it does not exist for $d=2$, i.e., that 2-qudit $\SU(2)$-invariant unitaries are semi-universal on qubits, is related to the fact that in this case there is no representation associated with the diagram $\ysub{2,1,1}$.} For sufficiently many qudits $n$, as the dimension of the Hilbert space of each qudit $d$ grows, there will be more such constraints among different sectors \cite{marvian2022quditcircuit}. Furthermore, in addition to such correlations, which correspond to type \textbf{IV} constraints, we also see type \textbf{III} constraints. For instance, for $n=5$ qudits with $d\ge 3$, type \textbf{III} constraints appear: in the charge sector $\sydiag{3,1,1}$ we get the group $\SO(6)$ rather than $\SU(6)$ \cite{marvian2022quditcircuit}, which is required by semi-universality. In addition, the sectors $\sydiag{3,2}$ and $\sydiag{2,2,1}$ are correlated with each other, and similarly, the three sectors $\sydiag{4,1}$, $\sydiag{3,1,1}$, and $\sydiag{2,1,1,1}$ are all correlated.  It is worth noting again that these phenomena, i.e., correlations between unitaries in sectors with inequivalent irreps of $G=\SU(d)$ (type $\textbf{IV}$) and appearance of irreducible subgroups of the realizable unitaries in one sector (type $\textbf{III}$), cannot occur in the case of Abelian symmetry groups $G$ \cite{marvian2024theoryabelian}.

\subsection{
  Generic 3-qudit gates are semi-universal}\label{sec:4ex3loc}

Next, we show that 2-qudit $\SU(d)$-invariant gates together with any $\SU(d)$-invariant unitary $Y\in \mathcal{V}^{(4)}$ 
that breaks the constraint in \cref{const} become semi-universal on $n=4$ qudits. Furthermore, such $Y$ can be chosen as a 3-qudit gate that acts on 3 out of 4 qudits in the system (tensor product with the identity operator on the 4th qudit).

\begin{proposition}\label{4qudit}
  Any 4-qudit $\SU(d)$-invariant unitary $Y$ together with 2-qudit $\SU(d)$-invariant unitaries are semi-universal on $n=4$ qudits, i.e., $\mathcal{SV}^{(4)}\subset \langle Y, V: V \in \mathcal{V}_2^{(4)}\rangle$ if, and only if,
  \begin{equation}\label{cond1}
    [J] [\pi_{\ysub{3, 1}}(Y)] [J]^{\text{T}} \neq \e^{\i\phi}[\pi_{\ysub{2, 1, 1}}(Y)]^\ast\ ,
  \end{equation}
  for any phase $\e^{\i\phi}$, 
  where $[\pi_{\ysub{3, 1}}(Y)]$ and $[\pi_{\ysub{2, 1, 1}}(Y)]$ are the matrix representations of the components of $Y$ in $\mathcal{M}_{\ysub{3, 1}}$ and $\mathcal{M}_{\ysub{2, 1,1}}$, in the basis defined in the \cref{tab:n=4} where matrix $[J]$ is also defined. 
\end{proposition}
This is proven in \cref{app:3quditgate}. Note that if the unitary $Y\in\mathcal{V}^{(4)}$ satisfies
\be\label{eq:trcond}
|\Tr(\pi_{\ysub{3, 1}}(Y))|\neq |\Tr(\pi_{\ysub{2, 1, 1}}(Y))|\ ,
\ee
then \cref{cond1} cannot hold as equality for any phase $\phi$ and unitary $J$, which means, together with 2-qudit unitaries, $Y$ is semi-universal. This allows using the characterization of elements of $\mathcal{V}_2^{(3)}$ in \cref{prop:2local-criteria} to demonstrate that a generic 3-qudit $\SU(d)$-invariant unitary, together with 2-qudit ones, is semi-universal.

\begin{proposition}[Generic 3-qudit gates are semi-universal on $n = 4$ qudits]\label{prop:gen3q} For any 3-qudit unitary $S\in\mathcal{V}^{(3)}$ that is not realizable with 2-qudit $\SU(d)$-invariant gates, i.e., $S\in \mathcal{V}^{(3)}\setminus\mathcal{V}_2^{(3)}\ ,$ the 4-qudit unitary 
  $Y=S\otimes \mathbb{I}$ satisfies the requirement in \cref{cond1} of the above proposition, and therefore the 3-qudit gate $S$ together with 2-qudit unitaries $\exp(\i \theta \P_{ij}): \theta\in[0,2\pi)$ are semi-universal, {i.e., $\mathcal{SV}^{(4)}\subset \langle S\otimes \mathbb{I}, V: V\in\mathcal{V}^{(4)}_2 \rangle$. 
  }
\end{proposition}

Together with the results of \cref{sec:over}, in particular \cref{eq:4semi}, this implies that any 3-qudit gate $S \in \mathcal{V}^{(3)} \setminus \mathcal{V}_2^{(3)}$ together with 2-qudit $\SU(d)$-invariant gates is semi-universal on $n$ qudits, for all $n \geq 3$.

\cref{prop:gen3q} is proven in \cref{app:3quditgate}. Next, we present an example of a unitary $Y$ satisfying the constraint in the above lemma.

\subsection*{Example: Semi-universality with 3-qudit reflections}

Recall the three-qudit reflection unitaries $R_+=\exp({\i\pi\Pi_{\ysub{3}}})$ and $R_-=\exp({\i\pi\Pi_{\ysub{1,1,1}}})$ in \cref{reflection}.  Suppose in addition to the 2-qudit $\SU(d)$-invariant gates, one can also use one of the gates $R_+$ or $R_-$. Is this set semi-universal on $n=4$ qudits? 

It is easy to verify that $R_\pm \not\in \mathcal{V}_2^{(3)}$ when $d \geq 3$ using \cref{prop:2local-criteria}. For instance, 
rewriting as
\begin{equation}\label{eq:Rre}
  R_+ = -\Pi_{\ysub{3}} + \Pi_{\ysub{2, 1}} + \Pi_{\ysub{1,1,1}},
\end{equation}
we find that \cref{eq:Vdet} does not hold, i.e.,
\begin{equation}
    1 = \det \pi_{\ysub{2,1}}(R_+) \neq (\det \pi_{\ysub{3}}(R_+)) (\det \pi_{\ysub{1,1,1}}(R_+)) = -1.
\end{equation}
Equivalently, we can see this using the criterion in \cref{4qudit} or trace condition in \cref{eq:trcond}.
In particular, it can be easily seen that $R_\pm\otimes \mathbb{I}$ acts trivially on one of $\mathcal{M}_{\ysub{2, 1,1 }}$ and $\mathcal{M}_{\ysub{3, 1}}$ and non-trivially in the other (depending on the sign $\pm$). For instance, while $R_+\otimes \mathbb{I}$ acts trivially in $\mathcal{M}_{\ysub{2, 1,1 }}$, inside $\mathcal{M}_{\ysub{3, 1}}$ it has the matrix representation 
\begin{align}
 [\pi_{\ysub{3,1 }}(R_+\otimes \mathbb{I})] & = \begin{+pmatrix}
 1 & & \vline \\
 & 1 & \\ \hline
 & & -1
 \end{+pmatrix} \ ,
\end{align}
which can be seen from \cref{eq:Rre} together with \cref{tab:n=4}. This means \cref{cond1} cannot hold as equality for any phase $\e^{\i\phi}$. In particular,
\begin{equation}
  \begin{split}
    |\Tr (\pi_{\ysub{3,1}}(R_+))| & = 1 \\
    |\Tr (\pi_{\ysub{2,1,1}}(R_+))| & = 3 \ ,
  \end{split}
\end{equation}
so the condition in \cref{eq:trcond} holds.

It is also worth noting that for generic values of $\theta\in [0,2\pi)$, the unitary realized by the following circuit is in 
$\mathcal{SV}^{(4)}$ but not in $\mathcal{SV}_2^{(4)}$.
$$\Qcircuit @C=1em @R=.7em {
 & \multigate{2}{R_\pm} & \qw & \multigate{2}{R_\pm} & \qw & \qw \\
 & \ghost{R_\pm} & \qw & \ghost{R_\pm} & \qw & \qw \\
 & \ghost{R_\pm} & \multigate{1}{\e^{\i \theta \P(34)}} & \ghost{R_\pm} & \multigate{1}{\e^{-\i \theta \P(34)}} & \qw \\
 & \qw & \ghost{{\e^{\i \theta \P(34)}}} & \qw 
& \ghost{{\e^{-\i \theta \P(34)}}} & \qw 
 }$$\\
Indeed, since $R_\pm$ acts trivially in $\mathcal{M}_{\ysub{2,2}}$ and one of $\mathcal{M}_{\ysub{2,1,1}}$ or $\mathcal{M}_{\ysub{3,1}}$, the unitary realized by the above circuit acts trivially in all sectors except $\mathcal{M}_{\ysub{3, 1}}$, or $\mathcal{M}_{\ysub{2, 1, 1}}$, depending on the sign $\pm$ of $R_\pm$. Furthermore, in this sector, the realized unitary does not act as a global phase, which can be seen by noting that it is not permutationally invariant. 
Therefore, \cref{4qudit} implies that each one of $R_+$ or $R_-$, together with 2-qudit $\SU(d)$-invariant gates are semi-universal on $n=4$ qudits.

It is worth noting that the reflection unitary for the charge sector $\sydiag{2,1}$, i.e., $
\exp({\i\pi\Pi_{\ysub{2,1}}})$, is already contained in $\mathcal{SV}^{(3)} \subseteq \mathcal{V}_2^{(3)}$, since $\det \pi_\lambda(\exp({\i\pi\Pi_{\ysub{2,1}}})) = 1$ for all $\lambda \in \Lambda_{3, d}$.

\section{Semi-universality of 3-qudit gates on arbitrary number of qudits}\label{sec:over}

Next, we prove the semi-universality of 3-qudit gates for systems with $n\ge 4$, i.e., \cref{thm:semi-universality}. To establish this result, we use induction. Applying \cref{MainLemma}, we show that by composing $m$-qudit unitaries in $\mathcal{SV}^{(m)}$, one can obtain any element of $\mathcal{SV}^{(m+1)}$ on $m+1$ qudits. More precisely, let
\begin{equation}\label{eq:Wdef}
    \mathcal{W}^{(m + 1)} = \angles{\P_{ij} (\mathcal{SV}^{(m)} \otimes \ident) \P_{ij} \given i \neq j} 
\end{equation}
denote the group generated by $m$-local unitaries in $\mathcal{SV}^{(m)}$ acting on $m + 1$ qudits. It is worth noting that $\mathcal{W}^{(m + 1)}$ is generated by those unitaries restricted to only acting on the first $m$ or the last $m$ qudits, i.e.\footnote{This follows because every even permutation, i.e. element of the alternating subgroup $\sigma \in \mathbb{A}_m \subseteq \mathbb{S}_m$, is contained in $\mathcal{SV}^{(m)}$, as $\det \P[\lambda](\sigma) = 1$ for all $\lambda \in \Lambda_n$. Thus the group on the right-hand side of \cref{eq:Wgens} is invariant under conjugation by $\P(\sigma) : \sigma \in \mathbb{A}_{m + 1}$ and by $\P_{23}$, which together generate all of $\mathbb{S}_{m + 1}$.}
\begin{equation}\label{eq:Wgens}
 \mathcal{W}^{(m + 1)} = \angles{ \mathcal{SV}^{(m)} \otimes \ident , \ident \otimes \mathcal{SV}^{(m)}}\ .
\end{equation}

Then, we show that
\begin{equation}\label{sf}
\mathcal{W}^{(m + 1)} = \mathcal{SV}^{(m + 1)} \quad \text{ for } m \geq 4\ .
\end{equation}
This means that for any $n\ge 4$, the group generated by the permuted versions of $\mathcal{SV}^{(4)}$ is equal to $\mathcal{SV}^{(n)}$, i.e., 
\begin{equation}\label{eq:4semi}
 \angles{\P_{ij} (\mathcal{SV}^{(4)} \otimes \ident^{\otimes (n-4)}) \P_{ij} : i\neq j}=\mathcal{SV}^{(n)}\ .
\end{equation}
As we showed in the previous section for $4$ qudits, $\mathcal{SV}^{(4)}$ can be generated by 3-qudit $\SU(d)$-invariant unitaries. Combining these results, we conclude that 3-qudit $\SU(d)$-invariant gates are semi-universal on $n$ qudits systems, i.e., they generate $\mathcal{SV}^{(n)}$, which proves \cref{thm:semi-universality}. 

To prove \cref{sf}, first we note that $\mathcal{W}^{(m+1)}\subseteq \mathcal{SV}^{(m+1)}$, which follows from the fact that 
\be\label{incl}
\mathcal{SV}^{(m)}\otimes \mathbb{I} \subseteq \mathcal{SV}^{(m+1)}\ .
\ee
This can be seen, for instance, by noting that $\mathcal{SV}^{(m)}\otimes \mathbb{I} \subseteq \mathcal{V}^{(m+1)}$, together with $\mathcal{SV}^{(m)}=[\mathcal{SV}^{(m)}, \mathcal{SV}^{(m)}]$. 

Therefore, to prove \cref{sf} we need to show that $\mathcal{W}^{(m+1)} $ contains all elements of $\mathcal{SV}^{(m+1)}$. The example case $n = 4$, discussed in \cref{sec:4ex}, introduces most of the main ideas that are needed in this proof. The only new ingredients are a few useful facts about the representation theory of $\mathbb{S}_n$ that are reviewed in the next section, and the following lemma which is of independent interest.
\begin{lemma}\label{lem-2-local}
On a system with $n$ qudits, consider the group generated by one-parameter families $\e^{\i t (\P_{ij}-\P_{kl})}$ acting on qudits $i\neq j$ and $k\neq l$. This group is equal to the commutator subgroup of $\mathcal{V}_2^{(n)}$, i.e. 
\begin{equation}\label{eq:swapirrep}
  \langle \e^{\i t (\P_{ij}- \P_{kl})} : t \in \real, i \neq j, k \neq l \rangle= \mathcal{SV}_2^{(n)}\ .
\end{equation}
In particular, for all irreps $\lambda\in\Lambda_{n,d}$, this group acts irreducibly on $\mathcal{M}_\lambda$.
\end{lemma}
This is proven in \cref{sec:commconn} using the fact that any swap $\i \P_{ij}$ can be written as a linear combination of the differences $\i(\P_{ij}- \P_{kl}): i \neq j, k \neq l $ and the permutationally-invariant operator $\i B=\i \sum_{i\neq j} \P_{ij}$, which is in the center of the Lie algebra generated by transpositions.

\subsection{Useful facts about the symmetric group \texorpdfstring{$\S_n$}{Sn}}\label{sec:symfacts}

Recall the isotypic decomposition in \cref{eq:schur-weyl},
\begin{equation}\nonumber
 (\complex^d)^{\otimes (m + 1)} = \bigoplus_{\lambda \in \Lambda_{m + 1, d}} \hilbert[Q]_\lambda \otimes \hilbert[M]_\lambda\ ,
\end{equation}
where $\Lambda_{m + 1, d}$ is the set of Young diagrams with $m + 1$ boxes and $\leq d$ rows, and $\hilbert[M]_\lambda$ is the irrep of $\mathbb{S}_{m + 1}$ labeled with the Young diagram $\lambda$. Consider $\mathbb{S}_m \subset \mathbb{S}_{m + 1}$ the subgroup which permutes the first $m$ qudits. Then, the branching rule {of $\mathbb{S}_m$} says that, for each $\lambda \in \Lambda_{m + 1, d}$, there is $\Gamma_\lambda \subseteq \Lambda_{m, d}$ such that
\begin{equation}\label{eq:branch}
 \hilbert[M]_\lambda = \bigoplus_{\gamma \in \Gamma_\lambda} \hilbert[M]_\gamma\ ,
\end{equation}
where {each irrep in $\gamma\in \Gamma_\lambda$ of $\mathbb{S}_m$ shows up with multiplicity one and the subspace} $\hilbert[M]_\gamma$ carries that irrep \cite{fulton2013representation,okounkov2005}.
For the proof of \cref{thm:semi-universality}, in addition to this fact, we will also use the following facts.\\ 
\begin{fact*}
 Whenever $m \geq 4$ and $\lambda \in \Lambda_{m + 1, d}$ is not one-dimensional (i.e., it is neither the trivial nor sign irrep of $\mathbb{S}_{m + 1})$,
\begin{enumerate}[label=\arabic*.,ref=Fact \arabic*]
 \item If $\lambda$ has $\leq d$ rows then each $\gamma \in \Gamma_\lambda$ has $\leq d$ rows.\label{fact:0}
 \item At least one irrep in the branching, $\gamma \in \Gamma_\lambda \subseteq \Lambda_{m, d}$, has $\dim \hilbert[M]_\gamma \geq 3$.\label{fact:1}
 \item For any distinct $\lambda \neq \lambda' \in \Lambda_{m + 1, d}$, there is some $\gamma \in \Lambda_{m, d}$ with $\dim \hilbert[M]_\gamma \geq 2$ which is in one of, but not both, $\Gamma_\lambda$ or $\Gamma_{\lambda'}$.\label{fact:2}
\end{enumerate}
\end{fact*}
These facts can be deduced from Young's lattice in \cref{fig:young} (in the caption we explain these three facts for the example of $n=5$). Note that \ref{fact:1} and \ref{fact:2} do not hold for $n=4$, and this is another way to understand why 2-qudit gates are not semi-universal on $n=4$ qudits (or, more precisely,  why the induction starts at $m=4$ rather than $m=3$).
 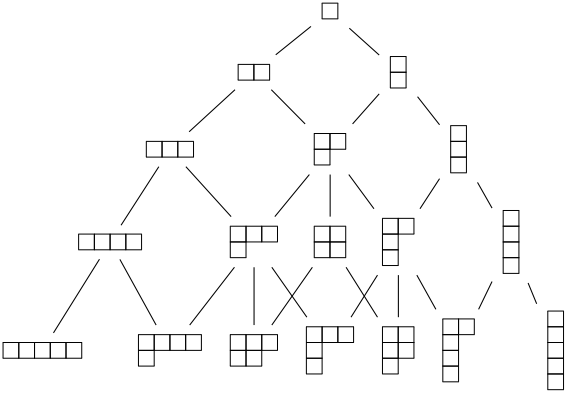
\begin{figure}
\begin{tikzcd}[tips=false,column sep=0.3em,row sep=0.8em]
 &&&& \ydiag{1} \ar[dl] \ar[dr] &&&& \\
 &&& \ydiag{2} \ar[dl] \ar[dr] && \ydiag{1,1} \ar[dl] \ar[dr] &&& \\
 && \ydiag{3} \ar[dl] \ar[dr] && \ydiag{2,1} \ar[dl] \ar[d] \ar[dr] && \ydiag{1,1,1} \ar[dl] \ar[dr] && \\
 & \mathclap{\ydiag{4}} \ar[dl] \ar[dr] && \ydiag{3,1} \ar[dl] \ar[d] \ar[dr] & \ydiag{2,2} \ar[dl] \ar[dr] & \ydiag{2,1,1} \ar[dl] \ar[d] \ar[dr] && \ydiag{1,1,1,1} \ar[dl] \ar[dr] & \\
 \ydiag{5} && \ydiag{4,1} & \ydiag{3,2} & \ydiag{3,1,1} & \ydiag{2,2,1} & \ydiag{2,1,1,1} && \ydiag{1,1,1,1,1}
\end{tikzcd}
\caption{The beginning of Young's lattice, from which the branching rule may be read: the representations that show up in the branching \cref{eq:branch} of the Young diagram $\lambda$ are precisely the subdiagrams, which are connected by an edge in the lattice. The last layer corresponds to $n=5$ qudits. One can easily check 
\ref{fact:0}, \ref{fact:1}, and \ref{fact:2}. For instance,  \cref{eq:branch} implies that except for the two Young diagrams at the two ends of each layer, namely single-row and single-column Young diagrams, the rest of the diagrams have dimensions larger than 1, and indeed for diagrams with $n=5$, they all have dimensions 3 or larger.  
We can explicitly check each fact for $\lambda \in \Lambda_{5, d}$. According to \ref{fact:0}, for any irrep $\lambda\in \Lambda_{5, d}$, the corresponding irreps in their branching denoted as $\Gamma_\lambda$, which are the diagrams connected to $\lambda$ in the previous level, have less or equal number of rows. Furthermore, as stated in \ref{fact:1}, all the irreps $\lambda \in \Lambda_{5, d}$ that are not one-dimensional, namely everything except the diagrams at the left and right corners, have either $\sydiag{3, 1}$ or $\sydiag{2, 1, 1}$ in their branching, both of which are three-dimensional. Also, as stated in \ref{fact:2}, for any distinct irreps $\lambda \neq \lambda' \in \Lambda_{5, d}$ which are not one-dimensional, at least one of $\sydiag{3, 1}$, $\sydiag{2,2}$, or $\sydiag{2, 1,1}$ is in the branching of $\lambda$ but not $\lambda'$.
\label{fig:young}}
\end{figure}

\subsection{Extending semi-universality from \texorpdfstring{$m$}{m} to \texorpdfstring{$m+1$}{m+1} qudits (Proof of Eq. (\ref{sf}))}

Here, we prove \cref{sf}. Similar to the proof in the special case of $n=4$, the proof is in 2 steps, namely we establish conditions \textbf{A} and \textbf{B} of \cref{MainLemma}. \\

\noindent $\bullet$ Condition \textbf{A}: We show that for each irrep $\lambda\in\Lambda_{m+1, d}$, the projection of $\mathcal{W}^{(m+1)}$ to irrep $\mathcal{M}_\lambda$ is equal to
\be\label{proj3}
\pi_\lambda(\mathcal{W}^{(m+1)})=\pi_\lambda(\mathcal{SV}^{(m+1)})=\SU(\mathcal{M}_\lambda)\ .
\ee
Since this condition is trvially satisfied when $\dim(\mathcal{M}_\lambda)=1$, in the following, without loss of generality, we assume $\dim(\mathcal{M}_\lambda)>1$. 
Similar to the case of $n=4$ qudits, first we consider elements of $\mathcal{W}^{(m+1)}$ that act trivially on qudit $m+1$, i.e., unitaries in the form $(V\otimes \mathbb{I}): V\in \mathcal{SV}^{(m)}$. Relative to the decomposition in \cref{eq:branch}, the projection of this subgroup of $\mathcal{W}^{(m+1)}$ to $\mathcal{M}_\lambda$ is equal to
\begin{align}
    \{ \pi_\lambda(V\otimes \mathbb{I}): V\in \mathcal{SV}^{m})\}= \bigoplus_{\gamma \in \Gamma_\lambda} \SU(\hilbert[M]_\gamma)\ .
\end{align}
Again, this block-diagonal form is a consequence of the fact that qudit $m+1$ does not interact with the rest of qudits. Then, we argue that as soon as we include interactions with this qudit we get the entire $\SU(\hilbert[M]_\lambda)$. To achieve this we apply \cref{thm:mulblocks}, whose assumptions are verified in the following: 

\begin{enumerate}
 
\item In decomposition $\hilbert[M]_\lambda = \bigoplus_{\gamma \in \Gamma_\lambda} \hilbert[M]_\gamma$, for at least one irrep $\gamma^\ast\in \Gamma_\lambda$, $\dim \hilbert[M]_{\gamma^\ast} \geq 3$. This follows immediately from \ref{fact:1}, which holds because $m \geq 4$ and $\lambda$ is not a 1D irrep of $\mathbb{S}_{m + 1}$.\\

\item The  one-parameter families
$\exp({\i t (\P_{ij}- \P_{kl})}): t \in \real$ for all $i \neq j, k \neq l\in\{1,\cdots, m+1\}$
are inside $\mathcal{W}^{(m + 1)}$, and according to \cref{lem-2-local}, the group generated by them acts irreducibly on $\mathcal{M}_\lambda$ for all $\lambda\in \Lambda_{m+1, d}$. To see why these families are inside $\mathcal{W}^{(m + 1)}$, note that according to \cref{lem-2-local}, $\exp({i t (\P_{ij}- \P_{kl})})$ is in the group $\mathcal{SV}^{(m)}$ defined on any 
$m$ out of $m+1$ qudits that contain qudits $i, j, k, l$. Then, together with \cref{incl}, this implies that on $m+1$ qudits, it is inside 
$\mathcal{W}^{(m + 1)}$. \\

\end{enumerate}

Therefore, both assumptions of \cref{thm:mulblocks} are satisfied for $\pi_\lambda(\mathcal{W}^{(m+1)})$, which implies $\pi_\lambda(\mathcal{W}^{(m+1)})=\SU(\mathcal{M}_\lambda)$. In summary, for all $\lambda\in\Lambda_{m+1, d}$, the subsystem universality holds, which means condition \textbf{A} of \cref{MainLemma} is satisfied. \\

\noindent $\bullet$ Condition \textbf{B}:
Consider an arbitrary pair of distinct irreps $\lambda, \lambda' \in \Lambda_{m + 1, d}$. If one or both of $\lambda$ and $\lambda'$ are 1D irreps of $\mathbb{S}_{m + 1}$, then condition \textbf{B} is satisfied for that pair. Therefore, without loss of generality, we assume they both have dimensions larger than one. Then, using 
\ref{fact:2} we know that the branching sets $\Gamma_\lambda \neq \Gamma_{\lambda'}$, with some differing irrep having dimension greater than two. Without loss of generality, let $\gamma \in \Gamma_\lambda$ with $\gamma \not\in \Gamma_{\lambda'}$ and $\dim \hilbert[M]_\gamma \geq 2$.
Consider the subgroup of $\mathcal{SV}^{(m)}$ isomorphic to $\SU(\hilbert[M]_\gamma)$ formed from $m$-qudit $\SU(d)$-invariant unitaruies that act trivially in all sectors except $\hilbert[M]_\gamma$. Then, for any unitary $V$ in this subgroup $V\otimes \mathbb{I}$ acts trivially in $\mathcal{M}_{\lambda'}$ and non-trivially in $\mathcal{M}_{\lambda}$. That is, 
$\pi_{\lambda'}(V\otimes \mathbb{I})= \mathbb{I}_{\lambda'}$ and $\pi_{\lambda}(V\otimes \mathbb{I})=\pi_{\gamma}(V)\oplus \mathbb{I}_\perp$, where $\mathbb{I}_\perp$ is the identity operator on the orthogonal complement of $\mathcal{M}_{\gamma}$ in $\mathcal{M}_{\lambda}$. We conclude that condition \textbf{B}' of \cref{MainLemma}, and thus, condition \textbf{B} of this lemma are also satisfied. \\

In conclusion, the subgroup $\mathcal{W}^{(m+1)}\subset \mathcal{SV}^{(m+1)}$ satisfies both assumptions \textbf{A} and \textbf{B} of \cref{MainLemma}, and therefore, this lemma implies that $\mathcal{W}^{(m+1)}= \mathcal{SV}^{(m+1)}$. This completes the proof of semi-universality of 3-qudit $\SU(d)$-invariant gates.

\section{Impossibility of achieving universality with (\texorpdfstring{$\bm{n - 1}$}{n - 1})-qudit gates}\label{sec:imposs}

We saw that semi-universality can be achieved with 3-qudit gates. The next natural question is what is the minimum locality of gates that is needed to achieve universality. More precisely, assuming we cannot use any ancilla qudits, what is the minimum $k$ for which $k$-qudit $\SU(d)$-invariant gates become universal on $n$ qudits, such that $\mathcal{V}^{(n)}_k=\mathcal{V}^{(n)}_n=\mathcal{V}^{(n)}$? 

For qubit systems with $\SU(2)$ symmetry, which corresponds to the special case of $d=2$ in the present paper, it was shown in \cite{Marvian2024Rotationally} that universality is achieved with $k = 2\lfloor n/2 \rfloor$. This means that 
there is an even-odd effect, which was shown in \cite{Marvian2024Rotationally} to be related to time-reversal symmetry. Specifically, when $n$ is odd, $(n - 1)$-local symmetric gates are universal, but they are not when $n$ is even. Here we show that, for $d \geq 3$, there is no such effect: without using ancillary qudits, $(n - 1)$-qudit $\SU(d)$-invariant gates are not universal. This proves \cref{prop:no}.

According to the general results of \cite{Marvian2022Restrict}, which apply to arbitrary symmetry groups, the semi-universality of $k$-qudit $\SU(d)$-invariant gates for $k\ge 3$ together with the fact $\SU(d)$ is a connected Lie group imply that   the difference between the dimensions of the Lie group of all $\SU(d)$-invariant unitaries on $n$ qudits, and the subgroup generated by $k$-qudit $\SU(d)$-invariant gates is
\begin{equation}\label{eq:centerdiff}
  \dim(\mathcal{V}_n^{(n)}) - \dim(\mathcal{V}_k^{(n)}) = \abs{\Lambda_{n, d}} - \abs{\Lambda_{k, d}}\ .
\end{equation}
Here, $|\Lambda_{n, d}|$ is the size of $\Lambda_{n, d}$, the set of inequivalent irreps of $\SU(d)$, or equivalently $\mathbb{S}_n$, that appear on $n$ qudits, which can be labeled by Young diagrams with $n$ boxes and $\leq d$ rows. In particular, in the case of $\SU(d)$ symmetry, the arguments of \cite{Marvian2022Restrict} imply that, because $\SU(d)$ is a connected Lie group, the center of the Lie algebra associated to $\mathcal{V}_k^{(n)}$ has dimension $\abs{\Lambda_{k, d}}$, whereas the center of the Lie algebra associated to $\mathcal{V}_n^{(n)}=\mathcal{V}^{(n)}$ has dimension $\abs{\Lambda_{n, d}}$, and because of semi-universality this is exactly the difference between dimensions of $\mathcal{V}_n^{(n)}$ and $\mathcal{V}_k^{(n)}$. In \cref{app:conj-proof}, we show that
\begin{restatable}[Strict monotonicity of the number of irreps, for $d\ge 3$]{lemma}{lemIrrepsSn}\label{lem:irrepsSn} 
  Let $|\Lambda_{k, d}|$ be the number of inequivalent irreps of $\SU(d)$, or equivalently $\mathbb{S}_k$, on $k$ qudits with the total Hilbert space $(\mathbb{C}^d)^{\otimes k}$. Then, for $d\geq 3$, $\abs{\Lambda_{k, d}}
  > \abs{\Lambda_{k-1, d}}$ for all $k$. 
\end{restatable}
In summary, for $d\ge 3$, we have $$\abs{\Lambda_{n, d}}= \abs{\Lambda_{k, d}}\ \ \ \Longleftrightarrow \ \ \ n=k\ .$$
Then, together with \cref{eq:centerdiff}, this proves the non-universality of $k$-qudit gates with $k<n$, and we have the following corollary,
\begin{corollary}\label{prop:no}
  Consider a system of $n$ qudits with $d \geq 3$. Then, $k$-local $\SU(d)$-invariant unitaries are universal, if and only if $k = n$. In other words, if $k < n$, then $\mathcal{V}_k^{(n)} \subsetneq \mathcal{V}_n^{(n)}$ is a proper subgroup.
\end{corollary}

We also note that while there is no simple general formula for $|\Lambda_{k, d}|$, it can be obtained \cite{erdos1941} using its generating function 
\begin{align}
  \sum_{k=0}^\infty |\Lambda_{k,d}| \ x^k = \prod_{l=1}^d \frac{1}{1-x^l}\ .
\end{align}
Using this method, in \cref{fig:partition-monotonicity} we plot $|\Lambda_{k, d}|$ as a function of $k$. We see that $|\Lambda_{k, d}|$ is strictly monotonic in $k$ for $d\geq 3$, while for $d=2$ it grows only every other step in $k$. 
\begin{figure}[htb]
  \centering
  \includegraphics[width=0.49\textwidth]{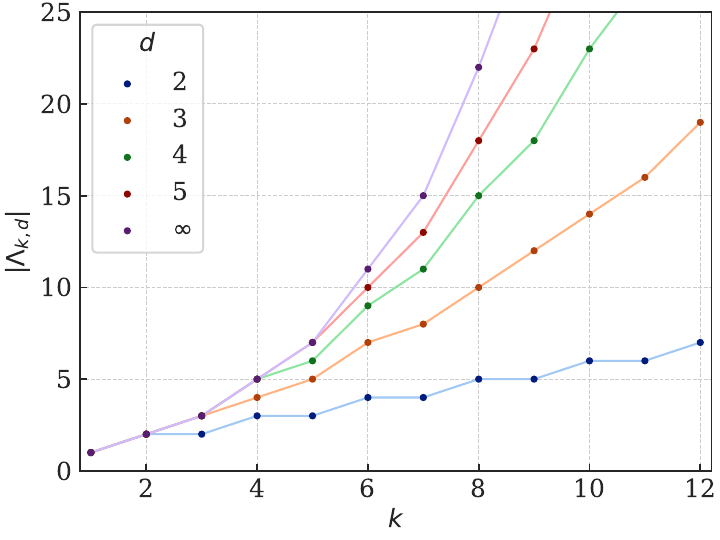}
  \caption{ The vertical axis is $|\Lambda_{k, d}|$,  the number of inequivalent irreps of $\SU(d)$ on $k$ qudits, which according to the results of \cite{Marvian2022Restrict}, is the the dimension of the center of the Lie algebra associated to $\mathcal{V}^{(n)}_k$, the group generated by $k$-qudit $\SU(d)$-invariant gates. \color{black} When $d\geq 3$, $|\Lambda_{k, d}|$ is strictly monotonic in $k$. When $d=2$, $|\Lambda_{k, d}|$ increases only every other step in $k$.}
  \label{fig:partition-monotonicity}
\end{figure}

Finally, we discuss the asymptotic behavior of $|\Lambda_{k,d}|$ in the regime $k\gg d$, which implies \cref{eq:asym-center} and explain the plot in \cref{fig:asym-rpn}. Recall that 
$|\Lambda_{k,d}|$ is the number of partitions of $k$ into $i \le d$ positive integers, i.e., $k_1\ge \cdots \ge k_i$ and $k=k_1+k_2+\cdots+ k_i$.

In the limit of large $k\gg d$, the number of such partitions for each $i\le d$ can be calculated through the number of the composition of $k$ into exactly $i$ parts, namely $\binom{k-1}{i-1}$. However, the order does not matter in integer partition, whereas it matters in integer composition, which should be corrected.  When $k\gg i$, for almost all partitions, $k_1,\cdots, k_i$ take distinct values, and we can correct it approximately by a factor of $1/i!$. This implies that\footnote{We note that this formula can also be obtained using Sylvester's denumerant formula \cite{alfonsin2005diophantine} or the asymptotic formula from \cite{erdos1941}. }
\begin{align}
  |\Lambda_{k,d}| \approx \sum_{i=1}^{d} \frac{1}{i!}\binom{k-1}{i-1} = \frac{k^{d-1}}{d!(d-1)!} + O(k^{d-2})\ .
\end{align}
Putting this into \cref{eq:centerdiff}, we arrive at
\begin{align}\nonumber
  \dim \mathcal{V}_k^{(n)} - \dim \mathcal{V}_3^{(n)} &= |\Lambda_{k, d}| - |\Lambda_{3, d}|= |\Lambda_{k, d}| - 3 \nonumber\\ &\approx \frac{k^{d-1}}{d!(d-1)!} + O(k^{d-2})\ ,\nonumber
\end{align}
where we have used the fact that $\Lambda_{3,d}$, given in \cref{Eq:Lambda_3}, has 3 elements. 
This proves \cref{eq:asym-center}.

\section{The use of ancilla qudits for achieving (semi-)universality}\label{sec:ancilla}

In this section, we consider the use of ancillae for achieving both universality and semi-universality and prove various results in this context, including those presented in \cref{tab:summary}. First, we consider achieving universality from semi-universality for arbitrary symmetry group $G$ with on-site representation, and prove \cref{prop:ufromsemi}. That is, we prove that universality can be achieved on $n$ qudits if semi-universality holds on $n+3$ qudits, and the 3 qudits are used as ancillae. {We note that in general, one might be able to achieve universality with fewer ancilla qudits (for instance, 1 qudit in the case of Abelian symmetries \cite{marvian2024theoryabelian}, and 2 qubits in the case of $\SU(2)$ symmetry with qubit systems \cite{marvian2022quditcircuit}).} 
Then, we focus on the case of 2-qudit gates with $\SU(d)$ symmetry. Even though 2-qudit gates are not semi-universal for $n>3$ qudits, by studying a system with 11 qudits with the total Hilbert space $(\mathbb{C}^d)^{\otimes 11}$, we show that 8 qudits can be used as ancillae to implement any $\SU(d)$-invariant unitary on the remaining 3 qudits. In the special case of $d=3$, this can be achieved with only 6 ancilla qutrits. 

Since the state of ancillae remains unchanged and uncorrelated with the rest of qudits, we can reuse it to realize arbitrary 3-qudit gates on a system with arbitrary $n>3$ qudits. 
We have previously shown in \cref{thm:semi-universality} that $3$-qudit gates are semi-universal on a system with arbitrary $n$. Combining these we conclude that

\begin{corollary}\label{cor:anc}
  Using 2-qudit $\SU(d)$-invariant unitaries, and 8 ancilla qudits we can realize any unitary in $\mathcal{SV}^{(n)}$, and using 11 ancilla qudits we can realize any unitary in $\mathcal{V}^{(n)}$. In the special case of $d=3$, these can be achieved with $6$ and $9$ ancilla qudits, respectively. 
\end{corollary}

\subsection{From semi-universality to universality with three ancilla qudits (Proof of Proposition \ref{prop:ufromsemi})}\label{sec:uni3}

Consider an arbitrary symmetry group $G$ with the on-site representation, as defined in \cref{prop:ufromsemi}. Consider 3 ancillary qudits, with the total Hilbert space 
\begin{equation}
  (\complex^d)^{\otimes 3} = \hilbert[Q]_{\ysub{3}} \oplus (\hilbert[Q]_{\ysub{2, 1}} \otimes \hilbert[M]_{\ysub{2, 1}}) \oplus \hilbert[Q]_{\ysub{1,1,1}}\ ,
\end{equation}
where we have decomposed this Hilbert space according to the irreps of the permutation group $\mathbb{S}_3$, and used the notation introduced in \cref{eq:schur-weyl}. Since it is permutationally-invariant, $u(g)^{\otimes 3}$ acts on $\hilbert[Q]_{\ysub{2, 1}} \otimes \hilbert[M]_{\ysub{2, 1}}$ as $u_{\ysub{2, 1}}(g) \otimes \ident_{\hilbert[M]_{\sysub{2, 1}}}$, where $u_{\ysub{2, 1}}$ is a possibly reducible representation of $G$ on $\hilbert[Q]_{\ysub{2, 1}}$. 

Let $Z_{\ysub{2, 1}}$ be a nonzero Hermitian operator on $\hilbert[M]_{\ysub{2, 1}}\cong \mathbb{C}^2$ with $\Tr Z_{\ysub{2, 1}} = 0$. Then, for any Hermitian $G$-invariant operator $H$ on $(\complex^d)^{\otimes n}$, we consider operator $\widetilde{H}$ on $n+3$ qudits, defined by 
\be
H\ \mapsto \ \widetilde{H}=H \otimes (\ident_{\hilbert[Q]_{\sysub{2, 1}}} \otimes Z_{\ysub{2, 1}})\ .
\ee
Then, $\widetilde{H}$ is also Hermitian, $G$-invariant, and in addition, it is centerless, because
\begin{equation}
  \begin{split}
    & \Tr \big(u(g)^{\otimes (n + 3)} \widetilde{H} \big) = \Tr (u(g)^{\otimes n} H) \Tr (u_{\ysub{2, 1}}(g)) \Tr(Z_{\ysub{2, 1}}) = 0,
  \end{split}
\end{equation}
for all $g\in G$, which follows from $\Tr Z_{\ysub{2, 1}} = 0$. We conclude that the family of unitaries $\exp(\i \widetilde{H} t): t\in \mathbb{R} $ are in the commutator subgroup of $G$-invariant unitaries on $n+3$ qudits,
and therefore they can be realized by any set of gates that are semi-universal on $n+3$ qudits. For instance, in the case of $\SU(d)$ symmetry discussed in the previous sections, $\exp(\i \widetilde{H} t)$ is in $\mathcal{SV}^{(n+3)}$, which means it is realizable with 3-qudit $\SU(d)$-invariant unitaries.

In general, $\widetilde{H}$ is an entangling Hamiltonian. However, by preparing the 3 ancilla qudit in an eigenstate $\ket{\eta} \in (\complex^d)^{\otimes 3}$ of
$\ident_{Q_{\sysub{2, 1}}} \otimes Z_{\ysub{2,1}}$, they remain unentangled with the system. That is, 
\begin{equation}
  \exp(\i \widetilde{H} t) (\ket{\psi} \otimes \ket{\eta}) = (\exp(\i{H}\alpha t)\ket{\psi}) \otimes \ket{\eta}\ ,
\end{equation}
where $\alpha$ is the eigenvalue of $\ident_{Q_{\sysub{2, 1}}} \otimes Z_{\ysub{2,1}}$ for eigenvector $|\eta\rangle$.

Some examples of choices of $\ket{\eta}$ and $Z_{\ysub{2, 1}}$ include
\begin{enumerate}[$\bullet$]
\item $\ket{\eta_1} = \frac{1}{2 \sqrt{3}}(2 \ket{100} + (\sqrt{3} - 1) \ket{010} - (1 + \sqrt{3}) \ket{001})$ and $\ident_{\hilbert[Q]_{\sysub{2, 1}}} \otimes Z_{\ysub{2, 1}} = \frac{1}{\sqrt{3}}(\P_{12} - \P_{13})$.
\item $\ket{\eta_2}= \frac{1}{\sqrt{2}}(\ket{01} - \ket{10}) \otimes \ket{0}$ and $\ident_{\hilbert[Q]_{\sysub{2, 1}}} \otimes Z_{\ysub{2, 1}} = \frac{1}{3}(-2\P_{12} + 2\P_{23} + \P_{(123)} - \P_{(132)})$.
\end{enumerate}
Here, $Z_{\ysub{2, 1}}$ is normalized so that its eigenvalues are $\pm 1$, and 
the given state is an eigenvector with eigenvalue $\alpha = 1$. 
Note that the above particular linear combination of permutations is restricted to irrep $\ydiag{2, 1}$, such that
\be
\big(\ident_{\hilbert[Q]_{\sysub{2, 1}}} \otimes Z_{\ysub{2, 1}}\big)\Pi_{\lambda}=0 \ :\ \ \ \  \lambda=\ydiag{3}\ ,~ \ydiag{1,1,1}\ ,
\ee 
which can be seen using the fact that $\P(\sigma)\Pi_{\ysub{3}}=\Pi_{\ysub{3}}$ and $\P(\sigma)\Pi_{\ysub{1,1,1}}=\sgn(\sigma)\Pi_{\ysub{1,1,1}}$.

\subsection{(Semi-)universality with 2-qudit \texorpdfstring{\bm{{$\SU(d)$}}}{SU(d)}-invariant gates and ancillae (Proof of Corollary \ref{cor:anc})}\label{sec:3ancilla}

\Cref{thm:semi-universality} shows that 3-qudit $\SU(d)$-invariant gates are semi-universal on arbitrary $n$ qudits. In this section we show how general 3-qudit $\SU(d)$-invariant gates can be realized with 2-qudit gates, provided that one can use 8 (or 6, in the special case of $d=3$) ancilla qudits. Since the state of the ancillae remains unchanged, we can reuse it arbitrarily many times to implement all 3-qudit gates in the circuit. Therefore, this together with \cref{thm:semi-universality}, proves \cref{cor:anc}.

As we proved in \cref{sec:2on3}, on $n=3$ qudits, 2-qudit gates are semi-universal, i.e., $\mathcal{SV}^{(3)}\subset \mathcal{V}^{(3)}_2$. Therefore, in this case, the only constraints that the locality of gates imposes on the realizable unitaries are on the relative phases between sectors with different irreps of symmetry, namely type $\textbf{I}$ constraints. Hence, to achieve universality it suffices to amend $\mathcal{SV}^{(3)}$ with 
relative phases between different sectors, i.e., unitaries in the form 
\be\label{phases}
\sum_{\lambda\in \Lambda_3} \e^{\i\theta_\lambda}\Pi_\lambda\ :\ \ \  \theta_\lambda\in [0,2\pi)\ , 
\ee
where $\Lambda_3=\set[\big]{\sydiag{3} \, , \; \sydiag{2, 1} \, , \; \sydiag{1, 1, 1}}$ (we note that a 2D Lie subgroup of this group is already included in $\mathcal{V}_2^{(3)}$).
This may sound similar to the problem we studied in the previous section, where we used 3 ancilla qudits to elevate semi-universality to universality. However, there is an important difference: to apply that technique we need semi-universality on $n+3$ qudits, i.e., the main qudits and the ancillae. Therefore, in the context of the problem at hand, this will require semi-universality on 6 qudits. However, as we discussed before, for $n>3$ qudits 2-qudit gates are not semi-universal. Therefore, the technique of the previous section is not applicable here, and we need to develop a new scheme that works even in the absence of semi-universality.

We have found a solution to this problem that uses 8 ancilla qudits. This solution is based on the following fact about systems with 11 qudits. Recall that the wedge product of states is defined as
\begin{equation}
  \ket{\psi_1} \wedge \dotsb \wedge \ket{\psi_m} := \frac{1}{\sqrt{m!}} \sum_{\sigma \in \mathbb{S}_m} \sgn(\sigma) \P(\sigma) (\ket{\psi_1} \otimes \dots \otimes \ket{\psi_m}).
\end{equation}

\begin{lemma}[Centerless Hamiltonians on 11 qudits]\label{lem:ancillauni} For any centerless $\SU(d)$-invariant Hamiltonian $H$ on $(\mathbb{C}^d)^{\otimes 11}$ ({i.e., a Hamiltonian satisfying $\Tr(H u(g)^{\otimes 11})=0$ for all $g\in\SU(d)$}) there exists a Hamiltonian $\widetilde{H}$ that is realizable with 2-qudit $\SU(d)$-invariant Hamiltonians, i.e., $\exp({\i t\widetilde{H}})\in\mathcal{V}_2: t\in{\mathbb{R}}$, which satisfies 
  \begin{align}\label{rt}
    \exp{(\i t\widetilde{H})} (|\psi\rangle \otimes |\eta\rangle) =\exp{(\i t {H})}(|\psi\rangle \otimes |\eta\rangle)\ ,
  \end{align}
  for all $|\psi\>\in(\mathbb{C}^d)^{\otimes 3}$, 
  where 
  \begin{align}\label{main:eq:8anc}
    |\eta\rangle=(\ket{0} \wedge \ket{1} \wedge \ket{2} \wedge \ket{3})^{\otimes 2}\in (\mathbb{C}^d)^{\otimes 8}\ , 
  \end{align}
  is an 8-qudit state. Furthermore, in the special case of $d=3$, the same fact holds for any centerless Hamiltonian on 9 qudits, and state \begin{align}\label{eq:6anc}
    |\eta'\rangle=(\ket{0} \wedge \ket{1})^{\otimes 2} \otimes |00\>\in (\mathbb{C}^d)^{\otimes 6} .
  \end{align}
\end{lemma}
This lemma can be shown directly by studying the Lie algebra generated by SWAPs on 11 qudits, or, as we do in \cref{sec:ancillapr}, by applying the seminal result of Marin \cite{marin2007algebre}, which characterizes the Lie algebra generated by transpositions (SWAPs) in terms of its simple factors.\footnote{Note that this is the only result in the current paper depending on Marin's results in \cite{marin2007algebre}.}

Next, we apply this lemma to a system containing 3 qudits plus 8 ancilla qudits, and show how we can realize unitaries in the form of \cref{phases}. 
More precisely, we show that Hamiltonian 
$\Pi_\lambda^{123}$ for all $\lambda\in \Lambda_3$ can be realized with 8 ancilla qudits.

Since $\Pi_\lambda^{123}$ is not centerless, first we construct a centerless Hamiltonian, namely we consider $ H = \Pi_\lambda^{123} - \Pi_\lambda^{456}$, where we have suppressed tensor product with the identity operators on the rest of qudits. This Hamiltonian is clearly centerless. In fact, any operator of the form $A - P_\sigma AP_\sigma^{-1}$ is centerless, because
\begin{align}
  \Tr\big(\Pi (A - P_\sigma AP_\sigma^{-1})\big) = \Tr\big(\Pi (A - A)\big) = 0,
\end{align}
where in the last step we used the cyclic property of the trace and $\Pi$ is an arbitrary element in the center. Therefore, according to the above lemma, there exists an $\SU(d)$-invariant Hamiltonian $\widetilde{H}$, such that \cref{rt} holds for all $ t\in{\mathbb{R}}$ and all $|\psi\rangle\in(\mathbb{C}^d)^{\otimes 3}$. 

Now for state $|\eta\rangle$ in \cref{main:eq:8anc} the reduced state of qudits 456 is restricted to the totally anti-symmetric subspace, corresponding to $\lambda=\sydiag{1, 1, 1}$, which means
\bes\label{eigen}
\begin{align}
  \Pi^{456}_{\lambda} |\eta\rangle &= 0 \ \ &&: \lambda=\ydiag{3} , \ydiag{2,1}\ ,\\ 
  \Pi^{456}_{\lambda} |\eta\rangle &= |\eta\rangle \ \ &&: \lambda=\ydiag{1,1,1}\ ,
\end{align}
\ees
where we have suppressed the tensor product with the identity operator on the rest of the qudits. 
Therefore, for Hamiltonian 
$ H = \Pi_\lambda^{123} - \Pi_\lambda^{456}$, \cref{rt} becomes
\begin{align}
  \e^{\i t\widetilde{H}} (|\psi\rangle\otimes |\eta\rangle)=
  \begin{cases}
    (\e^{\i t\Pi_\lambda^{123}} |\psi\rangle)\otimes |\eta\rangle &\lambda=\ydiag{3}, \ydiag{2, 1} \\
    (\e^{\i t(\Pi_\lambda^{123}-\mathbb{I})} |\psi\rangle)\otimes |\eta\rangle & \lambda=\ydiag{1, 1, 1}
  \end{cases}
\end{align}
Note that in the second case, we get an extra global phase, which is not physically relevant (e.g., it can be removed by shifting $\widetilde{H}$ with a multiple of the identity operator). 

We conclude that for all $\lambda\in \Lambda_3$ we can implement the Hamiltonian $\Pi_\lambda$ on 3 qudits, using 8 qudits as ancilla and 2-qudit $\SU(d)$-invariant gates. Since the state of ancilla remains unchanged we can reuse it to implement all unitaries in the form of \cref{phases}. 

Now consider the special case of $d=3$. Although in this case \cref{eigen} does not hold, because both states $(|0\rangle \wedge |1\rangle)\otimes |0\rangle$ and $(|0\rangle \wedge |1\rangle)\otimes |1\rangle$ live in $\mathcal{Q}_{\ysub{2, 1}} \otimes \mathcal{M}_{\ysub{2, 1}}$, we have
\bes\label{eigen2}
\begin{align} 
  \Pi^{456}_{\lambda} |\eta'\rangle &= 0 \ \ &&: \lambda=\ydiag{3} , \ydiag{1,1,1}, \\
  \Pi^{456}_{\lambda} |\eta'\rangle &= |\eta'\rangle \ \ &&: \lambda=\ydiag{2,1} .
\end{align}
\ees
Therefore, in this case for Hamiltonian 
$ H = \Pi_\lambda^{123} - \Pi_\lambda^{456}$ \cref{rt} simplifies to 
\begin{align}
  \e^{\i t\widetilde{H}} (|\psi\rangle\otimes |\eta'\rangle)=
  \begin{cases}
    (\e^{\i t(\Pi_\lambda^{123})} |\psi\rangle)\otimes |\eta'\rangle & \lambda=\ydiag{3}, \ydiag{1, 1, 1} \\
    (\e^{\i t(\Pi_\lambda^{123}-\mathbb{I})} |\psi\rangle)\otimes |\eta'\rangle &\lambda= \ydiag{2, 1} ,
  \end{cases}
\end{align}
and therefore the same results can be established with $6$ qudit state $|\eta'\rangle$.

\section{Statistical properties of random \texorpdfstring{$\SU(d)$}{SU(d)}-invariant circuits:  \texorpdfstring{$t$}{t}-designs }\label{Sec:design}

In this section, we briefly discuss the statistical properties of random 3-qudit circuits and explain a corollary of semi-universality of 3-qudit gates. 

First, recall that the no-go theorem of \cite{Marvian2022Restrict} puts strong constraints on the statistical properties of random symmetric circuits for general symmetries. According to the general results of \cite{Marvian2022Restrict}, for any symmetry group $G$, the group of unitaries generated by $k$-local $G$-invariant unitaries is a compact connected Lie group, and therefore it has a unique notion of invariant (Haar) measure. Furthermore, when the symmetry group is a continuous group, such as $\SU(d)$, \cite{Marvian2022Restrict} shows that for any fixed $k$, the dimension of this compact Lie group is strictly less than the compact Lie group of all $G$-invariant unitaries on $n$ qudits. Clearly, the uniform distribution over a compact manifold cannot be fully mimicked by distributions restricted to compact submanifolds with lower dimensions.  This can be more formally stated in terms of moments of the distributions.

\begin{corollary}[corollary of \cite{Marvian2022Restrict}]\label{cor45}
Consider the uniform distribution over the compact group $\mathcal{V}^G_k$ of $n$-qudit unitaries realized with $k$-qudit $G$-invariant gates. Then,
all moments of this distribution are equal to the corresponding moments of the uniform distribution over the group of all $G$-invariant unitaries, such that 
\be\label{design}
\mathbb{E}_{V\in\V^G}[V^{\otimes t}\otimes {V^\ast}^{\otimes t}]=\mathbb{E}_{V\in\V_k^G}[V^{\otimes t}\otimes {V^\ast}^{\otimes t}]\ ,
\ee
for all integer $t$, only if 
$|\text{Irreps}_G(k)|= |\text{Irreps}_G(n)|$, where $|\text{Irreps}_G(k)|$ is the number of inequivalent irreps of group $G$ appearing on $k$ qudits, i.e., in the representation $g\mapsto u(g)^{\otimes k}: g\in G$.  Furthermore, assuming semi-universality holds for $\mathcal{V}_k^G$,  and the group $G$ is connected, then 
this condition becomes necessary and sufficient.  

In the case of continuous symmetries such as $\SU(d)$, there is no finite $k$ such that $|\text{Irreps}_G(k)|= |\text{Irreps}_G(n)|$ for all $n> k$. Therefore, for any fixed $k$,  the above equation cannot hold for arbitrarily large $t$ and $n$. 
\end{corollary}

For instance, in the case of $\SU(d)$ symmetry discussed in this paper, the strict monotonicity of $|\Lambda_{n,d}|$ in \cref{lem:irrepsSn}, implies that for $d\ge 3$, unless $k=n$,  certain moments of the uniform distribution over $\mathcal{V}^{(n)}_{k}$ will be different from the corresponding moment for the Haar distribution over the group of all $\SU(d)$-invariant unitaries.

However, it is still possible that some moments of these distributions match. This can be formulated based on the notion of $t$-designs \cite{harrow2009random, emerson2005convergence, brandao2016efficient, brandao2016local}. We say the Haar distribution over the group $\mathcal{V}^{(n)}_{k}$ is a $t$-design for the Haar distribution over the group of all $\SU(d)$-invariant unitaries, if \cref{design} holds.

Our previous result in \cite{marvian2022quditcircuit} reveals that in the case of circuits with 2-qudit $\SU(d)$-invariant gates, for $d\ge 3$ this equation holds for $t=1$, but not $t=2$, which means the distribution generated by such random circuits converges to a 1-design but not a 2-design for the Haar distribution.  Indeed,  as it was noted in \cite{marvian2022quditcircuit}, the results of \cite{robert2015squares,zimboras2015symmetry} imply that whenever semi-universality does not hold then the distribution generated by the random circuits cannot be a 2-design for the Haar distribution over all symmetric unitaries. More recently, Ref.~\cite{Li:2023mek} shows that the distribution generated by circuits with 4-qudit $\SU(d)$-invariant gates is a $t$-design with $t$ quadratic in $n$.

In this paper, we find that a similar quadratic scaling can be achieved with 3-qudit gates in the case of $\SU(d)$ with $d\ge 3$, and 2-qudit gates in the case of $\SU(2)$.
\begin{proposition}[Quadratic scaling of $t$-designs with 3-  or  2-qudit gates]\label{prop:design}
  For systems with $n\geq 9$ qudits with $d< n-1$, the Haar distribution over the group generated by 3-qudit $\SU(d)$-invariant gates, $\V_3^{(n)}$ 
  is an exact $t$-design for the Haar distribution over the group of all $\SU(d)$-invariant unitaries $\V^{(n)}$ with $t < \frac{1}{2}n(n-3)$. In the case of $d=2$, i.e., qubits with $\SU(2)$ symmetry, $\V_3^{(n)}=\V_2^{(n)}$, which means the same quadratic scaling can be obtained with 2-qudit gates.
\end{proposition}
As we explain in Appendix \ref{App:tdesign}, this proposition follows from the 
semi-universality of 3-qudit gates, together with the following fact
\be
\pi_{\mu_0,\mu_1}(\mathcal{V}^{(n)}_3)=\pi_{\mu_0,\mu_1}(\mathcal{V}^{(n)})\cong  \U(\mathcal{M}_{\mu_0})\times \U(\mathcal{M}_{\mu_1})  \ ,
\ee
where $\mu_0$ and $\mu_1$ are, respectively the symmetric and the ($n-1$)-dimensional standard irrep of $\mathbb{S}_n$. In words, this means that when restricted to these two sectors, the group $\mathcal{V}^{(n)}_3$ and $\mathcal{V}^{(n)}$ are identical, which in turn implies that the design properties of $\mathcal{V}^{(n)}_3$ is determined by the smallest dimension of $\mathcal{M}_\lambda$, for $\lambda\in\Lambda_{n,d}$ other than $\mu_0,\mu_1$. This dimension is equal to $n(n-3)/2$, when $n\ge 9$ and $d<n-1$ \cite{rasala1977minimal}. We also use the standard techniques in
the context of $t$-designs \cite{harrow2009random, emerson2005convergence, brandao2016efficient, brandao2016local}, which are used previously,  e.g., in Ref.~\cite{Li:2023mek} in the context of $\SU(d)$-invariant circuits, and in \cite{hearth2023unitary} in the context of $\U(1)$-invariant circuits.

\section{From subsystem universality to semi-universality}\label{sec:tools}

In this section, we introduce 
a characterization of subsystem universality and various other techniques that are generally useful for understanding semi-universality (and its failure) for subgroups of $G$-invariant unitaries, for general symmetries and representations. We start by presenting a generalization of \cref{MainLemma}, namely \cref{Mainlemma_gen}, which is applicable even when condition 
$\textbf{B}$ of \cref{MainLemma}, i.e., pairwise independence, does not hold. Then, we show how both lemmas follow from a combination of Goursat's and Serre's lemmas (\cref{lem:SUgoursat} and \cref{lem:serre}). The ideas and techniques discussed in this section and in \cref{sec:proof} are more broadly useful in the context of quantum computing and control theory.

\subsection{Subdirect products}

Before going through this section, we recall the concept of subdirect products, which is useful for understanding semi-universality and also appears in the statements of Goursat's and Serre's lemmas.

\begin{definition}[Subdirect product]\label{def:subdirect}
  A \emph{subdirect product} of the groups $G_1, \dots, G_r$ is a subgroup of the direct product, $H \subseteq G_1 \times \dots \times G_r$, such that $\pi_i(H) = G_i$ for each factor, $i = 1, \dots, r$, where $\pi_i : G_1 \times \dots \times G_r \to G_i$ is the projection homomorphism.
\end{definition}
It is useful to consider different subdirect products as exhibiting different types of \emph{correlations}. In particular, if for any pair $i$ and $j$ of groups $G_i$ and $G_j$, the projection of $H$ to that pair is not surjective, i.e., $\pi_{ij}(H) \neq G_i \times G_j$, we say they are \emph{correlated} in $H$, and otherwise, we say they are \emph{independent} in $H$. This can be generalized to more factors, and in general, the correlations can be quite complicated (see, e.g., the example in \cref{sec:nec-suf} of a subdirect product of three $\U(1)$ subgroups which are pairwise independent but still exhibit a tripartite correlation).

As we explain below, the notion of subdirect products appears naturally in the context of $G$-invariant unitaries.

\subsection{Subsystem universality (condition \textbf{A})}

Consider any subgroup $\mathcal{T}\subset \mathcal{V}^G$ of $G$-invariant unitaries satisfying condition $\textbf{A}$ in \cref{MainLemma}, namely subsystem universality, which means that for any irrep $\lambda\in\Lambda$ the projection of $\mathcal{T}$ to $\mathcal{M}_{\lambda}$ contains $\SU(\mathcal{M}_{\lambda})$, i.e.
\be\label{qtr}
 \pi_{\lambda}(\mathcal{T})=\big\{\pi_{\lambda}(V): V\in \mathcal{T}\big\} \supseteq \SU(\mathcal{M}_{\lambda})\ .
\ee
Then under this assumption, the commutator subgroup of $\mathcal{T}$, denoted as $[\mathcal{T},\mathcal{T}]$, is a subdirect product of the groups $\SU(\mathcal{M}_\lambda): \lambda\in\Lambda$. Furthermore, semi-universality means that this subdirect product is indeed the Cartesian product $\mathcal{SV}^G \cong \prod_{\lambda\in\Lambda} \SU(\mathcal{M}_\lambda)$. 

It is also worth noting that subsystem universality implies that any unitary $V\in\mathcal{T}$ can be decomposed as the product of an element of $[\mathcal{T},\mathcal{T}]$ and a unitary in the subgroup
of relative phases
\be\label{eq:phases}
\mathcal{P}=\{\sum_{\lambda\in\Lambda} \e^{\i\theta_\lambda} \Pi_\lambda : \theta_\lambda\in[0,2\pi) \}\ ,
\ee
which is the center of $\mathcal{V}^G$. More precisely,

\begin{lemma}\label{lem:phases}
  A subgroup of $G$-invariant unitaries $\mathcal{T}\subset \mathcal{V}^G$ contains $\mathcal{SV}^G$ if, and only if, its commutator subgroup $[\mathcal{T}, \mathcal{T}]=\mathcal{SV}^G$. More generally, if $\mathcal{T}$ satisfies the subsystem universality condition in \cref{qtr} for all $\lambda\in\Lambda$, then any element of $\mathcal{T}$ can be decomposed as a product of an element of $[\mathcal{T},\mathcal{T}]$ and a unitary in the group of relative phases $\mathcal{P}$. That is, $\mathcal{T} \subseteq \mathcal{P} [\mathcal{T},\mathcal{T}]$. 
\end{lemma}
\begin{proof}
    Note that $\mathcal{V}^G = \mathcal{SV}^G \mathcal{P}$. If $\mathcal{SV}^G \subseteq \mathcal{T}$ then also $\mathcal{V}^G = \mathcal{P} \mathcal{T}$. Taking the commutator subgroup, since $\mathcal{P}$ is the center, i.e. commutes with everything else, it follows that $\mathcal{SV}^G = [\mathcal{T}, \mathcal{T}]$. The proof of the second part is presented in \cref{sec:mainproof}, by applying \cref{Mainlemma_gen}.
\end{proof}
Therefore, in the following, we often consider the properties of the commutator subgroup $[\mathcal{T}, \mathcal{T}]$ and, assuming subsystem universality holds, interpret it as a subdirect product.

\subsection{ A characterization of subsystem universality }\label{sec:fail}

Here, we present a useful characterization of subgroups of $G$-invariant unitaries $\mathcal{V}^G$ that satisfy the subsystem-universality condition in all charge sectors, i.e., condition $\textbf{A}$ of \cref{MainLemma}, but do not necessarily satisfy the condition 
$\textbf{B}$, pairwise independence (According to the numbering of constraints in \cref{sec:nec-suf}, this means that except type \textbf{I}, all constraints are of type \textbf{IV}, i.e., correlations between different sectors).

\begin{lemma}[Extension of \cref{MainLemma}]\label{Mainlemma_gen}
  Consider a subgroup $\mathcal{T}\subseteq \mathcal{V}^G$ of {$G$-invariant unitaries} satisfying condition \textbf{A} in \cref{MainLemma}. That is, for any irrep $\lambda\in\Lambda$ the projection of $\mathcal{T}$ to $\mathcal{M}_{\lambda}$ contains $\SU(\mathcal{M}_{\lambda})$, as stated in \cref{qtr}. Then, the set of irreps $\Lambda$ is partitioned into non-overlapping subsets $\Delta_1, \cdots, \Delta_s$, such that for all pairs $\lambda, \lambda'$ that belong to the same subset $\Delta_r$, $\dim (\mathcal{M}_\lambda)=\dim (\mathcal{M}_{\lambda'})\eqqcolon m_{r}$, and that the commutator subgroup of $\mathcal{T}$ is isomorphic to
  \be
  [\mathcal{T},\mathcal{T}]\cong\prod_{r=1}^s \SU(m_r) \ .
  \ee
  Furthermore, for any pair of irreps $\lambda, \lambda'\in \Delta_r$, and fixed orthonormal bases for $\hilbert[M]_{\lambda}$ and $\hilbert[M]_{\lambda'}$, there exists an $m_r \times m_r$ unitary matrix $W\in\SU(m_r)$, such that for all $V\in[\mathcal{T},\mathcal{T}]$ one of the followings holds
  \begin{subequations}
    \begin{align}
    [\pi_{\lambda'}(V)]&= W[\pi_{\lambda}(V)] W^\dag\ \ , \text{or} \\
    [\pi_{\lambda'}(V)]&= W[\pi_{\lambda}(V)]^\ast W^\dag\ ,
  \end{align}
  \end{subequations}
  where $[\pi_{\lambda}(V)]^\ast$ is the complex conjugate of $\pi_{\lambda}(V)$ relative to this basis.
\end{lemma}

In words, this lemma means that 
for charge sectors that belong to the same part $\Delta$, the unitary realized in one sector uniquely determines the unitaries realized in all the other sectors. Namely, they are all equal, up 
to a change of basis, a possible complex conjugation, and relative phases between sectors (i.e. an element of the group $\mathcal{P}$ defined in \cref{eq:phases}).  Note that, as stated in \cref{lem:phases}, when condition $\textbf{A}$ holds for all sectors, the commutator subgroup $[\mathcal{T},\mathcal{T}]$ determines $\mathcal{T}$, up to additional freedoms on the relative phases between different charge sectors. Therefore, this lemma characterizes group $\mathcal{T}$, up to this freedom.

We next mention an example that has been already discussed in detail in \cref{sec:2on4}.

\subsection*{Example: Revisiting group $\mathcal{V}^{(4)}_2$}
In the special case of $G=\SU(d)$ symmetry and $n=4$ qudits, in 
\cref{sec:2on4} we  characterized the group $\mathcal{V}^{(4)}_2$ generated 
by 2-qudit $\SU(d)$-invariant gates. 
In the context of the above lemma, $\mathcal{V}^{(4)}_2$ corresponds to the group $\mathcal{T}$, which is a subgroup of $\mathcal{V}^{(4)}$ the group of all $\SU(d)$-invariant unitaries. Then, in the language of the above lemma, we found that relative to this group, the set of irreps $\Lambda_4$ is partitioned into 3 sets, namely 
\be
\Delta_1=\big\{\ydiag{2, 2} \big\} \ ,\ \Delta_2=\big\{\ydiag{3, 1},\ydiag{2, 1, 1} \big\}\ ,\ \Delta_3=\big\{\ydiag{4} , \ydiag{1, 1, 1, 1} \big\} \ ,
\ee
with corresponding dimensions $m_1=2$, $m_2=3$, and $m_3=1$. Furthermore, we showed while the commutator subgroup of all $\SU(d)$-invariant unitaries, denoted as $\mathcal{SV}^{(4)}$, is isomorphic to $\SU(2)\times \SU(3)\times \SU(3)$, the commutator subgroup of $\mathcal{V}^{(4)}_2$, denoted as $\mathcal{SV}^{(4)}_2$ is isomorphic to
\be
\SU(m_1)\times \SU(m_2)\times \SU(m_3)= \SU(2)\times \SU(3)\ . 
\ee
Moreover, in \cref{const} we saw that for any $V\in\mathcal{V}^{(4)}_2$, the unitaries realized in $\mathcal{M}_{\ysub{3, 1}}$ and $\mathcal{M}_{\ysub{2, 1,1}}$ are equal, up to a complex conjugation and a change of basis, as described by \cref{Mainlemma_gen}.\\

Next, we explain the tools that are needed to prove \cref{MainLemma,Mainlemma_gen}.

\subsection{Two useful lemmas}

\subsubsection{Goursat's lemma: A characterization of pairwise correlations}\label{sec:goursat}

Goursat's lemma \cite{lang2005algebra} classifies all possible subdirect products $H \subseteq G_1 \times G_2$ of arbitrary groups $G_1$ and $G_2$.
We postpone the general form of this lemma, which applies to arbitrary groups $G_1$ and $G_2$ to \cref{lem:agoursat} in \cref{app:Goursat}. Instead, here we discuss its specialized form applied to special unitary groups, which are the most relevant groups for applications in quantum computing and control.

\begin{lemma}[Goursat's lemma for special unitary groups]\label{lem:SUgoursat}
  Let $l, l' \geq 2$ and let $H \subseteq \SU(l) \times \SU(l')$ be a subdirect product. There are two possibilities:
  \begin{enumerate}[(i)]
  \item $H = \SU(l) \times \SU(l')$.
  \item $l = l'$ and $H \cong \SU(l) \times \mathbb{Z}_q$, where $q$ divides $l$.
  \end{enumerate}
  Furthermore, in the second case, there exists an isomorphism $\Phi : \SU(l) \to \SU(l)$ such that
  \begin{equation}\label{eq:corr}
    H = \set{(U, \e^{\i \theta} \Phi(U)) \given U \in \SU(l) \text{ and } \e^{\i \theta} \in \mathbb{Z}_q \subset \SU(l)}.
  \end{equation}
\end{lemma}

\begin{remark}\label{rem:SUiso}
  When interpreted as explicit $l\times l$ determinant-one unitary matrices, then the above isomorphism in case (ii) must be in one of the following forms:
  \bes\label{eq:isos}
  \begin{align}
    \Phi(U) &= W U W^\dagger\\
    \Phi(U) &= W U^\ast W^\dagger,
  \end{align}
  \ees
  where $W$ is a unitary matrix and $U^\ast$ is the component-wise complex conjugate of matrix $U$. For $l>2$, these two possibilities correspond to two inequivalent representations of the group $\SU(l)$,\footnote{For $l = 2$, these are not distinct possibilities since the defining representation is self-dual: with Pauli $Y$ defined in the chosen basis, we have $U^\ast = Y U Y^\dagger$ for all $U \in \SU(2)$.} and these are the only two non-trivial representations of dimension $l$ \cite{fulton2013representation}.\footnote{Note that this fact follows from ``automatic continuity''. See \cref{app:auto} for further discussion.}
\end{remark}

Note that, in the second case of \cref{lem:SUgoursat}, upon taking the commutator subgroup, the discrete phase vanishes, i.e. $[H, H] \cong \SU(l)$. In fact, this holds even if we consider $H$ as a subgroup of $\U(l) \times \U(l)$. Goursat's lemma and the following corollary are proved in \cref{app:Goursat}.

\begin{corollary}\label{cor:Ugoursat}
    Suppose that $H \subseteq \U(l) \times \U(l')$ is a subgroup with $[H, H] \subseteq \SU(l) \times \SU(l')$ a subdirect product. Then either $[H, H] \cong \SU(l) \times \SU(l)$, or $l = l'$ and there is an isomorphism $\Phi : \SU(l) \to \SU(l)$ such that $[H, H] = \set{(U, \Phi(U)) \given U \in \SU(l)} \cong \SU(l)$. Furthermore, in the second case every element of $H$ is of the form $(\e^{\i \theta} U, \e^{\i \phi} \Phi(U))$, for some $\theta,\phi\in[0,2\pi)$.
\end{corollary}

\subsubsection{Serre's lemma: Pairwise independence implies full independence}\label{sec:serre}

Serre's lemma characterizes when the subdirect product of perfect groups is actually a direct product, and roughly states that, if perfect groups are pairwise independent, then they are fully independent.\footnote{Serre's lemma for finite groups was first stated in a paper by Kenneth Ribet \cite{Ribet1974adic}, where the proof is attributed to Jean-Pierre Serre. Later Terence Tao restated the lemma in a slightly different way and named it Serre's lemma \cite{tao2021goursat}. Our statement and proof follow Ribet \cite{Ribet1974adic}.} Recall that a group $G$ is called \emph{perfect} if it is equal to its commutator subgroup, i.e., $
[G, G] = G$.
For example, the unitary group $\U(d)$ is not perfect, but the special unitary group $\SU(d)$ is perfect (in fact, all finite-dimensional semi-simple Lie groups are perfect).

\begin{lemma}[Serre's lemma \cite{Ribet1974adic}]\label{lem:serre}
  Let $H \subseteq G_1 \times \dots \times G_r$ be a subgroup such that $\pi_{ij}(H) = G_i \times G_j$ for all pairs $1 \leq i < j \leq r$, where $\pi_{ij} : G_1 \times \dots \times G_r \to G_i \times G_j$ is the projection homomorphism. If each $G_i$ is perfect, then $H = G_1 \times \dots \times G_r$.
\end{lemma}
We present the proof of this lemma in \cref{app:serre}.

\subsection{Characterizations of semi-universality and subsystem universality: Proofs of Lemmas \ref{MainLemma} and \ref{Mainlemma_gen} }\label{sec:mainproof}

Finally, we show how Goursat's and Serre's lemmas can be applied to the case of $G$-invariant unitaries, and prove \cref{MainLemma} and \cref{Mainlemma_gen}. We start by \cref{MainLemma}.

Given that $\pi_{\Lambda}(\mathcal{SV}^G) = \prod_{\lambda} \SU(\hilbert[M]_\lambda)$, it is clear that if $\mathcal{T} \supseteq \mathcal{SV}^G$, then $\mathcal{T}$ satisfies both conditions \textbf{A} and \textbf{B} (see \cref{sec:nec-suf}). Therefore, to prove the lemma, it suffices to prove the converse direction.

First, we focus only on a pair of irreps $\lambda_1, \lambda_2\in\Lambda$, and discuss the implications of conditions $\textbf{A}$ and $\textbf{B}$ for this pair. Suppose the subgroup $\mathcal{T}\subset\mathcal{V}^G$ satisfies condition $\textbf{A}$ for $\lambda_1,\lambda_2\in\Lambda$. That is, \cref{qtr} holds for $\lambda=\lambda_1,\lambda_2$.
Then, case (i) of Goursat's lemma immediately implies that if $\dim \mathcal{M}_{\lambda_1}\neq \dim \mathcal{M}_{\lambda_2}$, 
the joint projection of $[\mathcal{T},\mathcal{T}]$ to $\lambda_1$ and $\lambda_2$, denoted as $\pi_{\lambda_1 \lambda_2}([\mathcal{T},\mathcal{T}])$, is isomorphic to the Cartesian product $\SU(\hilbert[M]_{\lambda_1}) \times \SU(\hilbert[M]_{\lambda_2})$. However, if $m \coloneqq \dim \mathcal{M}_{\lambda_1}= \dim \mathcal{M}_{\lambda_2}\ge 2$, then there exists a second possibility, namely case (ii) of Goursat's lemma. In particular, according to \cref{rem:SUiso,cor:Ugoursat}, for any fixed orthonormal bases on $\hilbert[M]_{\lambda_1}$ and $\hilbert[M]_{\lambda_2}$, there exists an $m \times m$ unitary matrix $W$ such that, for all $V\in[\mathcal{T},\mathcal{T}]$, one of the following holds
\bes\label{trq1}
\begin{align}
  [\pi_{\lambda_2}(V)] &= W [\pi_{\lambda_1}(V)] W^\dagger\\
  [\pi_{\lambda_2}(V)] &= W [\pi_{\lambda_1}(V)]^\ast W^\dagger.
\end{align}
\ees
Furthermore, the second part of \cref{cor:Ugoursat} implies that, for all $V \in \mathcal{T}$,
\begin{subequations}\label{eq:isophase}
\begin{align}
    [\pi_{\lambda_2}(V)] & = \e^{\i \phi} W [\pi_{\lambda_1}(V)] W^\dagger \\
    [\pi_{\lambda_2}(V)] & = \e^{\i \phi} W [\pi_{\lambda_1}(V)]^\ast W^\dagger,
\end{align}
\end{subequations}
for some phase depending on $V$.

Now to prove \cref{MainLemma}, we note that in both cases in \cref{eq:isophase}, $|\Tr(\pi_{\lambda_1}(V))|=|\Tr(\pi_{\lambda_2}(V))|$.  Therefore, if condition $\textbf{B}$ of \cref{MainLemma} holds, i.e., there exists a unitary $V\in \mathcal{T}$ such that $|\Tr(\pi_{\lambda_1}(V))|\neq |\Tr(\pi_{\lambda_2}(V))|$
then the case (ii) of Goursat's lemma is ruled out, which means 
\begin{equation}\label{eq:commind}
  \SU(\hilbert[M]_{\lambda_1}) \times \SU(\hilbert[M]_{\lambda_2})= \pi_{\lambda_1 \lambda_2}([\mathcal{T},\mathcal{T}]) \subseteq \pi_{\lambda_1 \lambda_2}(\mathcal{T})\ .
\end{equation}
Similarly, if condition \textbf{B}' holds, i.e., there exists $V\in\mathcal{T}$, such that $\pi_{\lambda_1}(V) = \ident_{\hilbert[M]_{\lambda_1}}$
and $\pi_{\lambda_2}(V) \neq \e^{\i\alpha} \ident_{\hilbert[M]_{\lambda_2}}$
for any phase $\alpha$, then, again \cref{eq:isophase} cannot hold, which implies the same result. 

This proves \cref{MainLemma} in the special case where $\Lambda$ contains only two irreps. Finally, we note that because the special unitary group $\SU(\mathcal{M}_\lambda)$ is perfect, and \cref{eq:commind} applies to all pairs $\lambda_1, \lambda_2 \in \Lambda$, applying Serre's lemma to the subdirect product $\pi_\Lambda([\mathcal{T},\mathcal{T}])\subseteq \prod_{\lambda\in\Lambda} \SU(\mathcal{M}_\lambda)$, we conclude that 
$\prod_{\lambda\in\Lambda} \SU(\hilbert[M]_{\lambda}) \cong \mathcal{SV}^G\subset \mathcal{T}$, 
which completes the proof of \cref{MainLemma}.

Next, we prove \cref{Mainlemma_gen}, which does not assume condition $\textbf{B}$. Then, in this case, the case (ii) in Goursat's lemma, i.e., correlations between different sectors, cannot be ruled out. This means that it is possible that for any $V\in\mathcal{T}$ its projection to two sectors $\lambda_1$ and $\lambda_2$, namely $ \pi_{\lambda_1}(V)$ and $\pi_{\lambda_2}(V)$, are related via one of the two possibilities in \cref{eq:isophase} where the unitary realized in $\lambda_1$ uniquely determines the unitary in $\lambda_2$, up to a possible phase. 

As we explain in the following, based on 
the two possibilities in Goursat's lemma, i.e., case (i) and case (ii), we can partition all irreps in $\Lambda$ into equivalency classes $\Delta_1, \cdots, \Delta_s$.  
Namely, for any $\lambda_1,\lambda_2\in \Lambda$ if case (i) of Goursat's lemma applies to $\pi_{\lambda_1 \lambda_2}([\mathcal{T},\mathcal{T}])$, then they are in different parts and if case (ii) applies they are in the same part.  Note that \cref{trq1} guarantees that this rule defines equivalency classes and can be applied consistently. In particular, if for $\lambda_1,\lambda_2, \lambda_3\in \Lambda$, the case (ii) applies to the pair $\lambda_1, \lambda_2$ and also to the pair $\lambda_2, \lambda_3$, then it also applies to the pair $\lambda_1,\lambda_3$, i.e., one of the two possibilities in \cref{trq1} is applicable to $\pi_{\lambda_1}(V)$ and $\pi_{\lambda_3}(V)$, and therefore they are in the same part as well. 
Also, note that the case (ii) of Goursat's lemma applies to $\lambda_1$ and $\lambda_2$ only if $\dim\mathcal{M}_{\lambda_1}=\dim\mathcal{M}_{\lambda_2}$, which means for all irreps in the same part this dimension is identical. 
We label these parts as $\Delta_1,\cdots, \Delta_s$, and  define the corresponding dimensions as $m_r=\dim\mathcal{M}_\lambda$ for any $\lambda\in \Delta_r$.

Next, we pick a representative $\lambda$ from each part $\Delta_r$, denoted as $\lambda_r$, which together define $\Lambda_\text{rep}=\{\lambda_r: r=1,\cdots, s\}$. Then, we apply Serre's lemma to the subdirect product
\be\label{proof2}
\prod_{\lambda \in \Lambda_\text{rep}} \pi_\lambda([\mathcal{T},\mathcal{T}]) \subseteq  \prod_{r=1}^s \SU(m_r) \ .
\ee
By the above construction, for any pair of parts $\Delta_r$ and $\Delta_{r'}$ their corresponding representatives 
$\lambda_{r}$ and $\lambda_{r'}$ satisfy 
 case (i) of Goursat's lemma, which means $\pi_{\lambda_r,\lambda_{r'}}([\mathcal{T},\mathcal{T}])=\SU(m_r)\times\SU(m_{r'})$. Therefore, by applying Serre's lemma, we conclude that \cref{proof2} holds as equality, i.e., the subdirect product is indeed a Cartesian product. This completes the proof of \cref{Mainlemma_gen}. 

Now we finish the proof of \cref{lem:phases}. Let $V \in \mathcal{T}$ be an arbitrary element. We show that there exists phases $\e^{\i \phi_\lambda} : \lambda \in \Lambda$ and $V' \in [\mathcal{T}, \mathcal{T}]$ such that $V = V' \sum_{\lambda \in \Lambda} \e^{\i \phi_\lambda} \Pi_\lambda$.

By \cref{proof2} there is an element $U \in [\mathcal{T}, \mathcal{T}]$ such that, for all $\lambda_r \in \Lambda_{\text{rep}}$, $\pi_{\lambda_r}(U V) = \e^{\i \phi_r} \ident_{\lambda_r}$ for some phase. By construction, every other $\lambda' \in \Lambda \setminus \Lambda_{\text{rep}}$ is contained in some equivalence class of irreps, say $\lambda' \in \Delta_r$. According to \cref{eq:isophase}, if $\pi_{\lambda_r}(U V)$ is proportional to the identity, $\pi_{\lambda'}(U V)$ must be as well. But this holds for all $\lambda' \in \Lambda$, hence $U V \in \mathcal{P}$ is in the subgroup of relative phases. In other words, $V = U^{-1} (U V)$ is of the desired form.

\section{Conclusion}

A series of recent works has shown that in the presence of symmetries, the locality of gates can severely restrict the set of realizable unitaries \cite{Marvian2022Restrict, marvian2022quditcircuit, Marvian2024Rotationally, marvian2024theoryabelian, kazi2024universality}. Interestingly, the type of restrictions significantly depends on the properties of the symmetry. In particular, in the case of Abelian symmetries, recent work \cite{marvian2024theoryabelian} has revealed a simple characterization of the group of all realizable unitaries with $k$-qudit gates, showing that some restrictions that appear in the case of non-Abelian symmetries, such as  $\SU(d)$ symmetry, cannot appear in the case of Abelian symmetries. In general, understanding the restrictions in the case of non-Abelian symmetries is significantly more complicated. Prior to the present work, a general formal framework or techniques for understanding such circuits did not exist.

In this work, we developed novel tools and a new framework for understanding semi-universality and universality in symmetric quantum circuits, which are particularly useful in the case of non-Abelian symmetries where the simpler characterization of \cite{marvian2024theoryabelian} is not applicable. We anticipate that, beyond the theory of symmetric quantum circuits, this framework and the new tools, such as our characterization of semi-universality in \cref{MainLemma}, characterization of subsystem universality in \cref{Mainlemma_gen}, and our \cref{thm:mulblocks} on extending controllability, will be more broadly useful in Quantum Control theory, Quantum Many-body Physics, and Quantum Thermodynamics. For instance, these new techniques could be useful for understanding the effect of non-Abelian conserved charges on the thermalization of quantum systems \cite{halpern2020noncommuting, majidy2023noncommuting}.

As an example of applications, we applied these new tools to settle a fundamental question in the context of $\SU(d)$-invariant circuits with $d\ge 3$. Namely, we showed the semi-universality of 3-qudit gates. We also discussed two corollaries of this result. Firstly, by studying 3-qudit circuits, we found a significantly simpler proof of the universality of 2-qudit gates, which was recently shown in a PhD thesis \cite{vanmeter2021universality}, based on the advanced results in Mathematical literature, namely, Marin's characterization of the Lie algebra generated by transpositions. Secondly, using this result, we showed that the distribution generated by random circuits with 3-qudit $\SU(d)$-invariant gates is a $t$-design, with $t$ growing quadratically with the number of qudits. Such quadratic scaling has been recently shown in \cite{marin2007algebre}, albeit using 4-qudit gates (We also emphasize that circuits with 2-qudit gates are not 2-design \cite{marvian2022quditcircuit}).

$\SU(d)$-invariant circuits, and more generally symmetric quantum circuits, have various applications in the context of quantum computation and control. In particular, they are useful for protecting information, e.g., via decoherence-free subspaces and noiseless subsystems \cite{lidar1998decoherence, Bacon:2000qf, divincenzo2000universal, Zanardi:97c, kempe2002exact, brod2013computational}. More broadly, the notion of charge-conserving unitaries appears in various areas of quantum information science, and it is crucial to understand how such unitaries can be realized. This is particularly relevant in the context of quantum thermodynamics \cite{janzing2000thermodynamic, FundLimitsNature, brandao2013resource, guryanova2016thermodynamics, lostaglio2015quantumPRX, halpern2016microcanonical, halpern2016beyond, lostaglio2017thermodynamic}, the resource theory of asymmetry \cite{gour2008resource, Marvian_thesis, marvian2013theory}, and other related areas such as Quantum reference frames \cite{QRF_BRS_07,marvian2008building}, and covariant error-correcting codes \cite{faist2020continuous, hayden2021error, woods2020continuous}. Furthermore, such circuits have been found useful in other contexts including variational quantum machine learning \cite{meyer2023exploiting, nguyen2022theory, sauvage2022building, zheng2023speeding}, and variational quantum eigensolvers for quantum chemistry \cite{barron2021preserving, shkolnikov2021avoiding, gard2020efficientsymmetry, streif2021quantum, wang2020x, barkoutsos2018quantum}. Besides such applications, the framework of symmetric quantum circuits has now become a standard tool in the area of many-body Physics, for modeling various physical phenomena, from quantum thermalization and quantum chaos \cite{khemani2018operator} to symmetry-protected topological order \cite{chen2010local, chen2011classification}. Therefore, we anticipate that the framework developed in this paper find applications beyond quantum computation. \\

\section*{Acknowledgements}
This work is supported by a collaboration between the US DOE and other Agencies. This material is based upon work supported by the U.S. Department of Energy, Office of Science, National Quantum Information Science Research Centers, Quantum Systems Accelerator. Additional support is
acknowledged from Army Research Office
(W911NF-21-1-0005), NSF Phy-2046195, and 
NSF QLCI grant OMA-2120757. H.L. was supported by the Quantum Science Center (QSC), a National Quantum Information Science Research Center of the U.S. Department of Energy (DOE) and by the U.S. Department of Energy, Office of Science, Office of Nuclear Physics (NP) contract DE-AC52-06NA25396.

\bibliography{refs}

\begin{thebibliography}{72}%
\makeatletter
\providecommand \@ifxundefined [1]{%
 \@ifx{#1\undefined}
}%
\providecommand \@ifnum [1]{%
 \ifnum #1\expandafter \@firstoftwo
 \else \expandafter \@secondoftwo
 \fi
}%
\providecommand \@ifx [1]{%
 \ifx #1\expandafter \@firstoftwo
 \else \expandafter \@secondoftwo
 \fi
}%
\providecommand \natexlab [1]{#1}%
\providecommand \enquote  [1]{``#1''}%
\providecommand \bibnamefont  [1]{#1}%
\providecommand \bibfnamefont [1]{#1}%
\providecommand \citenamefont [1]{#1}%
\providecommand \href@noop [0]{\@secondoftwo}%
\providecommand \href [0]{\begingroup \@sanitize@url \@href}%
\providecommand \@href[1]{\@@startlink{#1}\@@href}%
\providecommand \@@href[1]{\endgroup#1\@@endlink}%
\providecommand \@sanitize@url [0]{\catcode `\\12\catcode `\$12\catcode
  `\&12\catcode `\#12\catcode `\^12\catcode `\_12\catcode `\%12\relax}%
\providecommand \@@startlink[1]{}%
\providecommand \@@endlink[0]{}%
\providecommand \url  [0]{\begingroup\@sanitize@url \@url }%
\providecommand \@url [1]{\endgroup\@href {#1}{\urlprefix }}%
\providecommand \urlprefix  [0]{URL }%
\providecommand \Eprint [0]{\href }%
\providecommand \doibase [0]{https://doi.org/}%
\providecommand \selectlanguage [0]{\@gobble}%
\providecommand \bibinfo  [0]{\@secondoftwo}%
\providecommand \bibfield  [0]{\@secondoftwo}%
\providecommand \translation [1]{[#1]}%
\providecommand \BibitemOpen [0]{}%
\providecommand \bibitemStop [0]{}%
\providecommand \bibitemNoStop [0]{.\EOS\space}%
\providecommand \EOS [0]{\spacefactor3000\relax}%
\providecommand \BibitemShut  [1]{\csname bibitem#1\endcsname}%
\let\auto@bib@innerbib\@empty
\bibitem [{\citenamefont {DiVincenzo}(1995)}]{divincenzo1995two}%
  \BibitemOpen
  \bibfield  {author} {\bibinfo {author} {\bibfnamefont {D.~P.}\ \bibnamefont
  {DiVincenzo}},\ }\bibfield  {title} {\bibinfo {title} {Two-bit gates are
  universal for quantum computation},\ }\href@noop {} {\bibfield  {journal}
  {\bibinfo  {journal} {Physical Review A}\ }\textbf {\bibinfo {volume} {51}},\
  \bibinfo {pages} {1015} (\bibinfo {year} {1995})}\BibitemShut {NoStop}%
\bibitem [{\citenamefont {Lloyd}(1995)}]{lloyd1995almost}%
  \BibitemOpen
  \bibfield  {author} {\bibinfo {author} {\bibfnamefont {S.}~\bibnamefont
  {Lloyd}},\ }\bibfield  {title} {\bibinfo {title} {Almost any quantum logic
  gate is universal},\ }\href@noop {} {\bibfield  {journal} {\bibinfo
  {journal} {Physical Review Letters}\ }\textbf {\bibinfo {volume} {75}},\
  \bibinfo {pages} {346} (\bibinfo {year} {1995})}\BibitemShut {NoStop}%
\bibitem [{\citenamefont {Deutsch}\ \emph {et~al.}(1995)\citenamefont
  {Deutsch}, \citenamefont {Barenco},\ and\ \citenamefont
  {Ekert}}]{deutsch1995universality}%
  \BibitemOpen
  \bibfield  {author} {\bibinfo {author} {\bibfnamefont {D.~E.}\ \bibnamefont
  {Deutsch}}, \bibinfo {author} {\bibfnamefont {A.}~\bibnamefont {Barenco}},\
  and\ \bibinfo {author} {\bibfnamefont {A.}~\bibnamefont {Ekert}},\ }\bibfield
   {title} {\bibinfo {title} {Universality in quantum computation},\
  }\href@noop {} {\bibfield  {journal} {\bibinfo  {journal} {Proceedings of the
  Royal Society of London. Series A: Mathematical and Physical Sciences}\
  }\textbf {\bibinfo {volume} {449}},\ \bibinfo {pages} {669} (\bibinfo {year}
  {1995})}\BibitemShut {NoStop}%
\bibitem [{\citenamefont {Brylinski}\ and\ \citenamefont
  {Brylinski}(2002)}]{brylinski2002universal}%
  \BibitemOpen
  \bibfield  {author} {\bibinfo {author} {\bibfnamefont {J.-L.}\ \bibnamefont
  {Brylinski}}\ and\ \bibinfo {author} {\bibfnamefont {R.}~\bibnamefont
  {Brylinski}},\ }\bibfield  {title} {\bibinfo {title} {Universal quantum
  gates},\ }\href@noop {} {\bibfield  {journal} {\bibinfo  {journal}
  {Mathematics of quantum computation}\ }\textbf {\bibinfo {volume} {79}}
  (\bibinfo {year} {2002})}\BibitemShut {NoStop}%
\bibitem [{\citenamefont {Marvian}(2022)}]{Marvian2022Restrict}%
  \BibitemOpen
  \bibfield  {author} {\bibinfo {author} {\bibfnamefont {I.}~\bibnamefont
  {Marvian}},\ }\bibfield  {title} {\bibinfo {title} {Restrictions on
  realizable unitary operations imposed by symmetry and locality},\ }\href
  {https://doi.org/10.1038/s41567-021-01464-0} {\bibfield  {journal} {\bibinfo
  {journal} {Nature Physics}\ }\textbf {\bibinfo {volume} {18}},\ \bibinfo
  {pages} {283–289} (\bibinfo {year} {2022})}\BibitemShut {NoStop}%
\bibitem [{\citenamefont {Marvian}\ \emph {et~al.}(2022)\citenamefont
  {Marvian}, \citenamefont {Liu},\ and\ \citenamefont
  {Hulse}}]{marvian2022quditcircuit}%
  \BibitemOpen
  \bibfield  {author} {\bibinfo {author} {\bibfnamefont {I.}~\bibnamefont
  {Marvian}}, \bibinfo {author} {\bibfnamefont {H.}~\bibnamefont {Liu}},\ and\
  \bibinfo {author} {\bibfnamefont {A.}~\bibnamefont {Hulse}},\ }\href
  {https://arxiv.org/abs/2105.12877} {\bibinfo {title} {Qudit circuits with
  su(d) symmetry: Locality imposes additional conservation laws}} (\bibinfo
  {year} {2022}),\ \Eprint {https://arxiv.org/abs/2105.12877} {arXiv:2105.12877
  [quant-ph]} \BibitemShut {NoStop}%
\bibitem [{\citenamefont {Marvian}\ \emph {et~al.}(2024)\citenamefont
  {Marvian}, \citenamefont {Liu},\ and\ \citenamefont
  {Hulse}}]{Marvian2024Rotationally}%
  \BibitemOpen
  \bibfield  {author} {\bibinfo {author} {\bibfnamefont {I.}~\bibnamefont
  {Marvian}}, \bibinfo {author} {\bibfnamefont {H.}~\bibnamefont {Liu}},\ and\
  \bibinfo {author} {\bibfnamefont {A.}~\bibnamefont {Hulse}},\ }\bibfield
  {title} {\bibinfo {title} {Rotationally invariant circuits: Universality with
  the exchange interaction and two ancilla qubits},\ }\bibfield  {journal}
  {\bibinfo  {journal} {Physical Review Letters}\ }\textbf {\bibinfo {volume}
  {132}},\ \href {https://doi.org/10.1103/physrevlett.132.130201}
  {10.1103/physrevlett.132.130201} (\bibinfo {year} {2024})\BibitemShut
  {NoStop}%
\bibitem [{\citenamefont {Marvian}(2024)}]{marvian2024theoryabelian}%
  \BibitemOpen
  \bibfield  {author} {\bibinfo {author} {\bibfnamefont {I.}~\bibnamefont
  {Marvian}},\ }\href {https://arxiv.org/abs/2302.12466} {\bibinfo {title}
  {Theory of quantum circuits with abelian symmetries}} (\bibinfo {year}
  {2024}),\ \Eprint {https://arxiv.org/abs/2302.12466} {arXiv:2302.12466
  [quant-ph]} \BibitemShut {NoStop}%
\bibitem [{\citenamefont {Okounkov}\ and\ \citenamefont
  {Vershik}(2005)}]{okounkov2005}%
  \BibitemOpen
  \bibfield  {author} {\bibinfo {author} {\bibfnamefont {A.}~\bibnamefont
  {Okounkov}}\ and\ \bibinfo {author} {\bibfnamefont {A.}~\bibnamefont
  {Vershik}},\ }\bibfield  {title} {\bibinfo {title} {A new approach to
  representation theory of symmetric groups},\ }\href
  {https://doi.org/10.1007/BF02433451} {\bibfield  {journal} {\bibinfo
  {journal} {Selecta Mathematica}\ }\textbf {\bibinfo {volume} {2}},\ \bibinfo
  {pages} {581} (\bibinfo {year} {2005})}\BibitemShut {NoStop}%
\bibitem [{\citenamefont {Zheng}\ \emph
  {et~al.}(2023{\natexlab{a}})\citenamefont {Zheng}, \citenamefont {Li},
  \citenamefont {Liu}, \citenamefont {Strelchuk},\ and\ \citenamefont
  {Kondor}}]{Zheng_2023}%
  \BibitemOpen
  \bibfield  {author} {\bibinfo {author} {\bibfnamefont {H.}~\bibnamefont
  {Zheng}}, \bibinfo {author} {\bibfnamefont {Z.}~\bibnamefont {Li}}, \bibinfo
  {author} {\bibfnamefont {J.}~\bibnamefont {Liu}}, \bibinfo {author}
  {\bibfnamefont {S.}~\bibnamefont {Strelchuk}},\ and\ \bibinfo {author}
  {\bibfnamefont {R.}~\bibnamefont {Kondor}},\ }\bibfield  {title} {\bibinfo
  {title} {Speeding up learning quantum states through group equivariant
  convolutional quantum ansätze},\ }\bibfield  {journal} {\bibinfo  {journal}
  {PRX Quantum}\ }\textbf {\bibinfo {volume} {4}},\ \href
  {https://doi.org/10.1103/prxquantum.4.020327} {10.1103/prxquantum.4.020327}
  (\bibinfo {year} {2023}{\natexlab{a}})\BibitemShut {NoStop}%
\bibitem [{\citenamefont {van Meter}(2021)}]{vanmeter2021universality}%
  \BibitemOpen
  \bibfield  {author} {\bibinfo {author} {\bibfnamefont {J.~R.}\ \bibnamefont
  {van Meter}},\ }\href {https://arxiv.org/abs/2103.12303} {\bibinfo {title}
  {Universality of swap for qudits: a representation theory approach}}
  (\bibinfo {year} {2021}),\ \Eprint {https://arxiv.org/abs/2103.12303}
  {arXiv:2103.12303 [quant-ph]} \BibitemShut {NoStop}%
\bibitem [{\citenamefont {Marin}(2007)}]{marin2007algebre}%
  \BibitemOpen
  \bibfield  {author} {\bibinfo {author} {\bibfnamefont {I.}~\bibnamefont
  {Marin}},\ }\bibfield  {title} {\bibinfo {title} {L'alg{\`e}bre de lie des
  transpositions},\ }\href@noop {} {\bibfield  {journal} {\bibinfo  {journal}
  {Journal of Algebra}\ }\textbf {\bibinfo {volume} {310}},\ \bibinfo {pages}
  {742} (\bibinfo {year} {2007})}\BibitemShut {NoStop}%
\bibitem [{\citenamefont {Liu}\ \emph {et~al.}()\citenamefont {Liu},
  \citenamefont {Hulse},\ and\ \citenamefont {Marvian}}]{liu2024control}%
  \BibitemOpen
  \bibfield  {author} {\bibinfo {author} {\bibfnamefont {H.}~\bibnamefont
  {Liu}}, \bibinfo {author} {\bibfnamefont {A.}~\bibnamefont {Hulse}},\ and\
  \bibinfo {author} {\bibfnamefont {I.}~\bibnamefont {Marvian}},\ }\href@noop
  {} {\ }\bibinfo {note} {Under preparation}\BibitemShut {NoStop}%
\bibitem [{\citenamefont {Goodman}\ and\ \citenamefont
  {Wallach}(2009)}]{goodman2009symmetry}%
  \BibitemOpen
  \bibfield  {author} {\bibinfo {author} {\bibfnamefont {R.}~\bibnamefont
  {Goodman}}\ and\ \bibinfo {author} {\bibfnamefont {N.~R.}\ \bibnamefont
  {Wallach}},\ }\href@noop {} {\emph {\bibinfo {title} {Symmetry,
  representations, and invariants}}},\ Vol.\ \bibinfo {volume} {255}\ (\bibinfo
   {publisher} {Springer},\ \bibinfo {year} {2009})\BibitemShut {NoStop}%
\bibitem [{\citenamefont {Harrow}(2005)}]{harrow2005applications}%
  \BibitemOpen
  \bibfield  {author} {\bibinfo {author} {\bibfnamefont {A.~W.}\ \bibnamefont
  {Harrow}},\ }\bibfield  {title} {\bibinfo {title} {Applications of coherent
  classical communication and the schur transform to quantum information
  theory},\ }\href@noop {} {\bibfield  {journal} {\bibinfo  {journal} {arXiv
  preprint quant-ph/0512255}\ } (\bibinfo {year} {2005})}\BibitemShut {NoStop}%
\bibitem [{\citenamefont {Nielsen}\ and\ \citenamefont
  {Chuang}(2000)}]{nielsen2000quantum}%
  \BibitemOpen
  \bibfield  {author} {\bibinfo {author} {\bibfnamefont {M.}~\bibnamefont
  {Nielsen}}\ and\ \bibinfo {author} {\bibfnamefont {I.}~\bibnamefont
  {Chuang}},\ }\href@noop {} {\emph {\bibinfo {title} {Quantum Computation and
  Quantum Information}}},\ Cambridge Series on Information and the Natural
  Sciences\ (\bibinfo  {publisher} {Cambridge University Press},\ \bibinfo
  {year} {2000})\BibitemShut {NoStop}%
\bibitem [{\citenamefont {Fulton}\ and\ \citenamefont
  {Harris}(2013)}]{fulton2013representation}%
  \BibitemOpen
  \bibfield  {author} {\bibinfo {author} {\bibfnamefont {W.}~\bibnamefont
  {Fulton}}\ and\ \bibinfo {author} {\bibfnamefont {J.}~\bibnamefont
  {Harris}},\ }\href@noop {} {\emph {\bibinfo {title} {Representation theory: a
  first course}}},\ Vol.\ \bibinfo {volume} {129}\ (\bibinfo  {publisher}
  {Springer Science \& Business Media},\ \bibinfo {year} {2013})\BibitemShut
  {NoStop}%
\bibitem [{\citenamefont {Erd{\"o}s}\ and\ \citenamefont
  {Lehner}(1941)}]{erdos1941}%
  \BibitemOpen
  \bibfield  {author} {\bibinfo {author} {\bibfnamefont {P.}~\bibnamefont
  {Erd{\"o}s}}\ and\ \bibinfo {author} {\bibfnamefont {J.}~\bibnamefont
  {Lehner}},\ }\bibfield  {title} {\bibinfo {title} {{The distribution of the
  number of summands in the partitions of a positive integer}},\ }\href
  {https://doi.org/10.1215/S0012-7094-41-00826-8} {\bibfield  {journal}
  {\bibinfo  {journal} {Duke Mathematical Journal}\ }\textbf {\bibinfo {volume}
  {8}},\ \bibinfo {pages} {335 } (\bibinfo {year} {1941})}\BibitemShut
  {NoStop}%
\bibitem [{\citenamefont {Alfonsin}(2005)}]{alfonsin2005diophantine}%
  \BibitemOpen
  \bibfield  {author} {\bibinfo {author} {\bibfnamefont {J.}~\bibnamefont
  {Alfonsin}},\ }\href {https://books.google.com/books?id=sCwTDAAAQBAJ} {\emph
  {\bibinfo {title} {The Diophantine Frobenius Problem}}},\ Oxford Lecture
  Series in Mathematics and Its Applications\ (\bibinfo  {publisher} {OUP
  Oxford},\ \bibinfo {year} {2005})\BibitemShut {NoStop}%
\bibitem [{\citenamefont {Harrow}\ and\ \citenamefont
  {Low}(2009)}]{harrow2009random}%
  \BibitemOpen
  \bibfield  {author} {\bibinfo {author} {\bibfnamefont {A.~W.}\ \bibnamefont
  {Harrow}}\ and\ \bibinfo {author} {\bibfnamefont {R.~A.}\ \bibnamefont
  {Low}},\ }\bibfield  {title} {\bibinfo {title} {Random quantum circuits are
  approximate 2-designs},\ }\href@noop {} {\bibfield  {journal} {\bibinfo
  {journal} {Communications in Mathematical Physics}\ }\textbf {\bibinfo
  {volume} {291}},\ \bibinfo {pages} {257} (\bibinfo {year}
  {2009})}\BibitemShut {NoStop}%
\bibitem [{\citenamefont {Emerson}\ \emph {et~al.}(2005)\citenamefont
  {Emerson}, \citenamefont {Livine},\ and\ \citenamefont
  {Lloyd}}]{emerson2005convergence}%
  \BibitemOpen
  \bibfield  {author} {\bibinfo {author} {\bibfnamefont {J.}~\bibnamefont
  {Emerson}}, \bibinfo {author} {\bibfnamefont {E.}~\bibnamefont {Livine}},\
  and\ \bibinfo {author} {\bibfnamefont {S.}~\bibnamefont {Lloyd}},\ }\bibfield
   {title} {\bibinfo {title} {Convergence conditions for random quantum
  circuits},\ }\href@noop {} {\bibfield  {journal} {\bibinfo  {journal}
  {Physical Review A}\ }\textbf {\bibinfo {volume} {72}},\ \bibinfo {pages}
  {060302} (\bibinfo {year} {2005})}\BibitemShut {NoStop}%
\bibitem [{\citenamefont {Brandao}\ \emph
  {et~al.}(2016{\natexlab{a}})\citenamefont {Brandao}, \citenamefont {Harrow},\
  and\ \citenamefont {Horodecki}}]{brandao2016efficient}%
  \BibitemOpen
  \bibfield  {author} {\bibinfo {author} {\bibfnamefont {F.~G.}\ \bibnamefont
  {Brandao}}, \bibinfo {author} {\bibfnamefont {A.~W.}\ \bibnamefont
  {Harrow}},\ and\ \bibinfo {author} {\bibfnamefont {M.}~\bibnamefont
  {Horodecki}},\ }\bibfield  {title} {\bibinfo {title} {Efficient quantum
  pseudorandomness},\ }\href@noop {} {\bibfield  {journal} {\bibinfo  {journal}
  {Physical review letters}\ }\textbf {\bibinfo {volume} {116}},\ \bibinfo
  {pages} {170502} (\bibinfo {year} {2016}{\natexlab{a}})}\BibitemShut
  {NoStop}%
\bibitem [{\citenamefont {Brandao}\ \emph
  {et~al.}(2016{\natexlab{b}})\citenamefont {Brandao}, \citenamefont {Harrow},\
  and\ \citenamefont {Horodecki}}]{brandao2016local}%
  \BibitemOpen
  \bibfield  {author} {\bibinfo {author} {\bibfnamefont {F.~G.}\ \bibnamefont
  {Brandao}}, \bibinfo {author} {\bibfnamefont {A.~W.}\ \bibnamefont
  {Harrow}},\ and\ \bibinfo {author} {\bibfnamefont {M.}~\bibnamefont
  {Horodecki}},\ }\bibfield  {title} {\bibinfo {title} {Local random quantum
  circuits are approximate polynomial-designs},\ }\href@noop {} {\bibfield
  {journal} {\bibinfo  {journal} {Communications in Mathematical Physics}\
  }\textbf {\bibinfo {volume} {346}},\ \bibinfo {pages} {397} (\bibinfo {year}
  {2016}{\natexlab{b}})}\BibitemShut {NoStop}%
\bibitem [{\citenamefont {Zeier}\ and\ \citenamefont
  {Zimborás}(2015)}]{robert2015squares}%
  \BibitemOpen
  \bibfield  {author} {\bibinfo {author} {\bibfnamefont {R.}~\bibnamefont
  {Zeier}}\ and\ \bibinfo {author} {\bibfnamefont {Z.}~\bibnamefont
  {Zimborás}},\ }\bibfield  {title} {\bibinfo {title} {{On squares of
  representations of compact Lie algebras}},\ }\href
  {https://doi.org/10.1063/1.4928410} {\bibfield  {journal} {\bibinfo
  {journal} {Journal of Mathematical Physics}\ }\textbf {\bibinfo {volume}
  {56}},\ \bibinfo {pages} {081702} (\bibinfo {year} {2015})},\ \Eprint
  {https://arxiv.org/abs/https://pubs.aip.org/aip/jmp/article-pdf/doi/10.1063/1.4928410/15804949/081702\_1\_online.pdf}
  {https://pubs.aip.org/aip/jmp/article-pdf/doi/10.1063/1.4928410/15804949/081702\_1\_online.pdf}
  \BibitemShut {NoStop}%
\bibitem [{\citenamefont {Zimbor{\'a}s}\ \emph {et~al.}(2015)\citenamefont
  {Zimbor{\'a}s}, \citenamefont {Zeier}, \citenamefont
  {Schulte-Herbr{\"u}ggen},\ and\ \citenamefont
  {Burgarth}}]{zimboras2015symmetry}%
  \BibitemOpen
  \bibfield  {author} {\bibinfo {author} {\bibfnamefont {Z.}~\bibnamefont
  {Zimbor{\'a}s}}, \bibinfo {author} {\bibfnamefont {R.}~\bibnamefont {Zeier}},
  \bibinfo {author} {\bibfnamefont {T.}~\bibnamefont
  {Schulte-Herbr{\"u}ggen}},\ and\ \bibinfo {author} {\bibfnamefont
  {D.}~\bibnamefont {Burgarth}},\ }\bibfield  {title} {\bibinfo {title}
  {Symmetry criteria for quantum simulability of effective interactions},\
  }\href@noop {} {\bibfield  {journal} {\bibinfo  {journal} {Physical Review
  A}\ }\textbf {\bibinfo {volume} {92}},\ \bibinfo {pages} {042309} (\bibinfo
  {year} {2015})}\BibitemShut {NoStop}%
\bibitem [{\citenamefont {Li}\ \emph {et~al.}(2023)\citenamefont {Li},
  \citenamefont {Zheng}, \citenamefont {Liu}, \citenamefont {Jiang},\ and\
  \citenamefont {Liu}}]{Li:2023mek}%
  \BibitemOpen
  \bibfield  {author} {\bibinfo {author} {\bibfnamefont {Z.}~\bibnamefont
  {Li}}, \bibinfo {author} {\bibfnamefont {H.}~\bibnamefont {Zheng}}, \bibinfo
  {author} {\bibfnamefont {J.}~\bibnamefont {Liu}}, \bibinfo {author}
  {\bibfnamefont {L.}~\bibnamefont {Jiang}},\ and\ \bibinfo {author}
  {\bibfnamefont {Z.-W.}\ \bibnamefont {Liu}},\ }\bibfield  {title} {\bibinfo
  {title} {{Designs from Local Random Quantum Circuits with SU(d) Symmetry}},\
  }\href@noop {} {\  (\bibinfo {year} {2023})},\ \Eprint
  {https://arxiv.org/abs/2309.08155} {arXiv:2309.08155 [quant-ph]} \BibitemShut
  {NoStop}%
\bibitem [{\citenamefont {Rasala}(1977)}]{rasala1977minimal}%
  \BibitemOpen
  \bibfield  {author} {\bibinfo {author} {\bibfnamefont {R.}~\bibnamefont
  {Rasala}},\ }\bibfield  {title} {\bibinfo {title} {On the minimal degrees of
  characters of sn},\ }\href@noop {} {\bibfield  {journal} {\bibinfo  {journal}
  {Journal of Algebra}\ }\textbf {\bibinfo {volume} {45}},\ \bibinfo {pages}
  {132} (\bibinfo {year} {1977})}\BibitemShut {NoStop}%
\bibitem [{\citenamefont {Hearth}\ \emph {et~al.}(2023)\citenamefont {Hearth},
  \citenamefont {Flynn}, \citenamefont {Chandran},\ and\ \citenamefont
  {Laumann}}]{hearth2023unitary}%
  \BibitemOpen
  \bibfield  {author} {\bibinfo {author} {\bibfnamefont {S.~N.}\ \bibnamefont
  {Hearth}}, \bibinfo {author} {\bibfnamefont {M.~O.}\ \bibnamefont {Flynn}},
  \bibinfo {author} {\bibfnamefont {A.}~\bibnamefont {Chandran}},\ and\
  \bibinfo {author} {\bibfnamefont {C.~R.}\ \bibnamefont {Laumann}},\
  }\bibfield  {title} {\bibinfo {title} {Unitary k-designs from random
  number-conserving quantum circuits},\ }\href@noop {} {\bibfield  {journal}
  {\bibinfo  {journal} {arXiv preprint arXiv:2306.01035}\ } (\bibinfo {year}
  {2023})}\BibitemShut {NoStop}%
\bibitem [{\citenamefont {Lang}(2005)}]{lang2005algebra}%
  \BibitemOpen
  \bibfield  {author} {\bibinfo {author} {\bibfnamefont {S.}~\bibnamefont
  {Lang}},\ }\href {https://books.google.com/books?id=Fge-BwqhqIYC} {\emph
  {\bibinfo {title} {Algebra}}},\ Graduate Texts in Mathematics\ (\bibinfo
  {publisher} {Springer New York},\ \bibinfo {year} {2005})\BibitemShut
  {NoStop}%
\bibitem [{\citenamefont {Ribet}(1975)}]{Ribet1974adic}%
  \BibitemOpen
  \bibfield  {author} {\bibinfo {author} {\bibfnamefont {K.~A.}\ \bibnamefont
  {Ribet}},\ }\bibfield  {title} {\bibinfo {title} {On $\ell$-adic
  representations attached to modular forms},\ }\href
  {https://doi.org/10.1007/BF01425561} {\bibfield  {journal} {\bibinfo
  {journal} {Invent. Math.}\ }\textbf {\bibinfo {volume} {28}},\ \bibinfo
  {pages} {245–275} (\bibinfo {year} {1975})}\BibitemShut {NoStop}%
\bibitem [{\citenamefont {Tao}(2021)}]{tao2021goursat}%
  \BibitemOpen
  \bibfield  {author} {\bibinfo {author} {\bibfnamefont {T.}~\bibnamefont
  {Tao}},\ }\href
  {https://terrytao.wordpress.com/2021/05/07/goursat-and-furstenberg-weiss-type-lemmas/}
  {\bibinfo {title} {{Goursat} and {Furstenberg-Weiss} type lemmas}} (\bibinfo
  {year} {2021})\BibitemShut {NoStop}%
\bibitem [{\citenamefont {Kazi}\ \emph {et~al.}(2024)\citenamefont {Kazi},
  \citenamefont {Larocca},\ and\ \citenamefont
  {Cerezo}}]{kazi2024universality}%
  \BibitemOpen
  \bibfield  {author} {\bibinfo {author} {\bibfnamefont {S.}~\bibnamefont
  {Kazi}}, \bibinfo {author} {\bibfnamefont {M.}~\bibnamefont {Larocca}},\ and\
  \bibinfo {author} {\bibfnamefont {M.}~\bibnamefont {Cerezo}},\ }\bibfield
  {title} {\bibinfo {title} {On the universality of sn-equivariant k-body
  gates},\ }\href@noop {} {\bibfield  {journal} {\bibinfo  {journal} {New
  Journal of Physics}\ }\textbf {\bibinfo {volume} {26}},\ \bibinfo {pages}
  {053030} (\bibinfo {year} {2024})}\BibitemShut {NoStop}%
\bibitem [{\citenamefont {Halpern}\ \emph {et~al.}(2020)\citenamefont
  {Halpern}, \citenamefont {Beverland},\ and\ \citenamefont
  {Kalev}}]{halpern2020noncommuting}%
  \BibitemOpen
  \bibfield  {author} {\bibinfo {author} {\bibfnamefont {N.~Y.}\ \bibnamefont
  {Halpern}}, \bibinfo {author} {\bibfnamefont {M.~E.}\ \bibnamefont
  {Beverland}},\ and\ \bibinfo {author} {\bibfnamefont {A.}~\bibnamefont
  {Kalev}},\ }\bibfield  {title} {\bibinfo {title} {Noncommuting conserved
  charges in quantum many-body thermalization},\ }\href@noop {} {\bibfield
  {journal} {\bibinfo  {journal} {Physical Review E}\ }\textbf {\bibinfo
  {volume} {101}},\ \bibinfo {pages} {042117} (\bibinfo {year}
  {2020})}\BibitemShut {NoStop}%
\bibitem [{\citenamefont {Majidy}\ \emph {et~al.}(2023)\citenamefont {Majidy},
  \citenamefont {Braasch~Jr}, \citenamefont {Lasek}, \citenamefont {Upadhyaya},
  \citenamefont {Kalev},\ and\ \citenamefont
  {Yunger~Halpern}}]{majidy2023noncommuting}%
  \BibitemOpen
  \bibfield  {author} {\bibinfo {author} {\bibfnamefont {S.}~\bibnamefont
  {Majidy}}, \bibinfo {author} {\bibfnamefont {W.~F.}\ \bibnamefont
  {Braasch~Jr}}, \bibinfo {author} {\bibfnamefont {A.}~\bibnamefont {Lasek}},
  \bibinfo {author} {\bibfnamefont {T.}~\bibnamefont {Upadhyaya}}, \bibinfo
  {author} {\bibfnamefont {A.}~\bibnamefont {Kalev}},\ and\ \bibinfo {author}
  {\bibfnamefont {N.}~\bibnamefont {Yunger~Halpern}},\ }\bibfield  {title}
  {\bibinfo {title} {Noncommuting conserved charges in quantum thermodynamics
  and beyond},\ }\href@noop {} {\bibfield  {journal} {\bibinfo  {journal}
  {Nature Reviews Physics}\ }\textbf {\bibinfo {volume} {5}},\ \bibinfo {pages}
  {689} (\bibinfo {year} {2023})}\BibitemShut {NoStop}%
\bibitem [{\citenamefont {Lidar}\ \emph {et~al.}(1998)\citenamefont {Lidar},
  \citenamefont {Chuang},\ and\ \citenamefont {Whaley}}]{lidar1998decoherence}%
  \BibitemOpen
  \bibfield  {author} {\bibinfo {author} {\bibfnamefont {D.~A.}\ \bibnamefont
  {Lidar}}, \bibinfo {author} {\bibfnamefont {I.~L.}\ \bibnamefont {Chuang}},\
  and\ \bibinfo {author} {\bibfnamefont {K.~B.}\ \bibnamefont {Whaley}},\
  }\bibfield  {title} {\bibinfo {title} {Decoherence-free subspaces for quantum
  computation},\ }\href@noop {} {\bibfield  {journal} {\bibinfo  {journal}
  {Physical Review Letters}\ }\textbf {\bibinfo {volume} {81}},\ \bibinfo
  {pages} {2594} (\bibinfo {year} {1998})}\BibitemShut {NoStop}%
\bibitem [{\citenamefont {Bacon}\ \emph {et~al.}(2000)\citenamefont {Bacon},
  \citenamefont {Kempe}, \citenamefont {Lidar},\ and\ \citenamefont
  {Whaley}}]{Bacon:2000qf}%
  \BibitemOpen
  \bibfield  {author} {\bibinfo {author} {\bibfnamefont {D.}~\bibnamefont
  {Bacon}}, \bibinfo {author} {\bibfnamefont {J.}~\bibnamefont {Kempe}},
  \bibinfo {author} {\bibfnamefont {D.~A.}\ \bibnamefont {Lidar}},\ and\
  \bibinfo {author} {\bibfnamefont {K.~B.}\ \bibnamefont {Whaley}},\ }\bibfield
   {title} {\bibinfo {title} {Universal fault-tolerant quantum computation on
  decoherence-free subspaces},\ }\href
  {http://link.aps.org/doi/10.1103/PhysRevLett.85.1758} {\bibfield  {journal}
  {\bibinfo  {journal} {{Phys.~Rev.~Lett.}}\ }\textbf {\bibinfo {volume}
  {85}},\ \bibinfo {pages} {1758} (\bibinfo {year} {2000})}\BibitemShut
  {NoStop}%
\bibitem [{\citenamefont {DiVincenzo}\ \emph {et~al.}(2000)\citenamefont
  {DiVincenzo}, \citenamefont {Bacon}, \citenamefont {Kempe}, \citenamefont
  {Burkard},\ and\ \citenamefont {Whaley}}]{divincenzo2000universal}%
  \BibitemOpen
  \bibfield  {author} {\bibinfo {author} {\bibfnamefont {D.~P.}\ \bibnamefont
  {DiVincenzo}}, \bibinfo {author} {\bibfnamefont {D.}~\bibnamefont {Bacon}},
  \bibinfo {author} {\bibfnamefont {J.}~\bibnamefont {Kempe}}, \bibinfo
  {author} {\bibfnamefont {G.}~\bibnamefont {Burkard}},\ and\ \bibinfo {author}
  {\bibfnamefont {K.~B.}\ \bibnamefont {Whaley}},\ }\bibfield  {title}
  {\bibinfo {title} {Universal quantum computation with the exchange
  interaction},\ }\href@noop {} {\bibfield  {journal} {\bibinfo  {journal}
  {nature}\ }\textbf {\bibinfo {volume} {408}},\ \bibinfo {pages} {339}
  (\bibinfo {year} {2000})}\BibitemShut {NoStop}%
\bibitem [{\citenamefont {Zanardi}\ and\ \citenamefont
  {Rasetti}(1997)}]{Zanardi:97c}%
  \BibitemOpen
  \bibfield  {author} {\bibinfo {author} {\bibfnamefont {P.}~\bibnamefont
  {Zanardi}}\ and\ \bibinfo {author} {\bibfnamefont {M.}~\bibnamefont
  {Rasetti}},\ }\bibfield  {title} {\bibinfo {title} {Noiseless quantum
  codes},\ }\href@noop {} {\bibfield  {journal} {\bibinfo  {journal} {Phys.
  Rev. Lett.}\ }\textbf {\bibinfo {volume} {79}},\ \bibinfo {pages} {3306}
  (\bibinfo {year} {1997})}\BibitemShut {NoStop}%
\bibitem [{\citenamefont {Kempe}\ and\ \citenamefont
  {Whaley}(2002)}]{kempe2002exact}%
  \BibitemOpen
  \bibfield  {author} {\bibinfo {author} {\bibfnamefont {J.}~\bibnamefont
  {Kempe}}\ and\ \bibinfo {author} {\bibfnamefont {K.~B.}\ \bibnamefont
  {Whaley}},\ }\bibfield  {title} {\bibinfo {title} {Exact gate sequences for
  universal quantum computation using the xy interaction alone},\ }\href@noop
  {} {\bibfield  {journal} {\bibinfo  {journal} {Physical Review A}\ }\textbf
  {\bibinfo {volume} {65}},\ \bibinfo {pages} {052330} (\bibinfo {year}
  {2002})}\BibitemShut {NoStop}%
\bibitem [{\citenamefont {Brod}\ and\ \citenamefont
  {Childs}(2013)}]{brod2013computational}%
  \BibitemOpen
  \bibfield  {author} {\bibinfo {author} {\bibfnamefont {D.~J.}\ \bibnamefont
  {Brod}}\ and\ \bibinfo {author} {\bibfnamefont {A.~M.}\ \bibnamefont
  {Childs}},\ }\bibfield  {title} {\bibinfo {title} {The computational power of
  matchgates and the xy interaction on arbitrary graphs},\ }\href@noop {}
  {\bibfield  {journal} {\bibinfo  {journal} {arXiv preprint arXiv:1308.1463}\
  } (\bibinfo {year} {2013})}\BibitemShut {NoStop}%
\bibitem [{\citenamefont {Janzing}\ \emph {et~al.}(2000)\citenamefont
  {Janzing}, \citenamefont {Wocjan}, \citenamefont {Zeier}, \citenamefont
  {Geiss},\ and\ \citenamefont {Beth}}]{janzing2000thermodynamic}%
  \BibitemOpen
  \bibfield  {author} {\bibinfo {author} {\bibfnamefont {D.}~\bibnamefont
  {Janzing}}, \bibinfo {author} {\bibfnamefont {P.}~\bibnamefont {Wocjan}},
  \bibinfo {author} {\bibfnamefont {R.}~\bibnamefont {Zeier}}, \bibinfo
  {author} {\bibfnamefont {R.}~\bibnamefont {Geiss}},\ and\ \bibinfo {author}
  {\bibfnamefont {T.}~\bibnamefont {Beth}},\ }\bibfield  {title} {\bibinfo
  {title} {{Thermodynamic cost of reliability and low temperatures: tightening
  Landauer's principle and the Second Law}},\ }\href@noop {} {\bibfield
  {journal} {\bibinfo  {journal} {Int. J. Theor. Phys.}\ }\textbf {\bibinfo
  {volume} {39}},\ \bibinfo {pages} {2717} (\bibinfo {year}
  {2000})}\BibitemShut {NoStop}%
\bibitem [{\citenamefont {Horodecki}\ and\ \citenamefont
  {Oppenheim}(2013)}]{FundLimitsNature}%
  \BibitemOpen
  \bibfield  {author} {\bibinfo {author} {\bibfnamefont {M.}~\bibnamefont
  {Horodecki}}\ and\ \bibinfo {author} {\bibfnamefont {J.}~\bibnamefont
  {Oppenheim}},\ }\bibfield  {title} {\bibinfo {title} {{Fundamental
  limitations for quantum and nanoscale thermodynamics}},\ }\href@noop {}
  {\bibfield  {journal} {\bibinfo  {journal} {Nat. Commun.}\ }\textbf {\bibinfo
  {volume} {4}},\ \bibinfo {pages} {1} (\bibinfo {year} {2013})}\BibitemShut
  {NoStop}%
\bibitem [{\citenamefont {Brandao}\ \emph {et~al.}(2013)\citenamefont
  {Brandao}, \citenamefont {Horodecki}, \citenamefont {Oppenheim},
  \citenamefont {Renes},\ and\ \citenamefont {Spekkens}}]{brandao2013resource}%
  \BibitemOpen
  \bibfield  {author} {\bibinfo {author} {\bibfnamefont {F.~G.}\ \bibnamefont
  {Brandao}}, \bibinfo {author} {\bibfnamefont {M.}~\bibnamefont {Horodecki}},
  \bibinfo {author} {\bibfnamefont {J.}~\bibnamefont {Oppenheim}}, \bibinfo
  {author} {\bibfnamefont {J.~M.}\ \bibnamefont {Renes}},\ and\ \bibinfo
  {author} {\bibfnamefont {R.~W.}\ \bibnamefont {Spekkens}},\ }\bibfield
  {title} {\bibinfo {title} {Resource theory of quantum states out of thermal
  equilibrium},\ }\href@noop {} {\bibfield  {journal} {\bibinfo  {journal}
  {Physical review letters}\ }\textbf {\bibinfo {volume} {111}},\ \bibinfo
  {pages} {250404} (\bibinfo {year} {2013})}\BibitemShut {NoStop}%
\bibitem [{\citenamefont {Guryanova}\ \emph {et~al.}(2016)\citenamefont
  {Guryanova}, \citenamefont {Popescu}, \citenamefont {Short}, \citenamefont
  {Silva},\ and\ \citenamefont {Skrzypczyk}}]{guryanova2016thermodynamics}%
  \BibitemOpen
  \bibfield  {author} {\bibinfo {author} {\bibfnamefont {Y.}~\bibnamefont
  {Guryanova}}, \bibinfo {author} {\bibfnamefont {S.}~\bibnamefont {Popescu}},
  \bibinfo {author} {\bibfnamefont {A.~J.}\ \bibnamefont {Short}}, \bibinfo
  {author} {\bibfnamefont {R.}~\bibnamefont {Silva}},\ and\ \bibinfo {author}
  {\bibfnamefont {P.}~\bibnamefont {Skrzypczyk}},\ }\bibfield  {title}
  {\bibinfo {title} {Thermodynamics of quantum systems with multiple conserved
  quantities},\ }\href@noop {} {\bibfield  {journal} {\bibinfo  {journal}
  {Nature communications}\ }\textbf {\bibinfo {volume} {7}},\ \bibinfo {pages}
  {ncomms12049} (\bibinfo {year} {2016})}\BibitemShut {NoStop}%
\bibitem [{\citenamefont {Lostaglio}\ \emph {et~al.}(2015)\citenamefont
  {Lostaglio}, \citenamefont {Korzekwa}, \citenamefont {Jennings},\ and\
  \citenamefont {Rudolph}}]{lostaglio2015quantumPRX}%
  \BibitemOpen
  \bibfield  {author} {\bibinfo {author} {\bibfnamefont {M.}~\bibnamefont
  {Lostaglio}}, \bibinfo {author} {\bibfnamefont {K.}~\bibnamefont {Korzekwa}},
  \bibinfo {author} {\bibfnamefont {D.}~\bibnamefont {Jennings}},\ and\
  \bibinfo {author} {\bibfnamefont {T.}~\bibnamefont {Rudolph}},\ }\bibfield
  {title} {\bibinfo {title} {Quantum coherence, time-translation symmetry, and
  thermodynamics},\ }\href@noop {} {\bibfield  {journal} {\bibinfo  {journal}
  {Physical Review X}\ }\textbf {\bibinfo {volume} {5}},\ \bibinfo {pages}
  {021001} (\bibinfo {year} {2015})}\BibitemShut {NoStop}%
\bibitem [{\citenamefont {Halpern}\ \emph {et~al.}(2016)\citenamefont
  {Halpern}, \citenamefont {Faist}, \citenamefont {Oppenheim},\ and\
  \citenamefont {Winter}}]{halpern2016microcanonical}%
  \BibitemOpen
  \bibfield  {author} {\bibinfo {author} {\bibfnamefont {N.~Y.}\ \bibnamefont
  {Halpern}}, \bibinfo {author} {\bibfnamefont {P.}~\bibnamefont {Faist}},
  \bibinfo {author} {\bibfnamefont {J.}~\bibnamefont {Oppenheim}},\ and\
  \bibinfo {author} {\bibfnamefont {A.}~\bibnamefont {Winter}},\ }\bibfield
  {title} {\bibinfo {title} {Microcanonical and resource-theoretic derivations
  of the thermal state of a quantum system with noncommuting charges},\
  }\href@noop {} {\bibfield  {journal} {\bibinfo  {journal} {Nature
  communications}\ }\textbf {\bibinfo {volume} {7}},\ \bibinfo {pages} {12051}
  (\bibinfo {year} {2016})}\BibitemShut {NoStop}%
\bibitem [{\citenamefont {Halpern}\ and\ \citenamefont
  {Renes}(2016)}]{halpern2016beyond}%
  \BibitemOpen
  \bibfield  {author} {\bibinfo {author} {\bibfnamefont {N.~Y.}\ \bibnamefont
  {Halpern}}\ and\ \bibinfo {author} {\bibfnamefont {J.~M.}\ \bibnamefont
  {Renes}},\ }\bibfield  {title} {\bibinfo {title} {Beyond heat baths:
  Generalized resource theories for small-scale thermodynamics},\ }\href@noop
  {} {\bibfield  {journal} {\bibinfo  {journal} {Physical Review E}\ }\textbf
  {\bibinfo {volume} {93}},\ \bibinfo {pages} {022126} (\bibinfo {year}
  {2016})}\BibitemShut {NoStop}%
\bibitem [{\citenamefont {Lostaglio}\ \emph {et~al.}(2017)\citenamefont
  {Lostaglio}, \citenamefont {Jennings},\ and\ \citenamefont
  {Rudolph}}]{lostaglio2017thermodynamic}%
  \BibitemOpen
  \bibfield  {author} {\bibinfo {author} {\bibfnamefont {M.}~\bibnamefont
  {Lostaglio}}, \bibinfo {author} {\bibfnamefont {D.}~\bibnamefont
  {Jennings}},\ and\ \bibinfo {author} {\bibfnamefont {T.}~\bibnamefont
  {Rudolph}},\ }\bibfield  {title} {\bibinfo {title} {Thermodynamic resource
  theories, non-commutativity and maximum entropy principles},\ }\href@noop {}
  {\bibfield  {journal} {\bibinfo  {journal} {New Journal of Physics}\ }\textbf
  {\bibinfo {volume} {19}},\ \bibinfo {pages} {043008} (\bibinfo {year}
  {2017})}\BibitemShut {NoStop}%
\bibitem [{\citenamefont {Gour}\ and\ \citenamefont
  {Spekkens}(2008)}]{gour2008resource}%
  \BibitemOpen
  \bibfield  {author} {\bibinfo {author} {\bibfnamefont {G.}~\bibnamefont
  {Gour}}\ and\ \bibinfo {author} {\bibfnamefont {R.~W.}\ \bibnamefont
  {Spekkens}},\ }\bibfield  {title} {\bibinfo {title} {The resource theory of
  quantum reference frames: manipulations and monotones},\ }\href@noop {}
  {\bibfield  {journal} {\bibinfo  {journal} {New Journal of Physics}\ }\textbf
  {\bibinfo {volume} {10}},\ \bibinfo {pages} {033023} (\bibinfo {year}
  {2008})}\BibitemShut {NoStop}%
\bibitem [{\citenamefont {Marvian}(2012)}]{Marvian_thesis}%
  \BibitemOpen
  \bibfield  {author} {\bibinfo {author} {\bibfnamefont {I.}~\bibnamefont
  {Marvian}},\ }\emph {\bibinfo {title} {Symmetry, Asymmetry and Quantum
  Information, PhD thesis}},\ \href@noop {} {Ph.D. thesis},\ \bibinfo  {school}
  {University of Waterloo}, \bibinfo {address}
  {https://uwspace.uwaterloo.ca/handle/10012/7088} (\bibinfo {year}
  {2012})\BibitemShut {NoStop}%
\bibitem [{\citenamefont {Marvian}\ and\ \citenamefont
  {Spekkens}(2013)}]{marvian2013theory}%
  \BibitemOpen
  \bibfield  {author} {\bibinfo {author} {\bibfnamefont {I.}~\bibnamefont
  {Marvian}}\ and\ \bibinfo {author} {\bibfnamefont {R.~W.}\ \bibnamefont
  {Spekkens}},\ }\bibfield  {title} {\bibinfo {title} {The theory of
  manipulations of pure state asymmetry: I. basic tools, equivalence classes
  and single copy transformations},\ }\href@noop {} {\bibfield  {journal}
  {\bibinfo  {journal} {New Journal of Physics}\ }\textbf {\bibinfo {volume}
  {15}},\ \bibinfo {pages} {033001} (\bibinfo {year} {2013})}\BibitemShut
  {NoStop}%
\bibitem [{\citenamefont {Bartlett}\ \emph {et~al.}(2007)\citenamefont
  {Bartlett}, \citenamefont {Rudolph},\ and\ \citenamefont
  {Spekkens}}]{QRF_BRS_07}%
  \BibitemOpen
  \bibfield  {author} {\bibinfo {author} {\bibfnamefont {S.~D.}\ \bibnamefont
  {Bartlett}}, \bibinfo {author} {\bibfnamefont {T.}~\bibnamefont {Rudolph}},\
  and\ \bibinfo {author} {\bibfnamefont {R.~W.}\ \bibnamefont {Spekkens}},\
  }\bibfield  {title} {\bibinfo {title} {Reference frames, superselection
  rules, and quantum information},\ }\href@noop {} {\bibfield  {journal}
  {\bibinfo  {journal} {Reviews of Modern Physics}\ }\textbf {\bibinfo {volume}
  {79}},\ \bibinfo {pages} {555} (\bibinfo {year} {2007})}\BibitemShut
  {NoStop}%
\bibitem [{\citenamefont {Marvian}\ and\ \citenamefont
  {Mann}(2008)}]{marvian2008building}%
  \BibitemOpen
  \bibfield  {author} {\bibinfo {author} {\bibfnamefont {I.}~\bibnamefont
  {Marvian}}\ and\ \bibinfo {author} {\bibfnamefont {R.}~\bibnamefont {Mann}},\
  }\bibfield  {title} {\bibinfo {title} {Building all time evolutions with
  rotationally invariant hamiltonians},\ }\href@noop {} {\bibfield  {journal}
  {\bibinfo  {journal} {Physical Review A}\ }\textbf {\bibinfo {volume} {78}},\
  \bibinfo {pages} {022304} (\bibinfo {year} {2008})}\BibitemShut {NoStop}%
\bibitem [{\citenamefont {Faist}\ \emph {et~al.}(2020)\citenamefont {Faist},
  \citenamefont {Nezami}, \citenamefont {Albert}, \citenamefont {Salton},
  \citenamefont {Pastawski}, \citenamefont {Hayden},\ and\ \citenamefont
  {Preskill}}]{faist2020continuous}%
  \BibitemOpen
  \bibfield  {author} {\bibinfo {author} {\bibfnamefont {P.}~\bibnamefont
  {Faist}}, \bibinfo {author} {\bibfnamefont {S.}~\bibnamefont {Nezami}},
  \bibinfo {author} {\bibfnamefont {V.~V.}\ \bibnamefont {Albert}}, \bibinfo
  {author} {\bibfnamefont {G.}~\bibnamefont {Salton}}, \bibinfo {author}
  {\bibfnamefont {F.}~\bibnamefont {Pastawski}}, \bibinfo {author}
  {\bibfnamefont {P.}~\bibnamefont {Hayden}},\ and\ \bibinfo {author}
  {\bibfnamefont {J.}~\bibnamefont {Preskill}},\ }\bibfield  {title} {\bibinfo
  {title} {Continuous symmetries and approximate quantum error correction},\
  }\href@noop {} {\bibfield  {journal} {\bibinfo  {journal} {Physical Review
  X}\ }\textbf {\bibinfo {volume} {10}},\ \bibinfo {pages} {041018} (\bibinfo
  {year} {2020})}\BibitemShut {NoStop}%
\bibitem [{\citenamefont {Hayden}\ \emph {et~al.}(2021)\citenamefont {Hayden},
  \citenamefont {Nezami}, \citenamefont {Popescu},\ and\ \citenamefont
  {Salton}}]{hayden2021error}%
  \BibitemOpen
  \bibfield  {author} {\bibinfo {author} {\bibfnamefont {P.}~\bibnamefont
  {Hayden}}, \bibinfo {author} {\bibfnamefont {S.}~\bibnamefont {Nezami}},
  \bibinfo {author} {\bibfnamefont {S.}~\bibnamefont {Popescu}},\ and\ \bibinfo
  {author} {\bibfnamefont {G.}~\bibnamefont {Salton}},\ }\bibfield  {title}
  {\bibinfo {title} {Error correction of quantum reference frame information},\
  }\href@noop {} {\bibfield  {journal} {\bibinfo  {journal} {PRX Quantum}\
  }\textbf {\bibinfo {volume} {2}},\ \bibinfo {pages} {010326} (\bibinfo {year}
  {2021})}\BibitemShut {NoStop}%
\bibitem [{\citenamefont {Woods}\ and\ \citenamefont
  {Alhambra}(2020)}]{woods2020continuous}%
  \BibitemOpen
  \bibfield  {author} {\bibinfo {author} {\bibfnamefont {M.~P.}\ \bibnamefont
  {Woods}}\ and\ \bibinfo {author} {\bibfnamefont {{\'A}.~M.}\ \bibnamefont
  {Alhambra}},\ }\bibfield  {title} {\bibinfo {title} {Continuous groups of
  transversal gates for quantum error correcting codes from finite clock
  reference frames},\ }\href@noop {} {\bibfield  {journal} {\bibinfo  {journal}
  {Quantum}\ }\textbf {\bibinfo {volume} {4}},\ \bibinfo {pages} {245}
  (\bibinfo {year} {2020})}\BibitemShut {NoStop}%
\bibitem [{\citenamefont {Meyer}\ \emph {et~al.}(2023)\citenamefont {Meyer},
  \citenamefont {Mularski}, \citenamefont {Gil-Fuster}, \citenamefont {Mele},
  \citenamefont {Arzani}, \citenamefont {Wilms},\ and\ \citenamefont
  {Eisert}}]{meyer2023exploiting}%
  \BibitemOpen
  \bibfield  {author} {\bibinfo {author} {\bibfnamefont {J.~J.}\ \bibnamefont
  {Meyer}}, \bibinfo {author} {\bibfnamefont {M.}~\bibnamefont {Mularski}},
  \bibinfo {author} {\bibfnamefont {E.}~\bibnamefont {Gil-Fuster}}, \bibinfo
  {author} {\bibfnamefont {A.~A.}\ \bibnamefont {Mele}}, \bibinfo {author}
  {\bibfnamefont {F.}~\bibnamefont {Arzani}}, \bibinfo {author} {\bibfnamefont
  {A.}~\bibnamefont {Wilms}},\ and\ \bibinfo {author} {\bibfnamefont
  {J.}~\bibnamefont {Eisert}},\ }\bibfield  {title} {\bibinfo {title}
  {Exploiting symmetry in variational quantum machine learning},\ }\href@noop
  {} {\bibfield  {journal} {\bibinfo  {journal} {PRX Quantum}\ }\textbf
  {\bibinfo {volume} {4}},\ \bibinfo {pages} {010328} (\bibinfo {year}
  {2023})}\BibitemShut {NoStop}%
\bibitem [{\citenamefont {Nguyen}\ \emph {et~al.}(2022)\citenamefont {Nguyen},
  \citenamefont {Schatzki}, \citenamefont {Braccia}, \citenamefont {Ragone},
  \citenamefont {Coles}, \citenamefont {Sauvage}, \citenamefont {Larocca},\
  and\ \citenamefont {Cerezo}}]{nguyen2022theory}%
  \BibitemOpen
  \bibfield  {author} {\bibinfo {author} {\bibfnamefont {Q.~T.}\ \bibnamefont
  {Nguyen}}, \bibinfo {author} {\bibfnamefont {L.}~\bibnamefont {Schatzki}},
  \bibinfo {author} {\bibfnamefont {P.}~\bibnamefont {Braccia}}, \bibinfo
  {author} {\bibfnamefont {M.}~\bibnamefont {Ragone}}, \bibinfo {author}
  {\bibfnamefont {P.~J.}\ \bibnamefont {Coles}}, \bibinfo {author}
  {\bibfnamefont {F.}~\bibnamefont {Sauvage}}, \bibinfo {author} {\bibfnamefont
  {M.}~\bibnamefont {Larocca}},\ and\ \bibinfo {author} {\bibfnamefont
  {M.}~\bibnamefont {Cerezo}},\ }\bibfield  {title} {\bibinfo {title} {Theory
  for equivariant quantum neural networks},\ }\href@noop {} {\bibfield
  {journal} {\bibinfo  {journal} {arXiv preprint arXiv:2210.08566}\ } (\bibinfo
  {year} {2022})}\BibitemShut {NoStop}%
\bibitem [{\citenamefont {Sauvage}\ \emph {et~al.}(2022)\citenamefont
  {Sauvage}, \citenamefont {Larocca}, \citenamefont {Coles},\ and\
  \citenamefont {Cerezo}}]{sauvage2022building}%
  \BibitemOpen
  \bibfield  {author} {\bibinfo {author} {\bibfnamefont {F.}~\bibnamefont
  {Sauvage}}, \bibinfo {author} {\bibfnamefont {M.}~\bibnamefont {Larocca}},
  \bibinfo {author} {\bibfnamefont {P.~J.}\ \bibnamefont {Coles}},\ and\
  \bibinfo {author} {\bibfnamefont {M.}~\bibnamefont {Cerezo}},\ }\bibfield
  {title} {\bibinfo {title} {Building spatial symmetries into parameterized
  quantum circuits for faster training},\ }\href@noop {} {\bibfield  {journal}
  {\bibinfo  {journal} {arXiv preprint arXiv:2207.14413}\ } (\bibinfo {year}
  {2022})}\BibitemShut {NoStop}%
\bibitem [{\citenamefont {Zheng}\ \emph
  {et~al.}(2023{\natexlab{b}})\citenamefont {Zheng}, \citenamefont {Li},
  \citenamefont {Liu}, \citenamefont {Strelchuk},\ and\ \citenamefont
  {Kondor}}]{zheng2023speeding}%
  \BibitemOpen
  \bibfield  {author} {\bibinfo {author} {\bibfnamefont {H.}~\bibnamefont
  {Zheng}}, \bibinfo {author} {\bibfnamefont {Z.}~\bibnamefont {Li}}, \bibinfo
  {author} {\bibfnamefont {J.}~\bibnamefont {Liu}}, \bibinfo {author}
  {\bibfnamefont {S.}~\bibnamefont {Strelchuk}},\ and\ \bibinfo {author}
  {\bibfnamefont {R.}~\bibnamefont {Kondor}},\ }\bibfield  {title} {\bibinfo
  {title} {Speeding up learning quantum states through group equivariant
  convolutional quantum ans{\"a}tze},\ }\href@noop {} {\bibfield  {journal}
  {\bibinfo  {journal} {PRX Quantum}\ }\textbf {\bibinfo {volume} {4}},\
  \bibinfo {pages} {020327} (\bibinfo {year} {2023}{\natexlab{b}})}\BibitemShut
  {NoStop}%
\bibitem [{\citenamefont {Barron}\ \emph {et~al.}(2021)\citenamefont {Barron},
  \citenamefont {Gard}, \citenamefont {Altman}, \citenamefont {Mayhall},
  \citenamefont {Barnes},\ and\ \citenamefont
  {Economou}}]{barron2021preserving}%
  \BibitemOpen
  \bibfield  {author} {\bibinfo {author} {\bibfnamefont {G.~S.}\ \bibnamefont
  {Barron}}, \bibinfo {author} {\bibfnamefont {B.~T.}\ \bibnamefont {Gard}},
  \bibinfo {author} {\bibfnamefont {O.~J.}\ \bibnamefont {Altman}}, \bibinfo
  {author} {\bibfnamefont {N.~J.}\ \bibnamefont {Mayhall}}, \bibinfo {author}
  {\bibfnamefont {E.}~\bibnamefont {Barnes}},\ and\ \bibinfo {author}
  {\bibfnamefont {S.~E.}\ \bibnamefont {Economou}},\ }\bibfield  {title}
  {\bibinfo {title} {Preserving symmetries for variational quantum eigensolvers
  in the presence of noise},\ }\href@noop {} {\bibfield  {journal} {\bibinfo
  {journal} {Physical Review Applied}\ }\textbf {\bibinfo {volume} {16}},\
  \bibinfo {pages} {034003} (\bibinfo {year} {2021})}\BibitemShut {NoStop}%
\bibitem [{\citenamefont {Shkolnikov}\ \emph {et~al.}(2021)\citenamefont
  {Shkolnikov}, \citenamefont {Mayhall}, \citenamefont {Economou},\ and\
  \citenamefont {Barnes}}]{shkolnikov2021avoiding}%
  \BibitemOpen
  \bibfield  {author} {\bibinfo {author} {\bibfnamefont {V.}~\bibnamefont
  {Shkolnikov}}, \bibinfo {author} {\bibfnamefont {N.~J.}\ \bibnamefont
  {Mayhall}}, \bibinfo {author} {\bibfnamefont {S.~E.}\ \bibnamefont
  {Economou}},\ and\ \bibinfo {author} {\bibfnamefont {E.}~\bibnamefont
  {Barnes}},\ }\bibfield  {title} {\bibinfo {title} {Avoiding symmetry
  roadblocks and minimizing the measurement overhead of adaptive variational
  quantum eigensolvers. arxiv},\ }\href@noop {} {\bibfield  {journal} {\bibinfo
   {journal} {arXiv preprint arXiv:2109.05340}\ } (\bibinfo {year}
  {2021})}\BibitemShut {NoStop}%
\bibitem [{\citenamefont {Gard}\ \emph {et~al.}(2020)\citenamefont {Gard},
  \citenamefont {Zhu}, \citenamefont {Barron}, \citenamefont {Mayhall},
  \citenamefont {Economou},\ and\ \citenamefont
  {Barnes}}]{gard2020efficientsymmetry}%
  \BibitemOpen
  \bibfield  {author} {\bibinfo {author} {\bibfnamefont {B.~T.}\ \bibnamefont
  {Gard}}, \bibinfo {author} {\bibfnamefont {L.}~\bibnamefont {Zhu}}, \bibinfo
  {author} {\bibfnamefont {G.~S.}\ \bibnamefont {Barron}}, \bibinfo {author}
  {\bibfnamefont {N.~J.}\ \bibnamefont {Mayhall}}, \bibinfo {author}
  {\bibfnamefont {S.~E.}\ \bibnamefont {Economou}},\ and\ \bibinfo {author}
  {\bibfnamefont {E.}~\bibnamefont {Barnes}},\ }\bibfield  {title} {\bibinfo
  {title} {Efficient symmetry-preserving state preparation circuits for the
  variational quantum eigensolver algorithm},\ }\href@noop {} {\bibfield
  {journal} {\bibinfo  {journal} {npj Quantum Information}\ }\textbf {\bibinfo
  {volume} {6}},\ \bibinfo {pages} {10} (\bibinfo {year} {2020})}\BibitemShut
  {NoStop}%
\bibitem [{\citenamefont {Streif}\ \emph {et~al.}(2021)\citenamefont {Streif},
  \citenamefont {Leib}, \citenamefont {Wudarski}, \citenamefont {Rieffel},\
  and\ \citenamefont {Wang}}]{streif2021quantum}%
  \BibitemOpen
  \bibfield  {author} {\bibinfo {author} {\bibfnamefont {M.}~\bibnamefont
  {Streif}}, \bibinfo {author} {\bibfnamefont {M.}~\bibnamefont {Leib}},
  \bibinfo {author} {\bibfnamefont {F.}~\bibnamefont {Wudarski}}, \bibinfo
  {author} {\bibfnamefont {E.}~\bibnamefont {Rieffel}},\ and\ \bibinfo {author}
  {\bibfnamefont {Z.}~\bibnamefont {Wang}},\ }\bibfield  {title} {\bibinfo
  {title} {Quantum algorithms with local particle-number conservation: Noise
  effects and error correction},\ }\href@noop {} {\bibfield  {journal}
  {\bibinfo  {journal} {Physical Review A}\ }\textbf {\bibinfo {volume}
  {103}},\ \bibinfo {pages} {042412} (\bibinfo {year} {2021})}\BibitemShut
  {NoStop}%
\bibitem [{\citenamefont {Wang}\ \emph {et~al.}(2020)\citenamefont {Wang},
  \citenamefont {Rubin}, \citenamefont {Dominy},\ and\ \citenamefont
  {Rieffel}}]{wang2020x}%
  \BibitemOpen
  \bibfield  {author} {\bibinfo {author} {\bibfnamefont {Z.}~\bibnamefont
  {Wang}}, \bibinfo {author} {\bibfnamefont {N.~C.}\ \bibnamefont {Rubin}},
  \bibinfo {author} {\bibfnamefont {J.~M.}\ \bibnamefont {Dominy}},\ and\
  \bibinfo {author} {\bibfnamefont {E.~G.}\ \bibnamefont {Rieffel}},\
  }\bibfield  {title} {\bibinfo {title} {X y mixers: Analytical and numerical
  results for the quantum alternating operator ansatz},\ }\href@noop {}
  {\bibfield  {journal} {\bibinfo  {journal} {Physical Review A}\ }\textbf
  {\bibinfo {volume} {101}},\ \bibinfo {pages} {012320} (\bibinfo {year}
  {2020})}\BibitemShut {NoStop}%
\bibitem [{\citenamefont {Barkoutsos}\ \emph {et~al.}(2018)\citenamefont
  {Barkoutsos}, \citenamefont {Gonthier}, \citenamefont {Sokolov},
  \citenamefont {Moll}, \citenamefont {Salis}, \citenamefont {Fuhrer},
  \citenamefont {Ganzhorn}, \citenamefont {Egger}, \citenamefont {Troyer},
  \citenamefont {Mezzacapo} \emph {et~al.}}]{barkoutsos2018quantum}%
  \BibitemOpen
  \bibfield  {author} {\bibinfo {author} {\bibfnamefont {P.~K.}\ \bibnamefont
  {Barkoutsos}}, \bibinfo {author} {\bibfnamefont {J.~F.}\ \bibnamefont
  {Gonthier}}, \bibinfo {author} {\bibfnamefont {I.}~\bibnamefont {Sokolov}},
  \bibinfo {author} {\bibfnamefont {N.}~\bibnamefont {Moll}}, \bibinfo {author}
  {\bibfnamefont {G.}~\bibnamefont {Salis}}, \bibinfo {author} {\bibfnamefont
  {A.}~\bibnamefont {Fuhrer}}, \bibinfo {author} {\bibfnamefont
  {M.}~\bibnamefont {Ganzhorn}}, \bibinfo {author} {\bibfnamefont {D.~J.}\
  \bibnamefont {Egger}}, \bibinfo {author} {\bibfnamefont {M.}~\bibnamefont
  {Troyer}}, \bibinfo {author} {\bibfnamefont {A.}~\bibnamefont {Mezzacapo}},
  \emph {et~al.},\ }\bibfield  {title} {\bibinfo {title} {Quantum algorithms
  for electronic structure calculations: Particle-hole hamiltonian and
  optimized wave-function expansions},\ }\href@noop {} {\bibfield  {journal}
  {\bibinfo  {journal} {Physical Review A}\ }\textbf {\bibinfo {volume} {98}},\
  \bibinfo {pages} {022322} (\bibinfo {year} {2018})}\BibitemShut {NoStop}%
\bibitem [{\citenamefont {Khemani}\ \emph {et~al.}(2018)\citenamefont
  {Khemani}, \citenamefont {Vishwanath},\ and\ \citenamefont
  {Huse}}]{khemani2018operator}%
  \BibitemOpen
  \bibfield  {author} {\bibinfo {author} {\bibfnamefont {V.}~\bibnamefont
  {Khemani}}, \bibinfo {author} {\bibfnamefont {A.}~\bibnamefont
  {Vishwanath}},\ and\ \bibinfo {author} {\bibfnamefont {D.~A.}\ \bibnamefont
  {Huse}},\ }\bibfield  {title} {\bibinfo {title} {Operator spreading and the
  emergence of dissipative hydrodynamics under unitary evolution with
  conservation laws},\ }\href@noop {} {\bibfield  {journal} {\bibinfo
  {journal} {Physical Review X}\ }\textbf {\bibinfo {volume} {8}},\ \bibinfo
  {pages} {031057} (\bibinfo {year} {2018})}\BibitemShut {NoStop}%
\bibitem [{\citenamefont {Chen}\ \emph {et~al.}(2010)\citenamefont {Chen},
  \citenamefont {Gu},\ and\ \citenamefont {Wen}}]{chen2010local}%
  \BibitemOpen
  \bibfield  {author} {\bibinfo {author} {\bibfnamefont {X.}~\bibnamefont
  {Chen}}, \bibinfo {author} {\bibfnamefont {Z.-C.}\ \bibnamefont {Gu}},\ and\
  \bibinfo {author} {\bibfnamefont {X.-G.}\ \bibnamefont {Wen}},\ }\bibfield
  {title} {\bibinfo {title} {Local unitary transformation, long-range quantum
  entanglement, wave function renormalization, and topological order},\
  }\href@noop {} {\bibfield  {journal} {\bibinfo  {journal} {Physical review
  b}\ }\textbf {\bibinfo {volume} {82}},\ \bibinfo {pages} {155138} (\bibinfo
  {year} {2010})}\BibitemShut {NoStop}%
\bibitem [{\citenamefont {Chen}\ \emph {et~al.}(2011)\citenamefont {Chen},
  \citenamefont {Gu},\ and\ \citenamefont {Wen}}]{chen2011classification}%
  \BibitemOpen
  \bibfield  {author} {\bibinfo {author} {\bibfnamefont {X.}~\bibnamefont
  {Chen}}, \bibinfo {author} {\bibfnamefont {Z.-C.}\ \bibnamefont {Gu}},\ and\
  \bibinfo {author} {\bibfnamefont {X.-G.}\ \bibnamefont {Wen}},\ }\bibfield
  {title} {\bibinfo {title} {Classification of gapped symmetric phases in
  one-dimensional spin systems},\ }\href@noop {} {\bibfield  {journal}
  {\bibinfo  {journal} {Physical review b}\ }\textbf {\bibinfo {volume} {83}},\
  \bibinfo {pages} {035107} (\bibinfo {year} {2011})}\BibitemShut {NoStop}%
\bibitem [{\citenamefont {Jurdjevic}\ and\ \citenamefont
  {Sussmann}(1972)}]{jurdjevic1972control}%
  \BibitemOpen
  \bibfield  {author} {\bibinfo {author} {\bibfnamefont {V.}~\bibnamefont
  {Jurdjevic}}\ and\ \bibinfo {author} {\bibfnamefont {H.~J.}\ \bibnamefont
  {Sussmann}},\ }\bibfield  {title} {\bibinfo {title} {Control systems on lie
  groups},\ }\href@noop {} {\bibfield  {journal} {\bibinfo  {journal} {Journal
  of Differential equations}\ }\textbf {\bibinfo {volume} {12}},\ \bibinfo
  {pages} {313} (\bibinfo {year} {1972})}\BibitemShut {NoStop}%
\bibitem [{\citenamefont {Hofmann}\ and\ \citenamefont
  {Morris}(2020)}]{hofmann2020structure}%
  \BibitemOpen
  \bibfield  {author} {\bibinfo {author} {\bibfnamefont {K.}~\bibnamefont
  {Hofmann}}\ and\ \bibinfo {author} {\bibfnamefont {S.}~\bibnamefont
  {Morris}},\ }\href {https://books.google.com/books?id=tlw8EAAAQBAJ} {\emph
  {\bibinfo {title} {The Structure of Compact Groups: A Primer for the Student
  -- A Handbook for the Expert}}},\ De Gruyter Studies in Mathematics\
  (\bibinfo  {publisher} {De Gruyter},\ \bibinfo {year} {2020})\BibitemShut
  {NoStop}%
\bibitem [{\citenamefont {Georgi}(2018)}]{georgi2018lie}%
  \BibitemOpen
  \bibfield  {author} {\bibinfo {author} {\bibfnamefont {H.}~\bibnamefont
  {Georgi}},\ }\href {https://books.google.com/books?id=CUpaDwAAQBAJ} {\emph
  {\bibinfo {title} {Lie Algebras In Particle Physics: from Isospin To Unified
  Theories}}}\ (\bibinfo  {publisher} {CRC Press},\ \bibinfo {year}
  {2018})\BibitemShut {NoStop}%
\end{thebibliography}%

\onecolumngrid

\appendix

\newpage

\section*{Appendix}

\newcommand\appitem[2]{\hyperref[{#1}]{\textbf{\cref*{#1} \nameref*{#1}}} \dotfill \pageref{#1}\\ \begin{minipage}[t]{0.8\textwidth} #2\end{minipage}}
\begin{itemize}
\item \appitem{app:2local}{In this section, the group $\mathcal{V}_2^{(3)}$ of 3-qudit $\SU(d)$-invariant unitaries generated by 2-qudit symmetric gates is characterized. It is also shown that generic 3-qudit $\SU(d)$-invariant gates are semiuniversal together with 2-qudit ones.}
\item \appitem{sec:proof}{In this section, the remaining details of the proof of semi-universality of 3-qudit $\SU(d)$-invariant gates are given. In particular, \cref{thm:mulblocks} is proven, as well as Goursat's and Serre's lemmas.}
\item \appitem{app:conj-proof}{In this section, it is proven that, when $d \geq 3$, full $n$-local control is required for universality.}
\item \appitem{sec:ancillapr}{Here the result that semi-universality can be achieved with 2-local $\SU(d)$-invariant gates and $8$ ancilla qudits is proven.}
\item \appitem{App:tdesign}{In this section it is proven that 3-qudit $\SU(d)$-invariant gates are a $t$-design for all symmetric gates, with $t \approx n^2 / 2$.}
\item \appitem{sec:genSUprf}{A number of strategies useful for extending controllability are considered, including the idea of isolating Lie algebra elements to blocks.}
\end{itemize}

\newpage

\section{3-qudit gates}\label{app:2local}

In this section, we consider 2-qudit and 3-qudit $\SU(d)$-invariant gates. In \cref{app:V23char} we characterize the 3-qudit gates that can be generated by 2-qudit ones, and we prove in \cref{app:3quditgate} that the ones which cannot be generated, together with 2-qudit gates are semiuniversal. Together with the aforementioned characterization, this implies that generic $\SU(d)$-invariant 3-qudit gates are semiuniversal (\cref{prop:gen3q}).

\subsection{Which 3-qudit gates are realizable with 2-qudit gates: A characterization of \texorpdfstring{$\mathcal{V}_2^{(3)}$}{V2(3)}}\label{app:V23char}

In this section, we characterize the 3-qudit $SU(d)$-invariant gates that can be generated by 2-qudit ones, i.e. the elements of $\mathcal{V}_2^{(3)}$. In particular, because 2-qudit gates are semi-universal, i.e. $\mathcal{SV}^{(3)} \subseteq \mathcal{V}_2^{(3)}$, the only elements that cannot be generated are central and therefore act as relative phases between charge sectors.

On 
\begin{equation}
    (\complex^d)^{\otimes 3} = \hilbert_{\ysub{3}} \oplus \hilbert_{\ysub{2,1}} \oplus \hilbert_{\ysub{1,1,1}}
\end{equation}
consider the Hermitian operators
\begin{equation}\label{eq:Bs}
  \begin{split}
    B_1 & = \ident \\
    B_2 & = \P(12) + \P(13) + \P(23) \\
    B_3 & = \P(123) + \P(132).
  \end{split}
\end{equation}
These are both $\SU(d)$- and permutationally-invariant and therefore commute with each other and each element of $\mathcal{V}^{(3)}$. Furthermore, when $d \geq 3$, they are linearly independent. Thus, the group of unitaries that they generate, is equal to the center of $\mathcal{V}^{(3)}$,
\begin{equation}
    \angles{\e^{\i \alpha_1 B_1}, \e^{\i \alpha_2 B_1}, \e^{\i \alpha_3 B_3} \given \alpha_i \in \real} = \set{\e^{\i \phi_{\sysub{3}}} \Pi_{\ysub{3}} + \e^{\i \phi_{\sysub{2,1}}} \Pi_{\ysub{2,1}} + \e^{\i \phi_{\sysub{1,1,1}}} \Pi_{\ysub{1,1,1}} \given \phi_\lambda \in [0, 2 \pi)}.
\end{equation}

Using \cref{tab:n=4}, or using a character table for $\mathbb{S}_3$, it can be determined that
\begin{equation}\label{eq:Bprojs}
\begin{split}
    B_1 & = \Pi_{\ysub{3}} + \Pi_{\ysub{2,1}} + \Pi_{\ysub{1,1,1}} \\
    B_2 & = 3 \Pi_{\ysub{3}} - 3 \Pi_{\ysub{1,1,1}} \\
    B_3 & = 2 \Pi_{\ysub{3}} - \Pi_{\ysub{2,1}} + 2 \Pi_{\ysub{1,1,1}}.
\end{split}
\end{equation}
Note that, when $d = 2$, $\Pi_{\ysub{1,1,1}} = 0$, and so these three operators are linearly dependent. On the other hand, when $d \geq 3$, $\i B_3$ does not satisfy the $\mathbb{Z}_2$ condition of \cite{marvian2022quditcircuit}, so the group generated by $B_3$ is not achievable with 2-local $\SU(d)$-invariant unitaries (one can also verify by hand that the Lie algebra generated by transpositions $\i \P(jk)$ does not contain $\i B_3$). To fully characterize $\mathcal{V}_2^{(3)}$, we find the linear combination of projectors to charge sectors which is orthogonal to both $B_1$ and $B_2$ \cite{Marvian2022Restrict}.

\begin{proposition*}[re \ref{prop:2local-criteria}]
    For a system with $n=3$ qudits, the family of unitary evolutions $\exp(-\i H t): t\in\mathbb{R}$ is realizable with 2-qudit $\SU(d)$-invariant unitaries, i.e., $\exp(-\i H t)\in\mathcal{V}^{(3)}_2$, if and only if $\Tr(H C)=0$, where 
    \begin{align}\label{eq:Cop}
      C & = 2 (d - 1) (d - 2) \Pi_{\ysub{3}} - (d + 2) (d - 2) \Pi_{\ysub{2,1}} + 2 (d + 2) (d + 1) \Pi_{\ysub{1,1,1}} \\
        & = d^2(\P_{(123)} + \P_{(132)}) - 2 d (\P_{12} + \P_{13} + \P_{23}) + 4 \ident. \nonumber
    \end{align}
  Furthermore, when $d \geq 3$, the unitary $V \in \mathcal{V}^{(3)}$ is realizable by 2-qudit $\SU(d)$-invariant unitaries, i.e. $V \in \mathcal{V}_2^{(3)}$, if and only if
  \begin{equation}
      \det v_{\ysub{2,1}} = (\det v_{\ysub{3}}) (\det v_{\ysub{1,1,1}}).
  \end{equation}
  \end{proposition*}

\begin{proof}
Write $C = c_{\ysub{3}} \Pi_{\ysub{3}} + c_{\ysub{2,1}} \Pi_{\ysub{2,1}} + c_{\ysub{1,1,1}} \Pi_{\ysub{1,1,1}}$. We find $c_\lambda$, up to an overall normalization, using the conditions $\Tr C = \Tr B_1 C = 0$ and $\Tr B_2 C = 0$. Recall the decomposition into isotypic components, $\hilbert = \bigoplus_\lambda \hilbert[Q]_\lambda \otimes \hilbert[M]_\lambda$. Note that $\Tr \Pi_\lambda = d_\lambda m_\lambda$ where $d_\lambda = \dim \hilbert[Q]_\lambda$ and $m_\lambda = \dim \hilbert[M]_\lambda$.

First find that
\begin{equation}
    0 = \Tr B_2 C = 3 c_{\ysub{3}} d_{\ysub{3}} - 3 c_{\ysub{1,1,1}} d_{\ysub{1,1,1}},
\end{equation}
so $c_{\ysub{3}} d_{\ysub{3}} = c_{\ysub{1,1,1}} d_{\ysub{1,1,1}}$. Using this in
\begin{equation}
\begin{split}
    0 = \Tr C & = c_{\ysub{3}} d_{\ysub{3}} + 2 c_{\ysub{2,1}} d_{\ysub{2,1}} + c_{\ysub{1,1,1}} d_{\ysub{1,1,1}} \\
    & = 2 c_{\ysub{3}} d_{\ysub{3}} + 2 c_{\ysub{2,1}} d_{\ysub{2,1}},
\end{split}
\end{equation}
we find $c_{\ysub{2,1}} d_{\ysub{2,1}} = -c_{\ysub{3}} d_{\ysub{3}}$.

Finally, with the formulas for the dimensions,
\begin{equation}
  \begin{split}
    d_{\ysub{3}} & = \binom{d + 2}{3} \\
    d_{\ysub{2,1}} & = 2 \binom{d + 1}{3} \\
    d_{\ysub{1,1,1}} & = \binom{d}{3},
  \end{split}
\end{equation}
the first part of \cref{eq:Cop} is verified. The second part can be checked with \cref{eq:Bs,eq:Bprojs}.

The Hamiltonians $\ident = B_1$ and $B_2$ of \cref{eq:Bs} generate the center of $\mathcal{V}_2^{(3)}$, since $B_2$ is a sum of 2-local terms. According to \cref{eq:Bprojs},
\begin{equation}\label{eq:phase}
    \e^{\i \alpha_1 \ident} \e^{\i \alpha_2 B_2 / 3} = \e^{\i(\alpha_1 + \alpha_2)} \Pi_{\ysub{3}} + \e^{\i \alpha_1} \Pi_{\ysub{2,1}} + \e^{\i(\alpha_1 - \alpha_2)} \Pi_{\ysub{1,1,1}}.
\end{equation}
Each choice of values $\alpha_1, \alpha_2 \in [0, 2 \pi)$ corresponds to a unique operator in $\mathcal{V}_2^{(3)}$, i.e. this parameterization is one-to-one, and every element of the center of $\mathcal{V}_2^{(3)}$ is of the form \cref{eq:phase} for some $\alpha_1, \alpha_2$. Note that, when $\alpha_1 = \alpha_2 = \pi$, \cref{eq:phase} is actually inside $\mathcal{SV}^{(3)}$, since $-\ident_{\hilbert[M]_{\sysub{2,1}}} \in \SU(\hilbert[M]_{\ysub{2,1}})$. This is the only point of intersection, since any other choice of phases will not have determinant one on all multiplicity subsystems $\hilbert[M]_\lambda$.

Conversely, given an arbitrary unitary acting as relative phases, where $\phi_\lambda \in [0, 2 \pi)$,
\begin{equation}
    \e^{\i \phi_{\sysub{3}}} \Pi_{\ysub{3}} + \e^{\i \phi_{\sysub{2,1}}} \Pi_{\ysub{2,1}} + \e^{\i \phi_{\sysub{1,1,1}}} \Pi_{\ysub{1,1,1}},
\end{equation}
it is of the form \cref{eq:phase} if and only if there are $\alpha_1, \alpha_2 \in [0, 2 \pi)$ such that
\begin{equation}
\begin{split}
    \alpha_1 + \alpha_2 & = \phi_{\ysub{3}} \pmod{2 \pi} \\
    \alpha_1 & = \phi_{\ysub{2, 1}} \\
    \alpha_1 - \alpha_2 & = \phi_{\ysub{1,1,1}} \pmod{2 \pi}.
\end{split}
\end{equation}
Setting $\alpha_1 = \phi_{\ysub{2,1}}$, a phase $\alpha_2 \in [0, 2 \pi)$ satisfying these equations exists if and only if $\phi_{\ysub{3}} - \phi_{\ysub{2,1}} = \phi_{\ysub{2,1}} - \phi_{\ysub{1,1,1}} \pmod{2 \pi}$, or, in other words,
\begin{equation}\label{eq:phidet}
    \e^{\i 2 \phi_{\sysub{2,1}}} = \e^{\i \phi_{\sysub{3}}} \e^{\i \phi_{\sysub{1,1,1}}},
\end{equation}
i.e. $2 \phi_{\ysub{2,1}} = \phi_{\ysub{3}} + \phi_{\ysub{1,1,1}} \pmod{2 \pi}$.

Now consider an arbitrary element of $\mathcal{V}^{(3)}$:
\begin{equation}
    V = (\ident_{\hilbert[Q]_{\sysub{3}}} \otimes v_{\ysub{3}}) \oplus (\ident_{\hilbert[Q]_{\sysub{2,1}}} \otimes v_{\ysub{2,1}}) \oplus (\ident_{\hilbert[Q]_{\sysub{1,1,1}}} \otimes v_{\ysub{1,1,1}}).
\end{equation}
How can we check if $V \in \mathcal{V}_2^{(3)}$? The first thing to note is that, since $\mathcal{V}_2^{(3)}$ is semiuniversal, for any $V \in \mathcal{V}^{(3)}$ there is some $\widetilde{V} \in \mathcal{SV}^{(3)} \subseteq \mathcal{V}_2^{(3)}$ such that $V \widetilde{V}$ acts as phases,
\begin{equation}
    V \widetilde{V} = \e^{\i \phi_{\sysub{3}}} \Pi_{\ysub{3}} + \e^{\i \phi_{\sysub{2,1}}} \Pi_{\ysub{2,1}} + \e^{\i \phi_{\sysub{1,1,1}}} \Pi_{\ysub{1,1,1}}.
\end{equation}
Since $\widetilde{V} \in \mathcal{V}_2^{(3)}$, $V \in \mathcal{V}_2^{(3)}$ if and only if $V \widetilde{V} \in \mathcal{V}_2^{(3)}$. Furthermore, since $\det \widetilde{v}_\lambda = 1$ for each $\lambda = \sydiag{3}$, $\sydiag{2,1}$, $\sydiag{1,1,1}$, it also holds that
\begin{equation}
    \det v_\lambda = \det v_\lambda \widetilde{v}_\lambda = \e^{\i \phi_{\lambda} m_\lambda},
\end{equation}
where $m_\lambda = \dim \hilbert[M]_\lambda$. Since $V \widetilde{V} \in \mathcal{V}_2^{(3)}$ if and only if $2 \phi_{\ysub{2,1}} = \phi_{\ysub{3}} + \phi_{\ysub{1,1,1}} \pmod{2 \pi}$, it follows that $V \in \mathcal{V}_2^{(3)}$ if and only if
\begin{equation}
    \det v_{\ysub{2,1}} = (\det v_{\ysub{3}}) (\det v_{\ysub{1,1,1}}).
\end{equation}
\end{proof}

Note that, since $\mathcal{V}_2^{(3)}$ is semiuniversal, i.e. it contains $\mathcal{SV}^{(3)} \cong \SU(2)$ which acts nontrivially only in the 2D multiplicity subsystem $\hilbert[M]_{\ysub{2,1}}$ (the other two are 1D), the discussion after \cref{eq:phase} implies
\begin{equation}
    \mathcal{V}_2^{(3)} \cong \U(1) \times \U(2)
\end{equation}
since the $\U(1) \times \U(1)$ subgroup parameterized by \cref{eq:phase} intersects $\mathcal{SV}^{(3)}$ only in the element which acts on $\hilbert[M]_{\ysub{2,1}}$ as $-\ident$. (When $d = 2$ this isomorphism is obvious since then $\mathcal{V}_2^{(3)} \cong \mathcal{V}^{(3)}$.)

\subsection{Generic 3-qudit gates are semi-universal}\label{app:3quditgate}

Here we show \cref{4qudit}, that any 4-qudit gate which does not satisfy the condition
\begin{equation*}
    [J] [\pi_{\ysub{3, 1}}(Y)] [J]^{\text{T}} \neq \e^{\i\phi}[\pi_{\ysub{2, 1, 1}}(Y)]^\ast\ ,
  \end{equation*}
is enough to achieve semi-universality with, along with 2-qudit gates. Note that this condition is related to the $\mathbb{Z}_2$ condition of \cite{marvian2022quditcircuit}.

\begin{proof}[Proof of \cref{4qudit}]
  One can easily see the necessity of this condition: All unitaries in $\mathcal{V}^{(4)}_2$ satisfy the constraint in \cref{const}. If $Y$ also satisfies this constraint, then the group generated by them does as well, which means it cannot contain $\mathcal{SV}^{(4)}$ {(Note that if unitaries $V_1$ and $V_2$ satisfy \cref{const}, then their inverses and their product also satisfies this constraint, i.e., this condition defines a subgroup)}.

  Next, to see the sufficiency of the condition in \cref{cond1}, first recall that the projection of $\mathcal{V}^{(4)}_2$ to each of the multiplicity spaces $\mathcal{M}_{\ysub{3, 1}}$, $\mathcal{M}_{\ysub{2,2}}$, and $\mathcal{M}_{\ysub{2, 1, 1}}$, is equal to the full unitary group in these spaces, 
  which means condition \textbf{A} of \cref{MainLemma} is satisfied. 
  In this case condition \textbf{B} is only relevant for the pair $\lambda=\sydiag{2, 1, 1}$ and $\lambda=\sydiag{3, 1}$, which both have dimension $3$.

  Furthermore, for all such unitaries the components in $\mathcal{M}_{\ysub{3, 1}}$ and $\mathcal{M}_{\ysub{2, 1, 1}}$ satisfy the constraint in \cref{const}. Therefore, there exists a unitary $\widetilde{Y}\in \mathcal{V}^{(4)}_2$, such that $\pi_{\ysub{2, 1, 1}}(\widetilde{Y}) = \pi_{\ysub{2, 1, 1}}(Y)$. Thus $T=Y\widetilde{Y}^\dag$ acts trivially on $\sydiag{2, 1, 1}$, and inside $\sydiag{3, 1}$ acts non-trivially, i.e., it is not proportional to a global phase, such that 
    $\pi_{\ysub{3, 1}}(Y \widetilde{Y}^\dag) \neq \e^{\i\phi}\mathbb{I}_{\ysub{3, 1}}$ for all phases $\e^{\i\phi}$. 
    This together with \cref{ind} implies that condition \textbf{B}' in \cref{MainLemma} is also satisfied. Therefore, \cref{MainLemma} implies that $\mathcal{SV}^{(4)}\subset \langle Y, V: V\in \mathcal{V}_2^{(4)}\rangle$.
\end{proof}

We now prove \cref{prop:gen3q}, that for any 3-qudit $\SU(d)$-invariant unitary $V$ which cannot be generated by 2-qudit ones, the unitary $Y = V \otimes \ident$ satisfies the condition of \cref{4qudit}. In particular, due to \cref{prop:2local-criteria}, generic 3-qudit $\SU(d)$-invariant unitaries are semi-universal when combined with 2-qudit ones.

\begin{proof}[Proof of \cref{prop:gen3q}] To see this first we consider unitaries of the form
\be\label{eq:Sphases}
\widetilde{S}=\e^{\i \phi_{\sysub{3}} } \Pi_{\ysub{3}}+
\e^{\i \phi_{\sysub{2, 1}} } \Pi_{\ysub{2, 1}}+\e^{\i \phi_{\sysub{1, 1, 1}} } \Pi_{\ysub{1, 1, 1}}\ .
\ee
Then, 
for $Y=\widetilde{S}\otimes \mathbb{I}$, \cref{eq:trcond} is equivalent to
\be\label{eq:trcircle}
\abs[\big]{\exp\i(\phi_{\ysub{2, 1}}-\phi_{\ysub{3}})+\frac{1}{2}}\neq \abs[\big]{\exp\i(\phi_{\ysub{2, 1}}-\phi_{\ysub{1,1,1}})+\frac{1}{2}}\ .
\ee
This can only hold as equality if $\exp \i(\phi_{\ysub{2,1}} - \phi_{\ysub{3}})$ and $\exp \i(\phi_{\ysub{2,1}} - \phi_{\ysub{1,1,1}})$ are either equal or related by complex conjugate (see the diagram \cref{fig:geo-proof} for a proof). As seen in \cref{eq:phidet}, a unitary of the form $\widetilde{S}$ is an element of $\mathcal{V}_2^{(3)}$ if and only if they are complex conjugates, i.e.,
\begin{equation}\label{eq:comcon}
  \exp\i(\phi_{\ysub{2, 1}}-\phi_{\ysub{3}}) = \exp\i(\phi_{\ysub{2, 1}}-\phi_{\ysub{1,1,1}})^\ast.
\end{equation}

Now recall that, for any $V \in \mathcal{V}^{(3)}$ there exists $\widetilde{V} \in \mathcal{SV}^{(3)} \subseteq \mathcal{V}_2^{(3)}$ such that $V \widetilde{V}$ has the form \cref{eq:Sphases} and
\begin{equation}
\begin{split}
    \phi_{\ysub{3}} & = \arg \det v_{\ysub{3}} \\
    \phi_{\ysub{2,1}} & = \arg \det v_{\ysub{2,1}} + b \pi \pmod{2 \pi} \\
    \phi_{\ysub{1,1,1}} & = \arg \det v_{\ysub{1,1,1}}
\end{split}
\end{equation}
where $b = 0, 1$ is not specified by $V$. But, since $\e^{\i \pi} = \e^{-\i \pi}$, the value of $b$ does not affect the validity of \cref{eq:comcon}. In other words, the condition \cref{eq:comcon} for $V \widetilde{V}$ is determined by
\begin{equation}
    \det v_{\ysub{2,1}} \overset{?}{=} (\det v_{\ysub{3}}) (\det v_{\ysub{1,1,1}}),
\end{equation}
which is precisely the condition \cref{eq:Vdet} for $V \in \mathcal{V}_2^{(3)}$.
\end{proof}

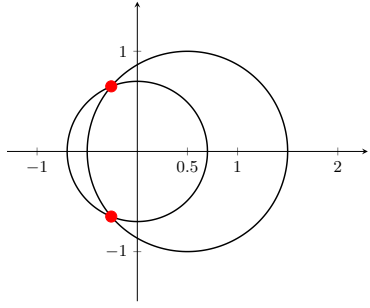
\begin{figure}[htb]
  \centering
  \begin{tikzpicture}[scale=0.7]
    \begin{axis}[
      axis lines = center,
      xmin=-1, xmax=2,
      ymin=-1, ymax=1,
      xtick={-1, 0, 0.5, 1, 2},
      ytick={-1, 0, 1},
      enlargelimits=true,
      axis equal
      ]
      
      \addplot[domain=0:2*pi, samples=100, smooth, thick] ({0.5 + cos(deg(x))}, {sin(deg(x))});
      \addplot[domain=0:2*pi, samples=100, smooth, thick] ({0.7 * cos(deg(x))}, {0.7 * sin(deg(x))});

      \addplot[only marks, mark=*, mark options={red, scale=1.5}] coordinates {(-0.26, 0.65)};
      \addplot[only marks, mark=*, mark options={red, scale=1.5}] coordinates {(-0.26, -0.65)};
      
      
    \end{axis}
  \end{tikzpicture}
  \caption{ A geometric description of the condition in 
  \cref{eq:trcircle}. This equation holds as equality if, and only if, $\exp \i(\phi_{\ysub{2,1}} - \phi_{\ysub{3}})$ and $\exp \i(\phi_{\ysub{2,1}} - \phi_{\ysub{1,1,1}})$ are either equal or complex conjugates. Both $\exp \i(\phi_{\ysub{2,1}}-\phi_{\ysub{3}}) + 1/2$ and $\exp \i(\phi_{\ysub{2,1}} - \phi_{\ysub{1,1,1}}) + 1/2$ sit on the translated unit circle centered at 1/2. \Cref{eq:trcircle} holds as equality only when these points are equidistant from the origin, i.e. lie on the same circle centered at the origin. The intersection of any two such circles consists of at most two distinct points, which are complex conjugates of each other. Thus $|\exp \i(\phi_{\ysub{2,1}}-\phi_{\ysub{3}}) + 1/2| = |\exp \i(\phi_{\ysub{2,1}} - \phi_{\ysub{1,1,1}}) + 1/2|$ if and only if they are either equal or complex conjugates, and since $1/2$ is real, the claim holds.}
  \label{fig:geo-proof}
\end{figure}

\clearpage

\section{Tools for semiuniversality}\label{sec:proof}

In this section,  we prove \cref{thm:mulblocks,lem:twoblock2}, which provide guarantees on controllability, and we prove Goursat's lemma for special unitary groups, \cref{lem:SUgoursat}, as a corollary of the asymmetric version of the original Goursat's lemma, \cref{lem:agoursat}. We also provide a proof of Serre's lemma, \cref{lem:serre}. Together with the argument in \cref{sec:over}, this completes the proof of semi-universality of 3-qudit $\SU(d)$-invariant unitaries.


\subsection{Extending controllability from a single block}\label{sec:lem2}

In this section, we prove theorems that provide guarantees on controllability. In particular, we demonstrate that irreducibility resulting from continuous families of unitaries together with particular subgroups of unitaries acting only on subspaces of the total Hilbert space $\hilbert$, is enough to achieve full controllability, i.e. all unitaries up to a global phase, $\SU(\hilbert)$.

In the following, we will often consider a Hilbert space with orthogonal decomposition $\hilbert = \bigoplus_{i = 1}^m \hilbert_i$. With respect to such a decomposition, we consider $\U(\hilbert_i) \subseteq \U(\hilbert)$ the subgroup of unitaries which act as the identity on $\hilbert_j$ for $j \neq i$. We also use the fact that, for \emph{any} matrix group $\mathcal{W} \subseteq \GL(\hilbert)$, the set of linear operators
\begin{equation}\label{eq:lie}
  \lie{w} = \set{X \in \hilbert[L](\hilbert) \given \forall t \in \real, \e^{t X} \in \mathcal{W}}
\end{equation}
is a Lie algebra \cite{jurdjevic1972control}, closed under real linear combination and commutator.\footnote{ While $\mathcal{W}$ may not be a Lie group itself, it at least contains a (possibly trivial) connected Lie group, namely the group generated by exponentials of its Lie algebra, $\langle \e^{\lie{w}} \rangle \subseteq \mathcal{W}$. In any case, we call $\lie{w}$ the Lie algebra of $\mathcal{W}$.}


\begin{lemma*}[re \ref{thm:mulblocks}]
  Let $\hilbert$ be a finite-dimensional Hilbert space with a subspace $\hilbert_1 \subset \hilbert$ such that $\dim \hilbert_1 \geq 3$. Let $A_i$, $i = 1, \dots, k$, be traceless anti-Hermitian operators on $\hilbert$ and consider the one-parameter groups $\mathcal{A}_i = \set{\e^{t A_i} \given t \in \real}$. If the group
  \begin{equation}
    \mathcal{W} = \angles{\mathcal{A}_i, \SU(\hilbert_1) \, \given \, i = 1, \dots, k}
  \end{equation}
  acts irreducibly on $\hilbert$, then $\mathcal{W} = \SU(\hilbert)$.
\end{lemma*}

\begin{remark}
  Let $\lie{g} \subseteq \su(\hilbert)$ be a matrix Lie algebra acting irreducibly on a finite-dimensional Hilbert space $\hilbert$. \Cref{thm:mulblocks} is equivalent to the statement that, if there is a subspace $\hilbert_1 \subseteq \hilbert$ such that $\su(\hilbert_1) \subseteq \lie{g}$ and $\dim \hilbert_1 \geq 3$, then $\lie{g} = \su(\hilbert)$.
\end{remark}

\begin{remark}\label{rm:2dim}
  The assumption that $\dim \hilbert_1 \geq 3$ can be relaxed to $\dim \hilbert_1 \geq 2$ if the possibility that $\mathcal{W} = \Sp(\hilbert)$ is allowed, when $\dim \hilbert$ is even (see \cite{liu2024control} for further details).
\end{remark}

For the proof of \cref{thm:mulblocks}, we will also need \cref{lem:iso-oneblock,lem:twoblockoff}: the first, which constructs an off-block-diagonal Lie algebra element, is proven in a number of ways \cref{sec:strat}; the second, which describes a generating set for the Lie algebra of determinant-one unitaries, is postponed until after the proof of \cref{thm:mulblocks}.

\begin{proof}[Proof of \cref{thm:mulblocks}]
  Suppose that $\hilbert' \subset \hilbert$ is any proper subspace of $\mathcal{H}$ which contains $\hilbert_1 \subseteq \hilbert'$ and such that $\SU(\hilbert') \subseteq \mathcal{W}$. We prove that there is a subspace $\hilbert'' \subseteq \hilbert$ properly containing $\hilbert' \subset \hilbert''$ and such that $\SU(\hilbert'') \subseteq \mathcal{W}$. In summary,
    \be
    \mathcal{H}_1\subseteq \mathcal{H}'\subset \mathcal{H}''\subseteq \mathcal{H} \ .
    \ee 
  
  Let $\Pi$ be the Hermitian projector to $\hilbert'$ and $\Pi_\perp$ the projector to its orthogonal complement. There necessarily exists some $A_i$, $i = 1, \dots, k$, such that $\Pi_\perp A_i \Pi \neq 0$, since otherwise $\mathcal{W}$ would act block-diagonally between $\hilbert'$ and its orthogonal complement, i.e. it would not act irreducibly on $\hilbert$.  In the following we denote this operator $A_i$ with $A$.

  Consider the non-zero subspace $\mathcal{F}=\Pi_\perp A \hilbert'$, that is the image of $\hilbert'$ under $A$, projected to $\Pi_\perp$.
  Let
  \begin{equation}
    \hilbert'' = \hilbert'\oplus \mathcal{F}= \hilbert' \oplus \Pi_\perp A \hilbert'\ .
  \end{equation}
  In the following, we will show that for any arbitrary pair of vectors $|\phi\rangle\in \hilbert'$ and $|\psi\rangle\in \mathcal{F}$, the Lie algebra $\mathfrak{w}$ contains $|\psi\rangle\langle \phi|-|\phi\rangle\langle\psi|$.

  First, note that by definition, the linear map $\tilde{A}=\Pi_\perp A \Pi$ has full support on subspace $\mathcal{F}=\Pi_\perp A \hilbert'$, which means for any vector $|\psi\rangle\in \mathcal{F}$, there exists a vector $|{\psi}'\rangle$ in $\mathcal{H}'$, such that
  \be
  |\psi\rangle=\tilde{A}|{\psi}'\rangle=\Pi_\perp A |\psi'\rangle\ . 
  \ee
  Next, we use \cref{lem:iso-oneblock}. This lemma implies that,  because $\dim(\mathcal{H}')\ge 3$ and $\SU(\mathcal{H}')\subseteq\mathcal{W}$, then    
 for any $B \in \hilbert[L](\hilbert')$ (i.e. $B = \Pi B \Pi$),
  $ {B^\dag} A \Pi_\perp + \Pi_\perp A B \in \lie{w}$.

  Suppose we choose $B=|\psi'\rangle\langle\phi|$, where $|\phi\rangle$ is an arbitrary vector in $\mathcal{H}'$. In this case, we have 
  \be
  \Pi_\perp A B=\Pi_\perp A |\psi'\rangle\langle\phi|= |\psi\rangle\langle\phi|\ .
  \ee
  Since $A$ is anti-Hermitian $B^\dag A \Pi_\perp =-( \Pi_\perp A B)^\dag=-|\phi\rangle\langle\psi|$. Therefore, \cref{lem:iso-oneblock} implies
  \be
  {B^\dag} A \Pi_\perp + \Pi_\perp A B= |\psi\rangle\langle\phi|-|\phi\rangle\langle\psi| \in \mathfrak{w}\ .
  \ee
  Since this holds for arbitrary $|\phi\rangle\in \hilbert'$ and $|\psi\rangle\in \mathcal{F}$, then applying \cref{lem:twoblockoff} to $\mathcal{H}''=\mathcal{H}'\oplus \mathcal{F}$, we conclude that $\su(\hilbert'') \subseteq \lie{w}$, and, by exponentiation, $\SU(\hilbert'') \subseteq \mathcal{W}$.

  Proceeding inductively, because $\hilbert$ is finite-dimensional and $\SU(\hilbert') \subseteq \mathcal{W}$, it follows that $\SU(\hilbert) \subseteq \mathcal{W}$.
\end{proof}

Now we state and prove \cref{lem:twoblockoff}.

\begin{lemma}\label{lem:twoblockoff}
  Consider a finite-dimensional Hilbert space with orthogonal decomposition $\hilbert = \hilbert_1 \oplus \hilbert_2$ and respective dimensions $d = d_1 + d_2$, with $d_1 \geq d_2 \geq 1$. The (real) Lie algebra $\lie{g}$ generated by all $\qout{\psi_1}{\psi_2} - \qout{\psi_2}{\psi_1}$, where $\ket{\psi_i} \in \hilbert_i$, is equal to $\su(\hilbert)$.
\end{lemma}

\begin{proof}
  Let $\ket{1}, \dots, \ket{d}$ be an orthonormal basis for $\hilbert$ with $\ket{1}, \dots, \ket{d_1} \in \hilbert_1$ and $\ket{d_1 + 1}, \dots, \ket{d_1 + d_2} \in \hilbert_2$. 
  To prove the claim it suffices to show that for any $j \neq k$, both 
  \be
  {Y}_{jk}=|j\rangle\langle k|-|k\rangle\langle j|
  \ee
  and 
  \be
  {X}_{jk}= \i (|j\rangle\langle k|+|k\rangle\langle j|)
  \ee
  are in the Lie algebra $\mathfrak{g}$, which implies 
  \be
  Z_{jk}=\i (|j\rangle\langle j|-|k\rangle\langle k|)=\frac{{1}}{2}[{X}_{jk}, {Y}_{jk}]
  \ee
  is also in the Lie algebra. Since operators ${X}_{jk}, {Y}_{jk}, {Z}_{jk}: 1\le j <k \le d$ form a basis for $\su(\hilbert)$, this proves the claim. 

  First, not that for any $1\le j\le d_1$ and $d_1+1\le k\le d$, the assumption of lemma implies that $Y_{jk}\in\mathfrak{g}$. Furthermore, choosing $|\psi_1\rangle=\i|j\rangle$ and 
  $|\psi_2\rangle=|k\rangle$, the assumption also implies $X_{jk}\in\mathfrak{g}$. This way we obtain all the elements of the above basis that have support in both subspaces $\mathcal{H}_1$ and $\mathcal{H}_2$. Next, we show how we can obtain elements with support restricted to one of these subspaces. We use the fact that for any 3 distinct $j , k , k'\in\{1,\cdots, d\}$, we have the commutation relations
  \be
  [Y_{jk}, Y_{kj'}]= Y_{jj'}\ ,
  \ee
  and 
  \be
  [X_{jk}, Y_{kj'}]=  X_{jj'}\ .
  \ee
  This completes the proof.
\end{proof}

\subsection{Two blocks}\label{app:twoblocks}

\Cref{thm:mulblocks} proves that irreducibility allows for extending controllability from three-dimensional blocks to the entire Hilbert space. For the proof of \cref{thm:semi-universality} it is also required that controllability can be extended from a two-dimensional block to the entire Hilbert space when $\dim \hilbert = 3$ (note that this holds more generally when $\dim \hilbert$ is odd, see \cite{liu2024control}). That is, we need to extend
\begin{align}
  \begin{+pmatrix}
    \SetCell[r=2,c=2]{c} \SU(2) & & \vline \\
    & & \\ \hline
    & & 1
  \end{+pmatrix} .
\end{align}
to the entire $\SU(3)$. Note that here we do not need to explicitly assume irreducibility on the entire Hilbert space: instead, the fact that the one-parameter group is not block-diagonal immediately implies it. (Note that the following lemma is actually more general than described above, but in the case of $\dim \hilbert_\alpha = 2$ and $\dim \hilbert_\beta = 1$ it is equivalent).  This lemma is also previously presented in \cite{marvian2022quditcircuit}.\footnote{The proof presented below was included in the first arXiv version of \cite{marvian2022quditcircuit}.}

\begin{lemma}(\cite{marvian2022quditcircuit})\label{lem:twoblock2}
  Consider a finite-dimensional Hilbert space with an orthogonal decomposition $\hilbert = \hilbert_\alpha \oplus \hilbert_\beta$ such that neither $\dim \hilbert_\alpha = \dim \hilbert_\beta = 2$ or $\dim \hilbert_\alpha = \dim \hilbert_\beta = 1$. The subgroup $\SU(\hilbert_\alpha) \subseteq \SU(\hilbert)$ acts as the identity on $\hilbert_\beta$, and vice-versa. Let $A$ be a Hermitian operator on $\hilbert$ which is not block-diagonal with respect to the above decomposition. Then the group
  \begin{equation}
    \mathcal{W} = \angles{\e^{\i t A}, \SU(\hilbert_\alpha), \SU(\hilbert_\beta) \given t \in \real},
  \end{equation}
  generated by the unitaries $\e^{\i t A}$ and each $\SU(\hilbert_i)$ contains all determinant-one unitaries, $\mathcal{W} \supseteq \SU(\hilbert)$.
\end{lemma}

\begin{remark}
  In the case of $\dim \hilbert_\alpha = \dim \hilbert_\beta = 1$, if we instead consider the subgroups of relative phases, i.e. $\U(\hilbert_i)$, the theorem holds.
\end{remark}
\begin{remark}\label{rmk:sp(2)}
  In the case of $\dim \hilbert_\alpha = \dim \hilbert_\beta = 2$, this result can fail in an interesting way. Choose a basis $\cset{\nket{m} \otimes \nket{n}}{m, n = 0, 1}$ for this space so that $\hilbert_\alpha = \operatorname{span}_\complex \sets{\nket{00}, \nket{01}}$ and $\hilbert_\beta = \operatorname{span}_\complex \sets{\nket{10}, \nket{11}}$. Then, for example, $A = \sigma_x \otimes \ident$ is not block-diagonal with respect to the decomposition $\hilbert_\alpha \oplus \hilbert_\beta$. However, $\cset{\e^{\i A t}}{t \in \real}$ together with $\SU(\hilbert_i)$ generate a subgroup $\Sp(2) \subseteq \SU(4)$.
\end{remark}

Note the change in notation from \cref{thm:mulblocks} to \cref{lem:twoblock2}: $A$ here is Hermitian rather than anti-Hermitian. Let $D_i = \dim \hilbert_i$ and $\lie{g} = \lie{su}(\hilbert_\alpha) \oplus \lie{su}(\hilbert_\beta)$.

\begin{proof}
  First, we prove the result in the special case where the Hermitian operator $A$ is in the form
  \be
  A=|\Theta_\alpha\rangle\langle \Theta_\beta|+|\Theta_\beta\rangle\langle \Theta_\alpha|\ ,
  \ee
  where $|\Theta_{\alpha,\beta}\rangle$ are normalized vectors in $\mathcal{H}_{\alpha, \beta}$, and then explain how the proof can be generalized. 

  Since $\hilbert_\alpha$ and $\hilbert_\beta$ play equivalent roles in the proof, without loss of generality we assume $D_\alpha \geq D_{\beta}$. In particular, the assumption of the lemma is that we cannot have $D_{\alpha} = D_\beta = 1$ or $D_{\alpha} = D_\beta = 2$, and so $D_\alpha \geq 2$, and when $D_\beta = 2$ we also have $D_\alpha \geq 3$. Let $\{|l,\alpha\rangle: l=1,\cdots, D_\alpha \}$ be an orthonormal basis for $\mathcal{H}_\alpha$ and $\{|m,\beta\rangle: m=1,\cdots, D_\beta \}$ be an orthonormal basis for $\mathcal{H}_\beta$, with the property that they contain $\nket{\Theta_{\alpha, \beta}}$, so that there exist $l_0, m_0$ such that $|l_0,\alpha\rangle=|\Theta_\alpha\rangle$ and $|m_0,\beta\rangle=|\Theta_\beta\rangle$, i.e., 
  \be\label{trqew}
  \begin{split}
    \i A& =\i(|\Theta_\alpha\rangle\langle \Theta_\beta|+|\Theta_\beta\rangle\langle \Theta_\alpha|)\\
        & =\i(|l_0,\alpha\rangle\langle m_0,\beta|+ |m_0,\beta\rangle\langle l_0,\alpha|)\\
        &\equiv X_{l_0 m_0}\ .
  \end{split}
  \ee 
  It can be easily shown that the operator $X_{l_0m_0}=\i A$ together with $\mathfrak{g}$ generates the full $\mathfrak{su}(D)$. In particular, note that for any $ l=1,\cdots, D_\alpha$ with $l \neq l_0$, the commutator of $X_{l_0m_0}=\i A$ with the traceless skew-Hermitian operator $|l_0,\alpha\rangle\langle l,\alpha|- |l,\alpha\rangle\langle l_0,\alpha|\in \mathfrak{g}$, is 
  \be\label{b1}
  \begin{split}
    X_{lm_0}&=\big[X_{l_0m_0}, |l_0,\alpha\rangle\langle l,\alpha|- |l,\alpha\rangle\langle l_0,\alpha| \big]\\
            &=\i(|l,\alpha\rangle\langle m_0,\beta|+ |m_0,\beta\rangle\langle l,\alpha|)\ .
  \end{split}
  \ee
  Furthermore, the commutator of this operator with the traceless skew Hermitian operator $\i(|l_0,\alpha\rangle\langle l_0,\alpha|- |l,\alpha\rangle\langle l,\alpha| )\in \mathfrak{g}$ is equal to
  \begin{align}\label{b2}
    \begin{split}
      Y_{l m_0}&\equiv\big[X_{lm_0}, \i(|l_0,\alpha\rangle\langle l_0,\alpha|- |l,\alpha\rangle\langle l,\alpha| ) \big]\\
               &= |m_0,\beta\rangle\langle l,\alpha|-|l,\alpha\rangle\langle m_0,\beta|\ .
    \end{split}
  \end{align}
  And, 
  \begin{align}
    Z_{lm_0}\equiv-\frac{1}{2}[X_{lm_0},Y_{lm_0}]&= \i(|m_0,\beta\rangle\langle m_0,\beta|-|l,\alpha\rangle\langle l,\alpha|)\ .
  \end{align}
  Therefore, for all $l=1,\cdots, D_\alpha$, operators $X_{lm_0}, Y_{lm_0}, Z_{lm_0}$ are in the Lie algebra $\langle \mathfrak{g}, \i A\rangle$. If $D_\beta=1$, then the linear combinations of these operators with traceless skew-Hermitian operators with support restricted to $\mathcal{H}_\alpha$ yield all traceless skew Hermitian operators on $\mathcal{H}_\alpha\oplus \mathcal{H}_\beta$. Therefore, in this special case we conclude $\mathfrak{su}(D)\subseteq \langle \mathfrak{g}, \{ \i A \} \rangle$, which proves the lemma.

  On the other hand, if $D_\beta>1$, then for any arbitrary $m\in\{1,\cdots, D_\beta\}$ with $m\neq m_0$, we define
  \be
  \begin{split}
    Y_{l m}&=\big[Y_{l m_0}, |m,\beta\rangle\langle m_0,\beta|- |m_0,\beta\rangle\langle m,\beta| \big]\\
           &= |l,\alpha\rangle\langle m,\beta|-|m,\beta\rangle\langle l,\alpha| \ .
  \end{split}
  \ee
  Similarly, considering the commutator of $Y_{l m}$ with operator with $\i(|m,\beta\rangle\langle m,\beta|- |m_0,\beta\rangle\langle m_0,\beta|)\in\mathfrak{g}$ we obtain operators $X_{lm}=\i( |l,\alpha\rangle\langle m,\beta|+|m,\beta\rangle\langle l,\alpha|)$, and from the commutator of $X_{lm}$ with $Y_{lm}$ we obtain $Z_{lm}=\frac{1}{2}[X_{lm}, Y_{lm}]=\i(|l,\alpha\rangle\langle l,\alpha|-|m,\beta\rangle\langle m,\beta|)\ $. Then, the 
  linear combination of operators
  \be
  X_{lm}, Y_{lm}, Z_{lm}:\ \ m=1,\cdots, D_\beta\ ; \ \ l=1,\cdots, D_\alpha\ , 
  \ee
  together with operators in $\mathfrak{g}$ contains all skew-Hermitian traceless operators, which implies $\mathfrak{su}(D)\subseteq \langle \{iA\}, \mathfrak{g}\rangle$. \\

  The above argument proves the lemma in the special case where $A=|\Theta_\alpha\rangle\langle \Theta_\beta|+|\Theta_\beta\rangle\langle \Theta_\alpha| $. To complete the proof, next we show that if $A$ is not block-diagonal with respect to $\mathcal{H}_\alpha\oplus \mathcal{H}_\beta$ then the Lie algebra generated by $\mathfrak{g}$ and $\{iA\}$ contains an operator in the form 
  $\i(|\Theta_\alpha\rangle\langle \Theta_\beta|+|\Theta_\beta\rangle\langle \Theta_\alpha|)$. Therefore, by applying the above argument this proves the lemma in the general case.

  To prove this, we consider the cases of $D_\beta=1$ and $D_\beta\ge 2$ separately. \\

  \noindent\textbf{The case of $D_\beta =1$:}
  Let $|\Theta_\beta\rangle\in \mathcal{H}_\beta$ be a normalized vector. Any Hermitian operator $A$ can be written as
  \be
  A=A_\alpha+ a |\Theta_\beta\rangle\langle \Theta_\beta|+ b |\Gamma\rangle\langle \Theta_\beta|+ b^\ast |\Theta_\beta\rangle\langle \Gamma | \ .
  \ee
  where $A_\alpha = A_\alpha^\dagger$ has support restricted to $\mathcal{H}_\alpha$, $a$ is real, $|\Gamma\rangle\in \mathcal{H}_\alpha $ is a normalized vector and the assumption that $A$ is not block-diagonal implies $b\neq 0$. 

  Let $\Pi_\alpha$ be the projector to $\mathcal{H}_\alpha$ and define the operator 
  \begin{align}\label{eq:Eop}
    E\equiv |\Gamma\rangle\langle\Gamma|-D_\alpha^{-1} {\Pi_\alpha}\ ,
  \end{align}
  which is nonzero because $D_\alpha > 1$. Note that $\i E$ is an element of $\mathfrak{g}$. Next, note that 
  \begin{equation}
    \begin{split}
      [\i A\ , \i E] &=(E A_\alpha-A_\alpha E) + b\ E|\Gamma\rangle\langle \Theta_\beta|-b^\ast\ |\Theta_\beta\rangle\langle \Gamma|E\\ &= (E A_\alpha-A_\alpha E) \\
                     &\mathrel{\phantom{=}} {} +[1-D^{-1}_\alpha]( 
                       b\ |\Gamma\rangle\langle \Theta_\beta|-b^\ast\ |\Theta_\beta\rangle\langle \Gamma|) \ .
    \end{split}
  \end{equation}
  The first term $E A_\alpha-A_\alpha E $ is a traceless skew-Hermitian operator with support restricted to $\mathcal{H}_\alpha$. Therefore $E A_\alpha-A_\alpha E $ is an element of $\mathfrak{g}$. This means the second term, i.e., $[1-D^{-1}_\alpha](b\ |\Gamma\rangle\langle \Theta_\beta|-b^\ast\ |\Theta_\beta\rangle\langle \Gamma|)$ is in the Lie algebra generated by $\mathfrak{g}$ and $\{\i A\}$. Furthermore, since $b\neq 0$ and $1-D^{-1}_\alpha\neq 0$, we conclude that
  \be
  \i(|\Theta_\alpha\rangle\langle \Theta_\beta|+|\Theta_\beta\rangle\langle \Theta_\alpha|) \in \langle \{\i A\}, \mathfrak{g}\rangle\ ,
  \ee
  where 
  \be
  |\Theta_\alpha\rangle=-\i \frac{b}{|b|}|\Gamma\rangle\ ,
  \ee 
  is a normalized state. Since this operator is in the form given in \cref{trqew} then we can proceed with the rest of the proof of the lemma, as presented above.\\

  \noindent\textbf{The case of $D_\beta \geq 2$:} The fact that $A$ is hermitian means that it can be written
  \begin{equation}
    A = A_\alpha + A_{\alpha \beta} + A_{\beta \alpha} + A_\beta,
  \end{equation}
  where $A_\alpha = A_\alpha^\dagger$ has support restricted to $\hilbert_\alpha$ (and likewise for $A_\beta = A_\beta^\dagger$) and $A_{\alpha \beta} = \Pi_\alpha A \Pi_\beta$ satisfies $A_{\alpha \beta} = A_{\beta \alpha}^\dagger$. The assumption that $A$ is not block-diagonal means that there exists a normalized state $\nket{\Gamma} \in \mathcal{H}_\alpha$ so that $0 \neq A_{\beta \alpha} \nket{\Gamma} \in \hilbert_\beta$. Consider again the operator $E$ from \cref{eq:Eop} and the commutator
  \begin{equation}
    [\i A, \i E] = (E A_\alpha - A_\alpha E) + E A_{\alpha \beta} - A_{\beta \alpha} E.
  \end{equation}
  Similar to before, $E A_\alpha - A_\alpha E$ is in $\mathfrak{g}$. Thus $B_1 \equiv E A_{\alpha \beta} - A_{\beta \alpha} E \in \langle \sets{\i A}, \mathfrak{g} \rangle$. Consider further the commutator
  \begin{equation}
    B_2 \equiv \bigl[[B_1, \i E], \i E\bigr] = A_{\beta \alpha} E^3 - E^3 A_{\alpha \beta}.
  \end{equation}
  Then
  \begin{equation}
    B_1 + D_\alpha^2 B_2 = [-2 + 3 D_\alpha - D_\alpha^2] \p[\big]{\qproj{\Gamma} A_{\alpha \beta} - A_{\beta \alpha} \qproj{\Gamma}},
  \end{equation}
  which is not zero because it is assumed that $D_\alpha > 2$ and so
  \begin{equation}\label{eq:sproot}
    -2 + 3 D_\alpha - D_\alpha^2 < 0.
  \end{equation}
  (Interestingly, the fact that this is zero when $D_\alpha = 2$ opens up the possibility for the proper subgroup $\Sp(2) \subset \SU(4)$ to appear when $D_\beta = 2$ also. See \cref{sec:spblocks} for more discussion.)

  With $\nket{\Theta_\alpha} \equiv \nket{\Gamma}$ and
  \begin{equation}
    \nket{\Theta_\beta} \equiv \frac{\i}{\nbra{\Gamma} A_{\alpha \beta} A_{\beta \alpha} \nket{\Gamma}} A_{\beta \alpha} \nket{\Gamma},
  \end{equation}
  it is apparent that $\i \p[\big]{\nket{\Theta_\alpha} \nbra{\Theta_\beta} + \nket{\Theta_\beta} \nbra{\Theta_\alpha}} \in \langle \sets{\i A}, \mathfrak{g} \rangle$ is in the form \cref{trqew}, and so the proof can proceed as before.
\end{proof}

\subsection{Goursat's lemma for the special unitary group}\label{app:Goursat}

Recall \cref{def:subdirect} of a subdirect product: a subgroup of the direct product whose projection to each factor is surjective. Goursat's lemma characterizes all subdirect products of a pair of groups $G_1$ and $G_2$.

\begin{lemma}[Asymmetric Goursat's lemma]\label{lem:agoursat}
  The set of all subdirect products $H \subseteq G_1 \times G_2$ is in bijection with the set of all surjective homomorphisms $G_1 \to G_2 / N_2$ where $N_2 \unlhd G_2$ is a normal subgroup.
\end{lemma}

\begin{proof}
  First, suppose that $H \subseteq G_1 \times G_2$ is a subdirect product and consider the set
  \begin{equation}
    N_2 = \set{g_2 \in G_2 \given (e_1, g_2) \in H}.
  \end{equation}
  This is a normal subgroup of $G_2$: if $g_2' \in G_2$ then, since $H$ is a subdirect product, there exists $g_1 \in G_1$ such that $(g_1, g_2') \in H$, and
  \begin{equation}
    H \ni (g_1, g_2') (e_1, g_2) (g_1, g_2')^{-1} = (e_1, g_2' g_2 g_2'^{-1}).
  \end{equation}
  For any $g_1 \in G_1$, assign the set of elements of $G_2$ that show up with $g_1$ in $H$,
  \begin{equation}
    g_1 \mapsto \tilde{\Phi}_{g_1} = \set{g_2 \in G_2 \given (g_1, g_2) \in H}.
  \end{equation}
  If $g_2, g_2' \in \tilde{\Phi}_{g_1}$, then $(g_1, g_2)^{-1} (g_1, g_2') = (e_1, g_2^{-1} g_2') \in H$, so $g_2^{-1} g_2' \in N_2$. In other words, the cosets $g_2 N_2 = g_2' N_2$, and in fact $\tilde{\Phi}_{g_1} = g_2 N_2$ since if $(e_1, n_2) \in H$ then $(g_1, g_2 n_2) \in H$. It follows that $\tilde{\Phi} : G_1 \to G_2 / N_2$ is surjective, since $H$ is subdirect, and it is a homomorphism since if $(g_1, g_2), (g_1', g_2') \in H$ then $g_2 g_2' \in \tilde{\Phi}_{g_1 g_1'}$.

  Conversely, given a surjective homomorphism $\tilde{\Phi} : G_1 \to G_2 / N_2$ for some normal subgroup $N_2 \unlhd G_2$, define the subset
  \begin{equation}
    H = \set{(g_1, g_2) \in G_1 \times G_2 \given \tilde{\Phi}_{g_1} = g_2 N_2}.
  \end{equation}
  It is easily verified that this is a subdirect product: it's a subgroup since $\tilde{\Phi}$ is a homomorphism and the projections $\pi_i(H) = G_i$ since $\tilde{\Phi}$ is surjective (every element $g_2 \in G_2$ shows up in some coset, namely $g_2 N_2$).
\end{proof}

In applying Goursat's lemma to the special unitary groups, we take advantage of two facts: that special unitary groups are simply connected simple Lie groups, and that they enjoy an ``automatic continuity'' property for their homomorphisms to compact groups (see \cref{thm:auto}). Note that simple Lie groups are examples of so-called ``quasisimple'' groups: in particular, every proper normal subgroup of $\SU(l)$ is contained in its center $\mathbb{Z}_l$.

\begin{lemma*}[re \ref{lem:SUgoursat}]
  Let $l, l' \geq 2$ and let $G \subseteq \SU(l) \times \SU(l')$ be a subdirect product. There are two possibilities:
  \begin{enumerate}[(i)]
  \item $G = \SU(l) \times \SU(l')$.
  \item $l = l'$ and $G \cong \SU(l) \times \mathbb{Z}_q$, where $q$ divides $l$.
  \end{enumerate}
  Furthermore, in the second case, there exists an isomorphism $\Phi : \SU(l) \to \SU(l)$ such that
  \begin{equation}
    G = \set{(U, \e^{\i \theta} \Phi(U)) \given U \in \SU(l) \text{ and } \e^{\i \theta} \in \mathbb{Z}_q \subset \SU(l)}.
  \end{equation}
\end{lemma*}

\begin{proof}
  Consider the proof of the forward direction of \cref{lem:agoursat}, and let
  \begin{equation}
    U \mapsto \tilde{\Phi}_U = \set{U' \in \SU(l') \given (U, U') \in G}
  \end{equation}
  define the surjective homomorphism $\tilde{\Phi} : \SU(l) \to \SU(l') / N'$, where $N' \unlhd \SU(l')$ is the set of elements that show up with the identity $\ident \in \SU(l)$. Since $\SU(l')$ is quasisimple, either $N' = \SU(l')$ or $N' \subseteq \mathbb{Z}_{l'} \subseteq \SU(l')$. In the first case, it immediately follows that $G = \SU(l) \times \SU(l')$. So suppose that $N' \lhd \SU(l')$ is a proper subgroup, which, as a subgroup of $\mathbb{Z}_{l'}$, is equal to a cyclic subgroup, $N' = \mathbb{Z}_{q'}$ for some $q'$ that divides $l'$.

  By \cref{thm:auto}, since $\SU(l)$ and $\SU(l') / \mathbb{Z}_{q'}$ are both compact Lie groups, it follows that $\tilde{\Phi}$ is continuous. Therefore, since $\SU(l)$ is simply connected, there is a (continuous) surjective homomorphism $\Phi : \SU(l) \to \SU(l')$ which lifts $\tilde{\Phi}$.\footnote{That is, for each $U \in \SU(l)$, some $U' \in \tilde{\Phi}_U$ can be chosen so that the assignment $\Phi: U \mapsto U'$ is a proper homomorphism.} The kernel of any homomorphism should be a proper normal subgroup of $\SU(l)$, namely in the center of SU($l$), which is the discrete $\mathbb{Z}_l$ subgroup of phases $e^{\i 2\pi k/l } \mathbb{I}: k=0,\cdots, l-1$. Since $\Phi$ is continuous this implies that $\SU(l')$ has the same dimension as $\SU(l)$ as a Lie group (i.e. manifold), which means $l=l'$ and  therefore $\Phi$ is an isomorphism.

  Setting $q = q'$, it is clear that $G$ consists of the elements of the form \cref{eq:corr}.
\end{proof}

From this result, it immediately follows that the commutator subgroup $[G, G] \subseteq \SU(l)' \times \SU(l')$ is perfect and connected, equal to either $\SU(l) \times \SU(l')$ or $\set{(U, \Phi(U)) \given U \in \SU(l)} \cong \SU(l)$. Using the following description of normal subgroups of $\U(l)$, we extend this result to subgroups of $\U(l) \times \U(l')$ in \cref{cor:Ugoursat}. Let $\mathcal{P} = \set{e^{\i \theta} \ident \given \theta \in [0, 2 \pi)} \cong \U(1)$ denote the subgroup of overall phases and let $\mathbb{Z}_l = \mathcal{P} \cap \SU(l)$ denote the subgroup with determinant one, isomorphic to the cyclic group of order $l$.

\begin{lemma}\label{lem:Unorm}
    Let $N \unlhd \U(l)$ be a normal subgroup. Then either $N \subseteq \mathcal{P}$ or $N \supseteq \SU(l)$.
\end{lemma}

\begin{proof}
    Every element of $\U(l)$, and, in particular every element of $N$, can be written $\e^{\i \theta} U$ for some $\e^{\i \theta} \in \mathcal{P}$ and some $U \in \SU(l)$. Thus, $U \in \SU(l) \cap (\mathcal{P} N)$. Furthermore, since the product and intersection of normal subgroups are normal, it follows that
    \begin{equation}
        \SU(l) \cap (\mathcal{P} N) \unlhd \SU(l).
    \end{equation}
    Since $\SU(l)$ is quasisimple, (see the discussion before \cref{lem:SUgoursat}), either $\SU(l) \cap (\mathcal{P} N) = \SU(l)$, in which case $\SU(l) \subset \mathcal{P} N$, or $\SU(l) \cap (\mathcal{P} N) \subseteq \mathbb{Z}_l$. In the former case, taking commutators gives, since  $\SU(l)$ is perfect, $\SU(l) \subseteq [N, N]$. In the latter case, every $U \in \SU(l) \cap (\mathcal{P} N)$ is a phase, so $N$ consists entirely of phases, i.e. $N \subseteq \mathcal{P}$. This completes the proof.
\end{proof}

\begin{corollary*}[re \ref{cor:Ugoursat}]
    Suppose that $H \subseteq \U(l) \times \U(l')$ is a subgroup with $[H, H] \subseteq \SU(l) \times \SU(l')$ a subdirect product. Then either $[H, H] \cong \SU(l) \times \SU(l)$, or $l = l'$ and there is an isomorphism $\Phi : \SU(l) \to \SU(l)$ such that $[H, H] = \set{(U, \Phi(U)) \given U \in \SU(l)} \cong \SU(l)$. Furthermore, in the second case every element of $H$ is of the form $(\e^{\i \theta} U, \e^{\i \phi} \Phi(U))$.
\end{corollary*}

\begin{proof}
    Since $[H, H] \subseteq \SU(l) \times \SU(l')$ is a subdirect product, \cref{lem:SUgoursat} implies there are two possibilities: either it is the direct product, or $l = l'$ and $[H, H]$ consists entirely of elements of the form $(U, \e^{\i \theta} \Phi(U))$ where $\Phi : \SU(l) \to \SU(l)$ is an isomorphism and $\e^{\i \theta} \in \mathbb{Z}_q \subseteq \mathbb{Z}_l$ for some $q$ that divides $l$.

    Suppose the second case holds and let $(\e^{\i \theta_1} U_1, \e^{\i \theta_2} U_2) \in H$ be an arbitrary element, where $U_1, U_2 \in \SU(l)$. Then also $(\e^{\i \theta_1} \ident, \e^{\i \theta_2} \Phi(U_1)^{-1} U_2) \in H$. But the group $\angles{W_2 \in \U(l) \given (\e^{\i \theta_1} \ident, W_2) \in H}$ is a normal subgroup (note that it is sufficient to check only that this set is invariant under conjugation by $W_2' \in \SU(l)$ since overall phases vanish). Thus by \cref{lem:Unorm} it either consists entirely of phases or it contains all $\SU(l)$ elements.

    But if for every $W_2 \in \SU(l)$ it held that $(\e^{\i \phi} \ident, W_2) \in H$ for some phase $\e^{\i \phi}$, then taking commutators with elements $(U, \Phi(U)) \in [H, H]$ would give $(\ident, V_2) \in [H, H]$ for every commutator $V_2 = W_2 \Phi(U) W_2^\dagger \Phi(U)^\dagger$, which is in contradiction with the conclusion from \cref{lem:SUgoursat} that $(\ident, V_2) \in [H, H]$ if and only if $V_2 \in \mathbb{Z}_q$. Thus, it must hold that $\Phi(U_1)^{-1} U_2$ is a phase, i.e. $U_2 = \e^{\i \theta} \Phi(U_1)$. This proves that all elements of $H$ are of the form $(\e^{\i \phi_1} U, \e^{\i \phi_2} \Phi(U))$. Upon taking commutators, we find that the phases vanish, i.e. $[H, H] = \set{(U, \Phi(U))}$.
\end{proof}

\subsection{Automatic continuity}\label{app:auto}

There is a subtlety in the argument for \cref{lem:SUgoursat}, namely that in order to lift the projective homomorphism, continuity is required. This is guaranteed by the following result, which is an instance of a more general property known as ``automatic continuity''.

\begin{theorem}[Theorem 5.64 of \cite{hofmann2020structure}]\label{thm:auto}
  Assume that $f: G \to H$ is a group homomorphism where $H$ is a compact group and $G$ is a linear Lie group such that $\lie{g} = [\lie{g}, \lie{g}]$, i.e. $G_0 = [G_0, G_0]$ (where $\lie{g}$ is the Lie algebra of $G$ and $G_0 \subseteq G$ is the connected component of the identity). Then $f$ is continuous.
\end{theorem}

Note that ``linear Lie group'' in \cite{hofmann2020structure} is more general than the typical meaning of ``matrix Lie group'', so indeed all compact Lie groups are linear Lie groups.

The usual classification of irreps of $\SU(l)$ goes through its Lie algebra, and thus requires an assumption of continuity. The above result implies that all representations are automatically continuous. This ensures, for instance, that all isomorphisms $\Phi: \SU(l) \to \SU(l)$ are of the form $U \mapsto W U W^\dagger$ or $U \mapsto W U^\ast W^\dagger$, as in \cref{eq:isos}.

\subsection{Serre's lemma}\label{app:serre}

Now we consider a subdirect product $H \subseteq G_1 \times \dots \times G_r$ which is pairwise independent. Recall that a group is perfect if it is equal to its commutator subgroup, $[G_i, G_i] = G_i$.

\begin{lemma*}[Serre's lemma \cite{Ribet1974adic}, re \ref{lem:serre}]
  Let $H \subseteq G_1 \times \dots \times G_r$ be a subgroup such that $\pi_{ij}(H) = G_i \times G_j$ for all pairs $1 \leq i < j \leq r$, where $\pi_{ij} : G_1 \times \dots \times G_r \to G_i \times G_j$ is the projection homomorphism. If each $G_i$ is perfect, then $H = G_1 \times \dots \times G_r$.
\end{lemma*}

\begin{proof}
  Without loss of generality, it suffices to prove that, for all $g_1 \in G_1$,
  \begin{equation}\label{eq:serre1}
    (g_1, e_2, \dots, e_r) \in H,
  \end{equation}
  where $e_i \in G_i$ is the identity element. In particular, if this is true for $G_1$ then it is symmetrically true for all factors $G_i$, and altogether these elements generate the entire direct product.

  We proceed by induction; the base case of $r = 2$ holds from the assumption of the lemma. The induction hypothesis applied to $G_1 \times G_3 \times \dots \times G_r$ and $G_1 \times G_2 \times G_4 \times \dots \times G_r$, respectively, implies that, for all $g_1, g_1' \in G_1$, there exist $g_2 \in G_2$ and $g_3 \in G_3$ such that
  \begin{subequations}
    \begin{align}
      h_1 & = (g_1, g_2, e_3, \dots, e_r) \in H \\
      h_2 & = (g_1', e_2, g_3, e_4, \dots, e_r) \in H.
    \end{align}
  \end{subequations}
  Therefore,
  \begin{equation}
    [h_1, h_2] = ([g_1, g_1'], e_2, \dots, e_r) \in H,
  \end{equation}
  where $[h_1, h_2] = h_1 h_2 h_1^{-1} h_2^{-1}$ is the group commutator. Since $G_1$ is assumed perfect, elements of the form $[g_1, g_1']$ generate $G_1$, and so \cref{eq:serre1} holds.
\end{proof}

\subsection{Commutator subgroup is connected component of identity}\label{sec:commconn}

Here we verify the statement that,
\begin{lemma}\label{lem:connect}
    Let $\mathcal{T} \subseteq \mathcal{V}^G$ be a subgroup such that $\mathcal{T}$ acts irreducibly on each multiplicity subsystem (i.e. $\pi_\lambda(\mathcal{T}) \subseteq \U(\hilbert[M]_\lambda)$ is irreducible for each $\lambda \in \Lambda$). Then the connected component of the identity of $\mathcal{T} \cap \mathcal{SV}^G$ is equal to the connected component of the identity of the commutator subgroup $[\mathcal{T}, \mathcal{T}]$.
\end{lemma}

In particular, when $[\mathcal{T}, \mathcal{T}]$ is connected (which is always the case when $\mathcal{T}$ is connected), it is equal to the identity component of $\mathcal{T} \cap \mathcal{SV}^G$. We note that when the group $\mathcal{T}$ is generated by $k$-local $G$-invariant gates, it is always compact and connected \cite{Marvian2022Restrict}. Furthermore, in the case of $G = \SU(d)$ on $d$-dimensional qudits, each of $\mathcal{V}_k^{(n)}$ acts irreducibly on the multiplicity subsystems since they contain $\P(\sigma) : \sigma \in \mathbb{S}_n$.

\begin{proof}
    Let $\mathfrak{v}$ and $\mathfrak{t}$ be the Lie algebras of $\mathcal{V}^G$ and $\mathcal{T}$, respectively (see \cref{eq:lie}). Both are groups of unitaries, hence their Lie algebras are reductive \cite{robert2015squares,zimboras2015symmetry}, meaning that they split into a direct sum of their commutator subalgebra and center,
    \begin{equation}
    \begin{split}
        \mathfrak{v} & = \mathfrak{sv} \oplus \mathfrak{p} = [\mathfrak{v}, \mathfrak{v}] \oplus \mathfrak{p} \\
        \mathfrak{t} & = \mathfrak{s} \oplus \mathfrak{z} = [\mathfrak{t}, \mathfrak{t}] \oplus \mathfrak{z},
    \end{split}
    \end{equation}
    where $\mathfrak{sv}$ and $\mathfrak{s}$ are  semisimple and $\mathfrak{p}$ and $\mathfrak{z}$ are commutative (the respective centers of $\mathfrak{v}$ and $\mathfrak{t}$). Since $\mathfrak{t} \subseteq \mathfrak{v}$, it follows that $\mathfrak{s} \subseteq \mathfrak{sv}$. By Schur's lemma and the fact that each $\pi_\lambda(\mathcal{T})$ acts irreducibly, elements of the center $\mathfrak{z}$ must be real linear combinations of the projectors $\i \Pi_\lambda : \lambda \in \Lambda$ to charge sectors, which span $\mathfrak{p}$. Thus $\mathfrak{z} \subseteq \mathfrak{p}$.
    
    The Lie subalgebra $\mathfrak{s} = [\mathfrak{t}, \mathfrak{t}]$ is also the Lie algebra of the commutator subgroup $[\mathcal{T}, \mathcal{T}]$, meaning that its exponentials generate its connected component of the identity. Because its Lie algebra $\mathfrak{s} \subseteq \mathfrak{sv}$, it follows that $\mathfrak{s}$ is also the Lie algebra of the intersection $\mathcal{T} \cap \mathcal{SV}^G$, i.e. the identity components coincide.
\end{proof}

It is worth noting that the above argument implies that every element of the connected component of the identity of $\mathcal{T}$ can be written as a product of an element of $[\mathcal{T}, \mathcal{T}]$ and a relative phase in the form $\e^{\i P}$ with $P \in \mathfrak{z}$.

We now prove \cref{lem-2-local}.

\begin{proof}
Recall that $\mathcal{V}_2^{(n)}=\langle \e^{\i t} \ident, \e^{\i t \P_{ij}} : t \in \real, i \neq j \rangle$ is a compact connected Lie group \cite{Marvian2022Restrict}, which means any unitary $V\in \mathcal{V}_2^{(n)}$ can be written as exponential $V=\exp(A)$, where $A$ is in the real Lie algebra generated by operators $\i\P_{ij}: i \neq j$ and the identity $\i \ident$. The Lie subalgebra generated by transpositions is equal to the real Lie algebra generated by operators $\i(\P_{ij}- \P_{kl}): i \neq j, k \neq l $ and the permutationally-invariant operator $\i B=\i \sum_{i\neq j} \P_{ij}$ (This can be seen by noting that the linear spans of the generators of the two Lie algebras are equal). Since $B$ is permutationally-invariant, it commutes with all $\i(\P_{ij}- \P_{kl}): i \neq j, k \neq l $. Therefore any unitary $V\in \mathcal{V}_2^{(n)}$ can be decomposed as $\e^{\i \theta} \widetilde{V} \exp(\i B s)= \e^{\i \theta} \exp(\i B s) \widetilde{V}$, where $\widetilde{V}$ is in the group on the left-hand side of \cref{eq:swapirrep}, denoted as $\mathcal{X}$, and $s\in\mathbb{R}$. It follows that $\mathcal{X}$ contains the commutator subgroup of $ \mathcal{V}_2^{(n)}$, namely $ \mathcal{SV}_2^{(n)}$. On the other hand, $\mathcal{X}$ itself is also contained in $\mathcal{SV}_2^{(n)}$. This can be seen, for instance, by noting that Hamiltonians $\P_{ij}- \P_{kl}$ are centerless, i.e., $\Tr(u^{\otimes n}[\P_{ij}- \P_{kl}])=0$, which means the unitaries realized by them are in $\mathcal{SV}^{(n)}\cap \mathcal{V}^{(n)}_2$, hence by \cref{lem:connect} they are in $\mathcal{SV}^{(n)}_2$. This completes the proof of the first part of the lemma. The second part of the lemma follows from the fact that $\mathcal{V}_2^{(n)}$ contains $\P(\sigma) :\sigma\in\mathbb{S}_n$, which acts irreducibly on $\mathcal{M}_\lambda$. This in turn implies its commutator subgroup $\mathcal{SV}_2^{(n)}$ also acts irreducibly on $\mathcal{M}_\lambda$. 
 \end{proof}

We note that, with \cref{lem:connect}, we could alternatively prove this by showing that $\i(\P_{ij} - \P_{kl})$ generate the semsisimple part of the Lie algebra of $\mathcal{V}_2^{(n)}$, for instance using the commutator identity
\begin{align}
 \i(\P_{12} - \P_{23}) = \frac{1}{2}\big[ [\i\P_{12}, \i\P_{23}], \i\P_{31} \big].
\end{align}

\newpage

\section{For \texorpdfstring{$d \geq 3$}{d >= 3}, universality requires \texorpdfstring{$n$}{n}-qudit interactions}\label{app:conj-proof}

In this appendix, we prove that, without ancillae, and when $d \geq 3$, $(n - 1)$-local $\SU(d)$-invariant unitaries are not universal.  As discussed before \cref{eq:centerdiff}, applying the results of \cite{Marvian2022Restrict}, we find that
\begin{equation}\nonumber
  \dim \mathcal{V}_n^{(n)} - \dim \mathcal{V}_k^{(n)} = \abs{\Lambda_{n, d}} - \abs{\Lambda_{k, d}}\ ,
\end{equation}
where $\abs{\Lambda_{k, d}}$ is the number of inequivalent irreps of $\SU(d)$, or, equivalently $\mathbb{S}_n$ that appear on $k$ qudits.

\lemIrrepsSn*

\begin{proof}
  We prove this lemma by defining an injective map $i: \Lambda_{n - 1, d} \rightarrow \Lambda_{n, d}$, which implies $|\Lambda_{n, d}| \geq |\Lambda_{n - 1, d}|$. Then, we show $\Lambda_{n, d}$ contains elements that are not in the image of $i$.
  
  This map $i$ is defined as follows. Let $\lambda \in \Lambda_{n -1, d}$, which can be labeled by a Young diagram with $n-1$ boxes. Then $i(\lambda)$ is defined by adding a box to the first row of $\lambda$. Clearly, this map is injective, because if $i(\lambda_1) = i(\lambda_2)$, then by removing a box from the first row, we must have $\lambda_1 = \lambda_2$.
  
  It is easy to see that for any $i(\lambda)$, the number of boxes in the first row is always greater than the number of boxes in the second row. Then as long as there exists a $\mu \in \Lambda_{n, d}$, such that the first two rows of $\mu$ have the same length, then $\mu \notin i(\Lambda_{n -1, d})$. When $d\geq 3$, such $\mu$ always exists: if $n$ is even, we can choose $\mu = \sydiag{2,2}\cdots\sydiag{2,2}$; if $n$ is odd, we can choose $\mu = \sydiag{2,2,1}\cdots\raisebox{0.3em}{$\sydiag{2,2}$}$.
\end{proof}
Note that when $n$ is odd and $d=2$, the above $\mu$ does not exist. Indeed, in this case we have $|\Lambda_{n, 2}| = |\Lambda_{n - 1, 2}|$ when $n$ is odd.

\newpage

\section{Semi-universality on \texorpdfstring{$n=3$}{n=3} qudits using \texorpdfstring{$8$}{8} ancilla qudits}\label{sec:ancillapr}

In this appendix, we consider unitaries that are realizable on 11 qudits using 2-qudit $\SU(d)$-invariant unitaries. More precisely, we focus on the behavior of the group $\mathcal{SV}_2^{(11)}$ in certain irreps that are relevant for the use of ancilla qudits. To understand the properties of this group, we use Marin's characterization of the Lie algebra generated by transpositions (SWAPs) as a subalgebra of the group algebra \cite{marin2007algebre}. Roughly speaking, according to this characterization, for any system with an arbitrary number of qudits, there are two sources of constraints: a relation between charge sectors labeled by L-shape diagrams, and a relation between $\lambda$ and $\lambda'$, where $\lambda'$ is the transpose (sometimes called conjugate) Young diagram (which may or may not be the same as $\lambda$). In particular, this result implies

\begin{proposition}[Marin \cite{marin2007algebre}]\label{prop:marin}
  Two-qudit $\SU(d)$ symmetric Hamiltonians are semi-universal in a subset of irreps labeled by $\Lambda$, provided that $\Lambda$ satisfies the following two criteria: 
  \begin{enumerate}
  \item $\Lambda$ does not contain L-shape Young diagrams;
  \item All $\lambda \in \Lambda$ satisfies $\lambda' \not\in \Lambda$.
  \end{enumerate}
\end{proposition}

The following lemma does not actually depend on \cref{prop:marin} in its full strength, which applies to a general number of qudits. In fact, we only need \cref{prop:marin} applied to 11 qudits.

Recall that an operator $H$ is called centerless if $\Tr H \Pi_\lambda = 0$ for all $\lambda \in \Lambda_{n, d}$.

\begin{lemma*}[re \ref{lem:ancillauni}]
  For any centerless $\SU(d)$-invariant Hamiltonian $H$, there exists a Hamiltonian $\tilde{H}$ that is realizable with 2-qudit $\SU(d)$-invariant Hamiltonians, i.e., $\e^{\i t\tilde{H}}\in\mathcal{V}_2: t\in{\mathbb{R}}$ such that 
  \begin{align}
    \e^{\i tH} (|\psi\rangle \otimes |\eta\rangle) = \e^{\i t\tilde{H}} (|\psi\rangle \otimes |\eta\rangle),
  \end{align}
  for all $|\psi\>\in(\mathbb{C}^d)^{\otimes 3}$,
  where for $d\ge 4$, 
  \begin{align}\label{eq:8anc}
    |\eta\rangle=(\ket{0} \wedge \ket{1} \wedge \ket{2} \wedge \ket{3})^{\otimes 2}\in(\mathbb{C}^d)^{\otimes 8}\ 
  \end{align}
  is a 8-qudit state, and for $d=3$, 
  \begin{align}
    |\eta'\rangle=(\ket{0} \wedge \ket{1})^{\otimes 2} \otimes |00\>\in(\mathbb{C}^3)^{\otimes 6}\ ,
  \end{align}
  is a 6-qutrit state.
\end{lemma*}
In addition to \cref{prop:marin}, the following proof of this lemma requires some elementary knowledge of manipulating Young diagrams, which can be found, for example, in \cite{georgi2018lie}.

\begin{proof}
  When $d\geq 4$, the 8-qudit state $|\eta\>$ is within the irrep labeled by the Young diagram $\ysub{2,2,2,2}$. Then, one can show that the 11-qudit state $|\psi\> \otimes |\eta\>$ only has components in the following 9 Young diagrams, which are obtained by adding $3$ boxes to $\ysub{2,2,2,2}$,
  \begin{align}
    \ydiag{5,2,2,2}~~ \ydiag{4,3,2,2}~~ \ydiag{4,2,2,2,1}~~ \ydiag{3,3,3,2}~~ \ydiag{3,3,2,2,1}~~ \ydiag{3,2,2,2,2}~~ \ydiag{3,2,2,2,1,1}~~ \ydiag{2,2,2,2,2,1}~~ \ydiag{2,2,2,2,1,1,1}~,
  \end{align}
  These Young diagrams will be denoted as $\Lambda$. Here we assume $d \geq 7$. When $d<7$, we just have to remove all the diagrams with more than $d$ rows. It can be easily checked that $\Lambda$ satisfies the above two criteria in \cref{prop:marin}. Therefore, 2-qudit Hamiltonians are semi-universal in $\Lambda$. In other words, for any centerless $\SU(d)$-invariant Hamiltonian $H$, there exists a Hamiltonian $\tilde{H}$ that is realizable with 2-qudit $\SU(d)$-invariant Hamiltonians, such that 
  \begin{align}
    \e^{\i tH} (|\psi\rangle \otimes |\eta\rangle) = \e^{\i t\tilde{H}} (|\psi\rangle \otimes |\eta\rangle).
  \end{align}
  
  { Next, we consider the 6-qutrit state $|\eta'\>\in(\mathbb{C}^3)^{\otimes 6}$. This state is within the irrep labeled by the Young diagram $\ysub{4,2}$.} By tensoring it with the 3 qutrits in arbitrary states, we find that the resulting set of Young diagrams is
  \begin{align}
    &\ydiag{7,2}~~ \ydiag{6,3}~~ \ydiag{6,2,1}~~ \ydiag{5,4}~~ \ydiag{5,3,1}~~ \ydiag{5,2,2}~~ \ydiag{4,4,1}~~ \ydiag{4,3,2}~,
  \end{align}
  which still satisfies the above two criteria in \cref{prop:marin}. By a similar argument, we know that $H$ is realizable using 2-qudit symmetric Hamiltonians.
\end{proof}

\newpage

\section{\texorpdfstring{$t$}{t}-designs: Proof of Proposition \ref{prop:design}} \label{App:tdesign}

In this section, we prove \cref{prop:design}. To prove this result, in addition to the semi-universality of 3-qudit gates, we also use the following fact: Let $\mu_0$ and $\mu_1$ be, respectively, the symmetric $\sydiag{6}$$\cdots$ and the standard irreps $\sydiag{5,1}$\raisebox{0.7ex}{$\cdots$} of $\mathbb{S}_n$. Then, the projection of $\mathcal{V}^{(n)}_3$ to these two sectors is equal to 
\be\label{sty}
\pi_{\mu_0,\mu_1}(\mathcal{V}^{(n)}_3)=\pi_{\mu_0,\mu_1}(\mathcal{V}^{(n)})\cong  \U(\mathcal{M}_{\mu_0})\times \U(\mathcal{M}_{\mu_1})  \ .
\ee
To see this note that $\mathcal{V}^{(n)}_3$ contains the subgroup $\exp(\i\phi_0 A_0) \exp(\i\phi_1 A_1) : \phi_0,\phi_1\in[0,2\pi)$, where
\begin{align}
  A_0=\frac{1}{n}B_2 - \frac{n-3}{2}\mathbb{I}\ , \ \  
  A_1=\frac{n-1}{2}\mathbb{I} - \frac{1}{n}B_2\ ,
\end{align}
and $B_2 = \frac{1}{2}\sum_{i,j} \P_{ij}$.  It can be easily seen that  $A_i \Pi_{\mu_j}= \delta_{i,j} \Pi_{\mu_i},$ for $i,j\in\{0,1\}$, which means when projected to the sectors $\mu_0$ and $\mu_1$, $\pi_{\mu_0,\mu_1}\{\exp(i\phi_0 A_0)\exp(\i\phi_1 A_1) : \phi_0,\phi_1\in[0,2\pi)\}\cong \U(1)\times \U(1)$. Together with $\mathcal{SV}^{(n)}\subset\mathcal{V}^{(n)}_3$ proves \cref{sty}.

Recall the decomposition $V=\bigoplus_{\lambda} (\mathbb{I}_{\mathcal{Q}_\lambda}\otimes v_\lambda)$. Applying this decomposition to both sides of \cref{design}, we see that this equation holds, if and only if 
\be\label{rew}
\mathbb{E}_{V\in\V_3^{(n)}}[\bigotimes_{i=1}^t v_{\lambda_i}\otimes  v^\ast_{\lambda'_i}]=\mathbb{E}_{V\in\V^{(n)}}[\bigotimes_{i=1}^t v_{\lambda_i}\otimes  v^\ast_{\lambda'_i}]\ ,
\ee
for all $\lambda_1,\cdots, \lambda_t, \lambda'_1,\cdots, \lambda'_t\in\Lambda_{n,d}$. To analyze these expectation values, we use the following standard fact: For the  expectation value with respect to the Haar measure over $\U(m)$, unless $r=r'$, 
$\mathbb{E}_{U} (U^{\otimes r}\otimes {U^\ast}^{\otimes r'})=0\ .$
This means that unless $\lambda'_1,\cdots,\lambda'_t$ is a permutation of $\lambda_1,\cdots,\lambda_t$ the right-hand side of \cref{rew} vanishes.  On the other hand, when they are permutations of each other, phases do not contribute to the expectation value, so semi-universality implies both sides are equal. 
Therefore, to guarantee \cref{rew}, we only need to consider cases where $\lambda'_1,\cdots,\lambda'_t$ is not a permutation of $\lambda_1,\cdots,\lambda_t$, and make sure the left-hand side of \cref{rew} vanishes.  

Next, we use the result of \cite{rasala1977minimal} which implies that 
when $n\geq 9$ and $d<n-1$, 
the three irreps of $\mathbb{S}_n$ with the lowest dimensions are (i) the symmetric irrep $\mu_0=\sydiag{6}$$\cdots$ with dimension $1$, (ii) the standard irrep $\mu_1$ with dimension $n-1$, and (iii) the irrep $\sydiag{4,2}$\raisebox{0.7ex}{$\cdots$} (two-row Young diagrams with two boxes in the second row) with dimension $\frac{1}{2}n(n-3)$.

Let $\Delta$ be the set of irreps in $\Lambda_{n,d}$ for which  the number of their occurrences in $\{\lambda_1,\cdots,\ \lambda_t\}$
and $\{\lambda'_1,\cdots,\ \lambda'_t\}$ are not equal.  There are two cases depending on whether $\Delta$ contains an element of $\Lambda_{n,d}\setminus\{\mu_0,\mu_1\}$.  First, assume it does contain such elements, denoted as $\delta$. Then, for all $t< \frac{1}{2}n(n-3)$,
\be
\mathbb{E}_{V\in\mathcal{SV}^{(n)}}[\bigotimes_{i=1}^t v_{\lambda_i}\otimes  v^\ast_{\lambda'_i}] \cong \mathbb{E}_{V\in\mathcal{SV}^{(n)}}[v_\delta^{\otimes n_\delta} \otimes v_{\delta}^{\ast \otimes n_\delta'} \otimes V_\mu \otimes V_\nu^\ast] = 0\  ,
\ee
where in the second expression we have separated the occurrences of $\delta$ from all other irreps, denoted $V_\mu \otimes V_\nu^\ast$, and $n_\delta$ and $n_\delta'$ are the numbers of times $\delta$ shows up in $\lambda_1, \dots, \lambda_t$ and $\lambda_1', \dots, \lambda_t'$, respectively. To show that this expectation value vanishes, it suffices to show that 
$\mathbb{E}_{v_{\delta}\in\SU(\mathcal{M}_\delta)} v_\delta^{\otimes n_{\delta}}\otimes {(v_\delta)^\ast}^{\otimes n_{\delta}'} = 0$.  
Since $\dim (\mathcal{M}_\delta) \ge \frac{1}{2}n(n-3)$, we have  $\dim (\mathcal{M}_\delta) > t \ge n_{\delta}, n'_{\delta}$. This, in turns, implies $|n_\delta - n'_\delta| < \dim (\mathcal{M}_\delta)$. We also know that because $\delta\in\Delta $,   $n_{\delta}\neq  n'_{\delta}$.  Recall that with the respect to the Haar measure over $\SU(m)$, we have $\mathbb{E}_{U\in\SU(m)} (U^{\otimes r}\otimes {U^\ast}^{\otimes r'})\neq 0$, if and only if  $r=r' (\mod  m)$. We conclude that 
$\mathbb{E}_{v_{\delta}\in\SU(\mathcal{M}_\delta)} v_\delta^{\otimes n_{\delta}}\otimes {(v_\delta)^\ast}^{\otimes n'_{\delta}}=0$,  which in turn proves the above equality. Finally, we note that because $\mathcal{SV}^{(n)}$ is a subgroup of $\V_3^{(n)}$ (semi-universality), the above identity implies that 
\be
\mathbb{E}_{V\in\V_3^{(n)}}[\bigotimes_{i=1}^t v_{\lambda_i}\otimes  v^\ast_{\lambda'_i}]=\mathbb{E}_{V\in\mathcal{SV}^{(n)}}[\bigotimes_{i=1}^t v_{\lambda_i}\otimes  v^\ast_{\lambda'_i}]=0\  ,
\ee
Next, we focus on the second case, i.e., when $\Delta$ does not have any element if $\Lambda_{n,k}\setminus\{\mu_0,\mu_1\}$, which means $\Delta\subseteq \{\mu_0,\mu_1\}$. In this case there exists $\delta\in\{\mu_0,\mu_1\}$ such that the number of its occurrence in  $\{\lambda_1,\cdots,\lambda_t\}$ and 
$\{\lambda'_1,\cdots,\lambda'_t\}$, denoted as $n_{\delta}$ and $n'_{\delta}$ are not equal. Then, in the case applying \cref{sty} we have 
\be
\mathbb{E}_{V\in\V_3^{(n)}}   v_{\delta}^{\otimes n_{\delta}}\otimes {(v_{\delta})^\ast}^{\otimes n'_{\delta}}=\mathbb{E}_{V\in\mathcal{V}^{(n)}}   v_{\delta}^{\otimes n_{\delta}}\otimes {(v_{\delta})^\ast}^{\otimes n'_{\delta}}= \mathbb{E}_{v_\delta \in \U(\mathcal{M}_\delta)}   v_{\delta}^{\otimes n_{\delta}}\otimes {(v_{\delta})^\ast}^{\otimes n'_{\delta}}=0\ .
\ee
Therefore, we conclude that for all $t< n (n - 3) / 2$, \cref{rew} holds for all $\{\lambda_1,\cdots,\lambda_t\}$ and 
$\{\lambda'_1,\cdots,\lambda'_t\}$. This completes the proof of \cref{prop:design}.

\newpage

\section{Techniques for controllability}\label{sec:genSUprf}

In this appendix, we describe techniques for controllability using the notion of isolation to subspaces of a Hilbert space. In \cref{sec:useiso} we define ``isolation'' and we discuss some of its uses. In \cref{sec:spblocks} we show that one of the techniques described in \cref{sec:useiso} can fail when the subspaces are two-dimensional. \Cref{sec:strat} provides strategies for generating isolated elements of Lie algebras.

\subsection{Using isolation}\label{sec:useiso}

In this section we describe some uses of isolating Lie algebra elements to off-diagonal blocks. First, we discuss what we mean by ``isolation''. Consider a finite-dimensional Hilbert space with an orthogonal decomposition $\hilbert = \bigoplus_{i = 1}^m \hilbert_i$ and a group of unitaries $\mathcal{W} \subseteq \U(\hilbert)$. A nonzero element $X \in \lie{w}$ of the Lie algebra of $\mathcal{W}$,
\begin{equation}
  \lie{w} = \set{X \in \hilbert[L](\hilbert) \given \forall t \in \real, \e^{t X} \in \mathcal{W}}, \tag{re \ref{eq:lie}}
\end{equation}
is called (off-diagonal) \emph{isolated} to block $i$ if
\begin{equation}\label{eq:one-isolated}
  X = \Pi_i X + X \Pi_i.
\end{equation}
Note that this implies that $\Pi_i X \Pi_i = 0$ by applying $\Pi_i$ on the right and left to both sides of \cref{eq:one-isolated}. We also consider isolation to two blocks, say $i$ and $j$: $X \neq 0$ is \emph{isolated} to these if
\begin{equation}\label{eq:two-isolated}
  X = \Pi_i X \Pi_j + \Pi_j X \Pi_i.
\end{equation}

Isolated elements in the Lie algebra are particularly helpful when there are nontrivial elements $U_i \in \mathcal{W}$ which act as the identity on $\hilbert_j$ for $j \neq i$. Let $\mathcal{W}_i \subseteq \mathcal{W}$ be the subgroup which acts nontrivially only on the $i$th block, defined via
\begin{equation}\label{eq:isub}
  \mathcal{W}_i = \set{U_i \in \mathcal{W} \given \Pi_j U_i = U_i \Pi_j = U_i^{\delta_{ij}} \Pi_j},
\end{equation}
where $\delta_{ij}$ is the Kronecker delta, i.e.,
\begin{equation}
  U_i^{\delta_{ij}} \Pi_j \coloneqq \begin{cases} U_i \Pi_i & j = i \\ \Pi_j & j \neq i. \end{cases}
\end{equation}
Note that if $W \in \mathcal{W}$ and $X \in \lie{w}$ then $W X W^\dagger \in \lie{w}$ since for all $t \in \real$
\begin{equation}\label{eq:adjoint}
  \e^{t W X W^\dagger} = W \e^{t X} W^\dagger \in \mathcal{W}.
\end{equation}
If $X \in \lie{w}$ is isolated to the $i$th block (i.e. it satisfies \cref{eq:one-isolated}), then $\Pi_i X U_i=\Pi_i X$, which, in turn, implies 
\begin{equation}
  U_i X U_i^\dagger = U_i \Pi_i X + X \Pi_i U_i^\dagger \in \lie{w}\ ,
\end{equation}

Furthermore, $\lie{w}$ is closed under real linear combinations, so for all
\begin{equation}
  B \in \operatorname{span}_\real \set{U_i \Pi_i \given U_i \in \mathcal{W}_i} = \operatorname{span}_\real \mathcal{W}_i \Pi_i,
\end{equation}
it holds that $B X + X B^\dagger \in \lie{w}$. In other words,
\begin{equation}\label{eq:span-iso}
  \set{B X + X B^\dagger \given B \in \operatorname{span}_\real \mathcal{W}_i \Pi_i} = \operatorname{span}_\real \set{U_i X U_i^\dagger \given U_i \in \mathcal{W}_i} \subseteq \lie{w}.
\end{equation}
As long as there is some $U_i \in \mathcal{W}_i$ such that $U_i X U_i^\dagger \neq X$, this can be used to generate new elements of $\mathcal{W}$ which act only in the subspace $\hilbert_i \oplus X \hilbert_i$ (where $X \hilbert_i$ is the image of $X \Pi_i$; note that $X \hilbert_i$ consists entirely of vectors orthogonal to $\hilbert_i$ since $\Pi_i X \Pi_i = 0$).

For instance, suppose that 
\begin{equation}
  \operatorname{span}_\real \set{\Pi_i U_i \given U_i \in \mathcal{W}_i} = \operatorname{span}_\complex \set{\Pi_i U_i \given U_i \in \mathcal{W}_i}
\end{equation}
(this implies that $\mathcal{W}_i$ does not act on $\hilbert_i$ by orthogonal or symplectic matrices), and that the action of $\mathcal{W}_i$ on $\hilbert_i$ is irreducible, so that 
\begin{equation}\label{eq:irrspan}
  \operatorname{span}_\complex \set{\Pi_i U_i \given U_i \in \mathcal{W}_i} = \hilbert[L](\hilbert_i) = \set{B \in \hilbert[L](\hilbert) \given B = \Pi_i B \Pi_i}.
\end{equation}
i.e. the span of the unitaries in $\mathcal{W}_i$, restricted to $\hilbert_i$, consists of all linear operators on $\hilbert_i$. Then the real span includes all rank-one projectors on $\hilbert_i$, and so \cref{eq:span-iso} includes all operators of the form $\qout{\psi}{\varphi} - \qout{\varphi}{\psi}$ where $\ket{\psi} \in \hilbert_i$ and $\ket{\varphi} \in X \hilbert_i$. Then, since this is a generating set for $\su(\hilbert_i \oplus X \hilbert_i) \subseteq \lie{w}$ (see \cref{lem:twoblockoff}),
\begin{equation}\label{eq:block-cntrl}
  \SU(\hilbert_i \oplus X \hilbert_i) \subseteq \mathcal{W}.
\end{equation}

Applying this technique in the section we prove the following lemma, which is used in \cref{sec:lem2} for the proof of \cref{thm:mulblocks}.
\begin{lemma}\label{lem:iso-oneblock}
  Let $\mathcal{W} \subseteq \U(\hilbert)$ be a group of unitaries on $\hilbert = \bigoplus_{i = 1}^m \hilbert_i$. Suppose that, for some $i = 1, \dots, m$ the dimension $\dim \hilbert_i \geq 3$ and $\SU(\hilbert_i) \subseteq \mathcal{W}_i$, where $\mathcal{W}_i \subseteq \mathcal{W}$ is the subgroup acting nontrivially only on $\hilbert_i$, defined in \cref{eq:isub}. Let $\Pi_i$ be the Hermitian projector to $\hilbert_i$ and $\Pi_\perp$ the projector to its orthogonal complement. If $A$ is an (anti-Hermitian) operator such that $\e^{t A} \in \mathcal{W}$ for all $t \in \real$, then for all $B = \Pi_i B \Pi_i$ and $t \in \real$,
  \begin{equation}\label{eq:isoLie}
    \exp t (B A \Pi_\perp + \Pi_\perp  A B^\dagger) \in \mathcal{W}.
  \end{equation}
\end{lemma}
After discussing isolation techniques, we prove this lemma in \cref{app:applications}. In the next subsection, we discuss why the assumption that $\dim \hilbert_i \geq 3$ is necessary.

It is worth noting that, if $\mathcal{W}_i$ has a nontrivial Lie algebra, denoted $\lie{w}_i$, then an isolated $X \in \lie{w}$ can possibly (as long as it does not commute with $\lie{w}_i$) be used to generate a larger subalgebra $\angles{\lie{w}_i, X} \subseteq \lie{w}$.

\subsection{Why two-dimensional blocks are special}\label{sec:spblocks}

The assumption that $\operatorname{span}_\real \mathcal{W}_i \Pi_i = \operatorname{span}_\complex \mathcal{W}_i \Pi_i$, discussed above, is not satisfied for a particularly important group: when $\mathcal{W}_i = \SU(\hilbert_i)$ and $\dim \hilbert_i = 2$. This is related to the exceptional isomorphism $\SU(2) \cong \Sp(1)$ with the compact symplectic group \cite{goodman2009symmetry}.

One way of seeing this is the following. Every element $U \in \SU(2)$ can be written
\begin{equation}
  U = \e^{\i \theta (\hat{\vct{n}} \cdot \vct{\sigma})} = \cos \theta \ident +i \sin \theta (\hat{\vct{n}} \cdot \vct{\sigma})\ ,
\end{equation}
where $\vct{\sigma}$ is the Pauli vector operator, $\hat{\vct{n}} \in \real^3$ is a unit vector, and $\theta \in [0, 2 \pi)$. This implies
\begin{equation}
  \operatorname{span}_\real \SU(2) = \operatorname{span}_\real \set{\ident, i \sigma_x, i \sigma_y, i \sigma_z}.
\end{equation}
Notably, the traceless part of any $B \in \operatorname{span}_\real \SU(2)$ is anti-Hermitian. Hence $\operatorname{span}_\real \SU(2)$ is clearly a proper \emph{real} subspace of $\hilbert[L](\complex^2)$, the space of linear operators over $\complex^2$:
\begin{subequations}
  \begin{equation}
    \operatorname{span}_\real \SU(2) \neq \operatorname{span}_\complex \SU(2) = \hilbert[L](\complex^2).
  \end{equation}
  In fact, with the identity $(\hat{\vct{n}} \cdot \vct{\sigma}) (\hat{\vct{n}}' \cdot \vct{\sigma}) = \hat{\vct{n}} \cdot \hat{\vct{n}}' \ident + i (\hat{\vct{n}} \times \hat{\vct{n}}') \cdot \vct{\sigma}$, it can be seen that $\operatorname{span}_\real \SU(2) \cong \mathbb{H}$ where $\mathbb{H}$ is the associative algebra of quaternions.

  On the other hand, for $\SU(d)$ with $d > 2$, the element $\e^{\i 2 \pi / d} \ident \in \SU(d)$. It follows that
  \begin{equation}\label{eq:d3span}
    d \geq 3 \implies \operatorname{span}_\real \SU(d) = \operatorname{span}_\complex \SU(d),
  \end{equation}
  since any complex scalar is a real number plus a real multiple of the root of unity $\e^{\i 2 \pi / d}$. Similarly,
  \begin{equation}
    \operatorname{span}_\real \U(2) = \operatorname{span}_\complex \U(2),
  \end{equation}
\end{subequations}
so this phenomenon does not occur for $\U(2)$ subgroups in the diagonal blocks of $\mathcal{W}$, or, more generally, for any subgroups that contain phases that are not real.

\subsection{Strategies for isolation}\label{sec:strat}

In this section we consider a number of strategies for isolating Lie algebra elements to off-diagonal blocks of $\hilbert = \bigoplus_{i = 1}^m \hilbert_i$. The strategies we discuss are summarized in \cref{fig:isolate}. For the following, fix a group of unitaries $\mathcal{W} \subseteq \U(\hilbert)$ and let $\mathcal{W}_i \subseteq \mathcal{W}$ be the subgroups defined in \cref{eq:isub}, which necessarily act as the identity on $\hilbert_j$ for $j \neq i$.

\begin{table*}
  \begin{tblr}{hline{2-Y}={solid}, vline{2-Y}={solid}, colspec={ccQ[l,m]c}, column{3}={5cm}, measure=vbox, stretch=-1}
    $\displaystyle \int_{\mathcal{V}_i} \mathop{}\!\mathrm{d} U_i \, f(U_i) U_i X U_i^\dagger$ & & assumptions & \\
    $D_{U_i}(X) = U_i X U_i^\dagger - X$ &
    $\begin{+pmatrix}[hline{2-Y}={solid}, vline{2-Y}={solid}, cells={red!30}, row{3}={blue!20}, column{3}={blue!20}]
      0 & & D_{U_i}(X_{1i}) & \\
      & \ddots & \vdots & \\
      D_{U_i}(X_{i1}) & \cdots & D_{U_i}(X_{ii}) & \cdots \\
      & & \vdots & \ddots
    \end{+pmatrix}$ &
    \setlength{\leftmargini}{.3cm} \begin{itemize}[nosep]
    \item $\mathcal{V}_i$ nontrivial
    \end{itemize} &
    \Cref{eq:DUi} \\
    $(\mathcal{I} - \mathcal{E}_i)(X)$ &
    $\begin{+pmatrix}[hline{2-Y}={solid}, vline{2-Y}={solid}, cells={red!30}, row{3}={white}, column{3}={white}]
      0 & & X_{1i} & \\
      & \ddots & \vdots & \\
      X_{i1} & \cdots & \SetCell{blue!20} X_{ii} - \mathcal{E}_i(X_{ii}) & \cdots \\
      & & \vdots & \ddots
    \end{+pmatrix}$ &
    \setlength{\leftmargini}{.3cm} \begin{itemize}[nosep]
    \item $\mathcal{V}_i$ is compact, with no nonzero invariant vectors in $\hilbert_i$
    \item $\mathcal{E}_i$ is uniform Haar integral over $\mathcal{V}_i$
    \end{itemize} &
    \Cref{eq:mHaar} \\
    $f(U_i) = \operatorname{tr} (B^\dagger U_i + B U_i^\dagger)$ &
    $\begin{+pmatrix}[hline{2-Y}={solid}, vline{2-Y}={solid}, cells={red!30}, row{3}={blue!20}, column{3}={blue!20}]
      0 & & X_{1i} B^\dagger & \\
      & \ddots & \vdots & \\
      B X_{i1} & \cdots & \SetCell{red!30} 0 & \cdots \\
      & & \vdots & \ddots
    \end{+pmatrix}$ &
    \setlength{\leftmargini}{.3cm} \begin{itemize}[nosep]
    \item $\mathcal{V}_i$ is compact, acts nontrivially and irreducibly on $\hilbert_i$
    \item $\hilbert_i \otimes \hilbert_i$ does not contain subrepresentation isomorphic to trivial or $\hilbert_i$
    \item $B = \Pi_i B \Pi_i$
    \end{itemize} &
    \Cref{eq:character} \\
    $(D_{U_j} \circ D_{U_i})(X)$ &
    $\begin{+pmatrix}[hline{2-Y}={solid}, vline{2-Y}={solid}, cells={red!30}]
      \ddots & & & & \\
      & 0 & & \SetCell{blue!20} \tilde{X}_{ij} & \\
      & & \ddots & & \\
      & \SetCell{blue!20} \tilde{X}_{ji} & & 0 & \\
      & & & & \ddots
    \end{+pmatrix}$ &
    \setlength{\leftmargini}{.3cm} \begin{itemize}[nosep]
    \item $\mathcal{V}_i$ and $\mathcal{V}_j$ nontrivial
    \item $\tilde{X}_{ij} = (\ident - U_i) X_{ij} (\ident - U_j^\dagger)$ and $\tilde{X}_{ji} = -\tilde{X}_{ij}^\dagger$
    \end{itemize} &
    \Cref{eq:Doff} \\
    $((\mathcal{I} - \mathcal{E}_j) \circ (\mathcal{I} - \mathcal{E}_i))(X)$ &
    $\begin{+pmatrix}[hline{2-Y}={solid}, vline{2-Y}={solid}, cells={red!30}]
      \ddots & & & & \\
      & 0 & & \SetCell{white} X_{ij} & \\
      & & \ddots & & \\
      & \SetCell{white} X_{ji} & & 0 & \\
      & & & & \ddots
    \end{+pmatrix}$ &
    \setlength{\leftmargini}{.3cm} \begin{itemize}[nosep]
    \item $\mathcal{V}_i$ and $\mathcal{V}_j$ are compact with no nonzero invariant vectors in $\hilbert_i$ and $\hilbert_j$, respectively
    \item $\mathcal{E}_i$ and $\mathcal{E}_j$ are uniform Haar integrals over $\mathcal{V}_i$ and $\mathcal{V}_j$, respectively
    \end{itemize} &
    \Cref{eq:Haaroff}
  \end{tblr}
  \caption{This table shows various isolation strategies for block $i$ of $\hilbert = \bigoplus_{j = 1}^m \hilbert_j$. The first three rows can be understood as an integral weighted by a real-valued function $f$ over the group $\mathcal{V}_i \subseteq \mathcal{W}_i \subseteq \U(\hilbert_i)$ which acts nontrivially on $\hilbert_i$ and trivially on $\hilbert_j$ for $j \neq i$ (really $f$ is to be understood as a distribution: for instance the Dirac delta distribution $\delta_{U_i}$, satisfying $\int \diff U_i' \, \delta_{U_i} K(U_i') = K(U_i)$, is used in the first row). The last two rows have two such integrals, over $\mathcal{V}_i$ and $\mathcal{V}_j$. The color scheme for blocks has (1) white background for unchanged blocks; \textcolor{red}{(2)} red background for blocks which necessarily become $0$; and \textcolor{blue}{(3)} blue background for blocks which may be changed but are not necessarily $0$ (the exact value depending on $X$ and $f$).\label{fig:isolate}}
\end{table*}

\subsubsection{Using one nontrivial group element: \texorpdfstring{$D_{U_i}$}{D}}\label{sec:DU}

For $X \in \lie{w}$, let $X_{jk} = \Pi_j X \Pi_k$. Then for $U_i \in \mathcal{W}_i$,
\begin{equation}\label{eq:adUi}
  U_i X U_i^\dagger = \sum_{j, k} U_i X_{jk} U_i^\dagger = \sum_{j, k} (U_i)^{\delta_{i j}} X_{jk} (U_i^\dagger)^{\delta_{ik}} \in \lie{w},
\end{equation}
where $X_{ij} = \Pi_i X \Pi_j$ and $\delta_{ij}$ is the Kronecker delta. Written in block-matrix form,
\begin{equation}
  U_i X U_i^\dagger = \begin{+pmatrix}[row{3}={blue!20},column{3}={blue!20}]
    X_{11} & \vline & \vline X_{1i} U_i^\dagger & \vline \\ \hline
    & \ddots & \vdots & \\ \hline
    U_i X_{i1} & \cdots & U_i X_{ii} U_i^\dagger & \cdots \\ \hline
    & & \vdots & \ddots
  \end{+pmatrix} \in \lie{w},
\end{equation}
where the blue-shaded region denotes all the blocks which are possibly distinct from $X$, namely those in the $i$th row and $i$th column. This immediately implies that subtracting off $X$ will set all non-shaded blocks to zero. Suggested by this, for any unitary $U$ and operator $X$, define
\begin{equation}
  D_U(X) = U X U^\dagger - X.
\end{equation}
To first order, the superoperator $D_U$ can be understood as a derivation: if $U = \e^{\varepsilon B}$ then expanding in $\varepsilon$ gives $D_U(X) = \varepsilon [B, X] + O(\varepsilon^2)$.

If $X \in \lie{w}$ and $U \in \mathcal{W}$, then $D_U(X) \in \lie{w}$ since it is a sum of terms of $\lie{w}$ with real coefficients. Hence $D_U$ can possibly be used as a tool for isolating to particular blocks. To reiterate, $D_{U_i}(X)$ has a large number of blocks which are guaranteed to be zero. In block-matrix form, $D_{U_i}(X)$ can be nonzero only in the blue-shaded region of
\begin{equation}\label{eq:DUi}
  D_{U_i}(X) = \begin{+pmatrix}[row{3}={blue!20},column{3}={blue!20}]
    \SetCell{red!30} 0 & \SetCell{red!30} \vline & \vline D_{U_i}(X_{1i}) & \SetCell{red!30} \vline \\ \hline
    \SetCell{red!30} & \SetCell{red!30} \ddots & \vdots & \SetCell{red!30} \\ \hline
    D_{U_i}(X_{i1}) & \cdots & D_{U_i}(X_{ii}) & \cdots \\ \hline
    \SetCell{red!30} & \SetCell{red!30} & \vdots & \SetCell{red!30} \ddots
  \end{+pmatrix} \in \lie{w}.
\end{equation}
Furthermore, when $j \neq i$, $D_{U_i}(X_{ij}) = (U_i - \ident) X_{ij}$, and $D_{U_i}(X_{ji}) = -D_{U_i}(X_{ij})$ since $X \in \lie{w}$ is anti-Hermitian. If  $X_{ij} \neq 0$  and if $\hilbert_i$ does not contain any non-zero invariant vector under  
 $\mathcal{W}_i$, which is always the case if $\mathcal{W}_i$ acts nontrivially and irreducibly on $\hilbert_i$,  then it is always possible to find $U_i \in \mathcal{W}_i$ such that $U_i X_{ij} \neq X_{ij}$, so that $D_{U_i}(X_{ij}) \neq 0$. 
 (Recall that if $\mathcal{W}_i$ acts nontrivially and irreducibly on $\hilbert_i$, for any vector $0 \neq |\psi\rangle \in \hilbert_i$, there exists $U_i\in \mathcal{W}_i$, such that $U_i|\psi\rangle\neq |\psi\rangle$. This, in turn, implies that for any nonzero operator $X_{ij} = \Pi_i X_{ij}$, there exists $U_i\in \mathcal{W}_i$ such that 
 $(U_i - \ident) X_{ij}\neq 0$.)

The utility of $D_{U_i}$ comes in two forms:
\begin{enumerate}[i)]
\item First, the possibly-remaining block-diagonal element $D_{U_i}(X_{ii}) \in \su(\hilbert_i)$ since it is traceless and anti-Hermitian, although any of the off-diagonal blocks in the $i$th row or $i$th column may be zero or nonzero. If $\dim \hilbert_i = 1$, then, if nonzero, this implies $D_{U_i}(X) \in \lie{w}$ is isolated to block $i$. Otherwise, if $\su(\hilbert_i) \subset \lie{w}$ then $D_{U_i}(X) - D_{U_i}(X_{ii}) \in \lie{w}$ is purely off-diagonal (i.e. isolated).
\item Second, by sequentially applying $D_{U_j} \circ D_{U_i}$ for distinct blocks $i \neq j$, we obtain an operator which, when nonzero, is guaranteed to be purely off-diagonal and isolated to just these two blocks. (It is worth noting that, since $[U_i, U_j] = 0$, also $[D_{U_i}, D_{U_j}] = 0$, i.e. it does not matter in which order they are applied.)
\end{enumerate}
To clarify the second point, note that the only surviving terms have to be in the overlap of the blue-shaded regions corresponding to $D_{U_i}$ and $D_{U_j}$,
\begin{equation}\label{eq:Doff}
  D_{U_j} \circ D_{U_i}(X) = \begin{+pmatrix}[cells={red!30}]
    \ddots & \vline & \vline & \vline & \vline \\ \hline
    & 0 & & \SetCell{blue!20} \tilde{X}_{ij} & \\ \hline
    & & \ddots & & \\ \hline
    & \SetCell{blue!20} \tilde{X}_{ji} & & 0 & \\ \hline
    & & & & \ddots
  \end{+pmatrix} \in \lie{w}.
\end{equation}
That is, only $D_{U_j} \circ D_{U_i}(X_{ij})$ and $D_{U_j} \circ D_{U_i}(X_{ji})$ can possibly be nonzero. Once again, if $X_{ij} \neq 0$ and $U_i$ and $U_j$ come from nontrivial groups which do not have nonzero invariant vectors in their respective subspaces, the group elements can always be chosen so that the result is nonzero.

\subsubsection{Integration over the uniform Haar measure}\label{sec:haar}

The scheme using $D_{U_i}$ ensures that the only nonzero blocks are in the $i$th row and $i$th column; however, any of the blocks may be changed (possibly becoming zero). We next discuss a technique using the Haar measure which does the same, except that it leaves all the off-diagonal blocks invariant. First, in full generality, we notice that integrating the adjoint action (\cref{eq:adjoint}) of $\mathcal{W}$ on its Lie algebra $\lie{w}$ over group elements with real coefficients stays in the Lie algebra, since it is a closed (real) subspace of the set of all operators.

\begin{proposition}\label{lem:haar}
  Let $\mathcal{W}$ be a matrix Lie group and $\mathcal{S} \subseteq \mathcal{W}$ a subset with a measure $\mathrm{d} s$ (e.g. a finite set with the counting measure), and let $f(s)$ be a real-valued integrable function (or, more generally, a real-valued distribution) on $\mathcal{S}$. For any matrix $A$, if $\e^{t A} \in \mathcal{W}$ for all $t \in \real$, then $\e^{t \tilde{A}} \in \mathcal{W}$ for all $t \in \real$, where
  \begin{equation}
    \tilde{A} = \int \mathop{}\!\mathrm{d} s \, f(s) s A s^{-1}.
  \end{equation}
\end{proposition}

In particular, assuming that there is a nontrivial compact subgroup $\mathcal{V}_i \subseteq \mathcal{W}_i \subseteq \mathcal{W}$ which does not have any nonzero invariant vectors on $\hilbert_i$, we will apply this to
\begin{equation}
  \mathcal{E}_i(X) = \int_{\mathcal{V}_i} \mathop{}\!\mathrm{d} U_i \, U_i X U_i^\dagger \in \lie{w}.
\end{equation}
Then $\mathcal{E}_i(X)$ is zero in the red-shaded region and the same as $X$ everywhere else except for the $i$th diagonal block:
\begin{equation}
  \mathcal{E}_i(X) = \begin{+pmatrix}
    X_{11} & \vline & \vline \SetCell{red!30} 0 & \vline \\ \hline
    & \ddots & \SetCell{red!30} \vdots & \\ \hline
    \SetCell{red!30} 0 & \SetCell{red!30} \cdots & \SetCell{blue!20} \mathcal{E}_i(X_{ii}) & \SetCell{red!30} \cdots \\ \hline
    & & \SetCell{red!30} \vdots & \ddots
  \end{+pmatrix} \in \lie{w}.
\end{equation}
This follows because
\begin{equation}
  \int_{\mathcal{V}_i} \mathop{}\!\mathrm{d} U_i \, U_i \Pi_i = \int_{\mathcal{V}_i} \mathop{}\!\mathrm{d} U_i \, U_i^\dagger \Pi_i = 0,
\end{equation}
is the projector to the subspace of invariant vectors of $\hilbert_i$, which consists only of the zero vector. Letting $\mathcal{I}$ be the identity superoperator, $\mathcal{I}(X) = X$ for all $X$, it follows that
\begin{equation}\label{eq:mHaar}
  (\mathcal{I} - \mathcal{E}_i)(X) = \begin{+pmatrix}
    \SetCell{red!30} 0 & \SetCell{red!30} \vline & \vline X_{1i} & \SetCell{red!30} \vline \\ \hline
    \SetCell{red!30} & \SetCell{red!30} \ddots & \vdots & \SetCell{red!30} \\ \hline
    X_{i1} & \cdots & \SetCell{blue!20} X_{ii} - \mathcal{E}_i(X_{ii}) & \cdots \\ \hline
    \SetCell{red!30} & \SetCell{red!30} & \vdots & \SetCell{red!30} \ddots
  \end{+pmatrix} \in \lie{w}.
\end{equation}
Note that $\Tr X_{ii} = \Tr \mathcal{E}_i(X_{ii})$. Thus $X_{ii} - \mathcal{E}_i(X_{ii})$ is traceless and anti-Hermitian, hence an element of $\su(\hilbert_i)$. If $\su(\hilbert_i) \subseteq \lie{w}$, then
\begin{equation}\label{eq:Haariso}
  (\mathcal{I} - \mathcal{E}_i)(X-X_{ii}) =(\mathcal{I} - \mathcal{E}_i)(X) - (X_{ii} - \mathcal{E}_i(X_{ii})) \in \lie{w}
\end{equation}
is an element isolated to the $i$th block. Regardless, if a compact subgroup $\mathcal{V}_j \subseteq \mathcal{W}_j$ also does not have invariant nonzero vectors on $\hilbert_j$ for some $j \neq i$, then
\begin{equation}\label{eq:Haaroff}
  (\mathcal{I} - \mathcal{E}_j) \circ (\mathcal{I} - \mathcal{E}_i)(X) = \begin{+pmatrix}[cells={red!30}]
    \ddots & \vline & \vline & \vline & \vline \\ \hline
    & 0 & & \SetCell{white} X_{ij} & \\ \hline
    & & \ddots & & \\ \hline
    & \SetCell{white} X_{ji} & & 0 & \\ \hline
    & & & & \ddots
  \end{+pmatrix} \in \lie{w}
\end{equation}
is isolated to blocks $i$ and $j$.

Another use-case involves using the Haar measure on multiple blocks to set many off-diagonal blocks to $0$, when there are many compact $\mathcal{V}_i \subseteq \mathcal{W}_i$ which do not have nonzero invariant vectors. This can be helpful if, for instance, there is a block $j$ such that $\mathcal{V}_j$ is trivial, but all other blocks are irreducible (for instance, as an alternative way to prove \cref{lem:twoblock2} in the case needed for semi-universality, i.e. when one of the blocks is one-dimensional). As a specific example, suppose that $\SU(\hilbert_j) \subseteq \mathcal{W}$ for $j \geq 2$, or, in other words, $\su(\hilbert_j) \subseteq \lie{w}$. Then, by integrating over each of these,
\begin{equation}\label{eq:traceblocks}
  (\mathcal{E}_m \circ \dots \circ \mathcal{E}_2)(X) = \begin{+pmatrix}
    X_{11} & \vline \SetColumn{red!30} 0 & \vline \SetColumn{red!30} & \vline \SetColumn{red!30} 0 \\ \hline
    \SetRow{red!30} 0 & \SetCell{blue!20} \mathcal{E}_2(X_{22}) & & 0 \\ \hline
    \SetRow{red!30} & & \SetCell{blue!20} \ddots & \\ \hline
    \SetRow{red!30} 0 & 0 & & \SetCell{blue!20} \mathcal{E}_m(X_{mm})
  \end{+pmatrix} \in \lie{w}.
\end{equation}
But then, since $X_{ii} - \mathcal{E}_i(X_{ii}) \in \su(\hilbert_i) \subseteq \lie{w}$ for $i \geq 2$, it follows that
\begin{equation}
  \sum_{i = 1}^m X_{ii} = \mathcal{E}_m \circ \dots \circ \mathcal{E}_2(X) + \sum_{i = 2}^m \p[\big]{X_{ii} - \frac{\Tr X_{ii}}{\dim \hilbert_i} \Pi_i} \in \lie{w}.
\end{equation}
And therefore,
\begin{equation}\label{eq:offblocks}
  \tilde{X} = X - \sum_{i = 1}^m X_{ii} = \begin{+pmatrix}
    \SetCell{red!30} 0 & \vline X_{12} & \vline \cdots & \vline X_{1m} \\ \hline
    X_{21} & \SetCell{red!30} 0 & \ddots & X_{2m} \\ \hline
    \vdots & \ddots & \SetCell{red!30} \ddots & \\ \hline
    X_{m1} & X_{m2} & & \SetCell{red!30} 0
  \end{+pmatrix} \in \lie{w}.
\end{equation}
Finally, we may isolate to, e.g., blocks $1$ and $2$:
\begin{equation}\label{eq:twoblock}
  \mathcal{E}_m \circ \dots \circ \mathcal{E}_3(\tilde{X}) = \begin{+pmatrix}[cells={red!30}]
    0 & \vline \SetCell{white} X_{12} & \vline \\ \hline
    \SetCell{white} X_{21} & 0 & \\ \hline
    & & \ddots
  \end{+pmatrix} \in \lie{w}.
\end{equation}

\subsubsection{Weighting with characteristic function}

Suppose that $\mathcal{V}_i \subseteq \mathcal{W}_i$ is compact and acts nontrivially and irreducibly. If we assume also that, as a representation of $\mathcal{V}_i$, $\hilbert_i \otimes \hilbert_i$ does not contain a subrepresentation isomorphic to $\hilbert_i$ or the trivial representation, then we can isolate without requiring any of the other subgroups $\mathcal{V}_j$. Note that the latter assumption is not satisfied for $\mathcal{V}_i \cong \SU(2)$, since $\hilbert_i \otimes \hilbert_i$ contains a trivial subrepresentation for all of its irreps. More generally, this condition is not satisfied if the group $\mathcal{V}_i$ is a subgroup of the orthogonal or symplectic unitary group.

Recall the notion of the characteristic function $\Tr B U(g)$ of the operator $B$ (see e.g. \cite{Marvian_thesis}). We weight the Haar measure by the real part of the characteristic function in the following construction.

\begin{lemma}\label{lem:character}
  Consider a Hilbert space with orthogonal decomposition $\hilbert = \hilbert_\lambda \oplus \hilbert_\perp$ and associated Hermitian projectors $\Pi_\lambda$ and $\Pi_\perp$, where $\hilbert_\lambda$ carries a nontrivial irreducible representation $U_\lambda(g)$ of a compact group $G$. Write $U(g) = U_\lambda(g) \oplus \ident_\perp \in \U(\hilbert)$ for the action on $\hilbert$. If the representation $U_\lambda(g) \otimes U_\lambda(g)$ on $\hilbert_\lambda \otimes \hilbert_\lambda$ does not contain the trivial representation as a subrepresentation, or a subrepresentation isomorphic to $\hilbert_\lambda$, then for any operators $A \in \hilbert[L](\hilbert)$ and $B = \Pi_\lambda B \Pi_\lambda \in \hilbert[L](\hilbert)$,
  \begin{subequations}
    \begin{align}
      \dim \hilbert_\lambda \int \diff g \, \Tr (B U(g)^\dagger) U(g) A U(g)^\dagger & = B A \Pi_\perp \\
      \dim \hilbert_\lambda \int \diff g \, \Tr (B U(g)) U(g) A U(g)^\dagger & = \Pi_\perp A B,
    \end{align}
  \end{subequations}
  where $\diff g$ is the normalized Haar measure on $G$.
\end{lemma}
\begin{remark}\label{rem:center}
  If the irreducible group of unitaries $U_\lambda(G) = \set{U_\lambda(g) \given g \in G}$ has a nontrivial center, then it always satisfies the condition that $\hilbert_\lambda \otimes \hilbert_\lambda$ does not contain a subrepresentation isomorphic to $\hilbert_\lambda$. To see this, note that 
 irreducibility implies that the center of this group should be in the form of phases $\e^{\i \theta} \ident_\lambda \in U_\lambda(G)$.  For two copies of nontrivial phases,  i.e., the unitary $\e^{\i\theta} \mathbb{I}_\lambda\otimes \e^{\i\theta} \mathbb{I}_\lambda$,   all eigenvalues are $\e^{\i2\theta }$, and therefore, there is no subspace of $\hilbert_\lambda \otimes \hilbert_\lambda$ in which this unitary acts as $\e^{\i\theta } (\mathbb{I}_\lambda \otimes \mathbb{I}_\lambda)$.
 \end{remark}

\begin{proof}
  Write $A_{ij} = \Pi_i A \Pi_j$ for $i, j = \lambda, \perp$ and expand
  \begin{equation}\label{eq:Aexpand}
    \begin{split}
      U(g) A U(g)^\dagger & = \sum_{i, j} U(g) A_{ij} U(g)^\dagger = U_\lambda(g) A_{\lambda\lambda} U_\lambda(g)^\dagger + U_\lambda(g) A_{\lambda\perp} + A_{\perp\lambda} U_\lambda(g)^\dagger + A_{\perp\perp}.
    \end{split}
  \end{equation}
  According to the Schur orthogonality theorem,  for any pair of irreps $\lambda$ and $\lambda'$ of a compact group $G$    
  \begin{equation}
    \dim \hilbert_\lambda \int \diff g \, U_{\lambda'}(g) \otimes U_\lambda(g)^\dagger =\delta_{\lambda,\lambda'}\ {\operatorname{SWAP}}_\lambda \  ,
  \end{equation}
  where $\delta_{\lambda,\lambda'}$ is 0 if $\lambda$ and $\lambda'$  are inequivalent representations, and is equal to 1 if $\lambda=\lambda'$, and    
  where $\operatorname{SWAP}_\lambda$ is the linear operator on $\hilbert_\lambda \otimes \hilbert_\lambda$ which swaps product states, $\ket{\psi} \otimes \ket{\varphi} \mapsto \ket{\varphi} \otimes \ket{\psi}$. Schur orthogonality plus the assumption that $\hilbert_\lambda \otimes \hilbert_\lambda$ does not contain $\hilbert_\lambda$ or the trivial 
  representation as subrepresentations implies
  \begin{subequations}
    \begin{align}
      \int \diff g \, U_\lambda(g) \otimes U_\lambda(g) & = 0 \\
      \int \diff g \, U_\lambda(g) \otimes U_\lambda(g) \otimes U_\lambda(g)^\dagger & = 0.
    \end{align}
  \end{subequations}
  Similarly, switching $U_\lambda(g) \leftrightarrow U_\lambda(g)^\dagger$ in these expressions also obtains vanishing integrals. Finally, since $U_\lambda(g)$ is assumed nontrivial, $\int \diff g \, U_\lambda(g) = 0$. Thus the only term from \cref{eq:Aexpand} that survives in the following has exactly one $U_\lambda(g)$ and one $U_\lambda(g)^\dagger$ in the integral:
  \begin{equation}
    \begin{split}
      \dim \hilbert_\lambda \int \diff g \, \Tr (B U(g)) U(g) A U(g)^\dagger & = \dim \hilbert_\lambda \int \diff g \, \Tr (B U_\lambda(g)) A_{\perp\lambda} U_\lambda(g)^\dagger \\
      & = \dim \hilbert_\lambda \Tr_{(1)} \int \diff g \, (B \otimes A_{\perp\lambda}) (U_\lambda(g) \otimes U_\lambda(g)^\dagger) \\
      & = \Tr_{(1)} (B \otimes A_{\perp\lambda}) \operatorname{SWAP}_\lambda \\
      & = A_{\perp\lambda} B = \Pi_\perp A B.
    \end{split}
  \end{equation}
  where $\Tr_{(1)}$ is the partial trace\footnote{This can also be shown using $ \dim \mathcal{H}_\lambda\int dg \Tr(B U_\lambda(g)) U^\dag_\lambda(g)=B.$}. The other integral is performed similarly.
\end{proof}

Supposing that the conditions of \cref{lem:character} are satisfied for the action of $\mathcal{V}_i$ on $\hilbert_i$, it follows that, for any $B = \Pi_i B \Pi_i \in \hilbert[L](\hilbert)$,
\begin{equation}\label{eq:character}
  \int_{\mathcal{V}_i} \diff U_i \, \Tr (B^\dagger U_i + B U_i^\dagger)\  U_i X U^\dag_i = \begin{+pmatrix}[row{3}={blue!20},column{3}={blue!20}]
    \SetCell{red!30} 0 & \SetCell{red!30} \vline & \vline X_{1i} B^\dagger & \SetCell{red!30} \vline \\ \hline
    \SetCell{red!30} & \SetCell{red!30} \ddots & \vdots & \SetCell{red!30} \\ \hline
    B X_{i1} & \cdots & \SetCell{red!30} 0 & \cdots \\ \hline
    \SetCell{red!30} & \SetCell{red!30} & \vdots & \SetCell{red!30} \ddots
  \end{+pmatrix} \in \lie{w},
\end{equation}
since the coefficients $\Tr (B^\dagger U_i + B U_i^\dagger)$ are real.

\subsection{Application of strategies: 3 different proofs of Lemma \ref{lem:iso-oneblock}}\label{app:applications}

\begin{proof}[Proof of \cref{lem:iso-oneblock}]
  If $\Pi_i A = 0$ then this is just the statement that $\ident \in \mathcal{W}$. So suppose that $\Pi_i A \neq 0$. Any of the first three rows of \cref{fig:isolate} may be used. In particular, the first two directly use \cref{eq:span-iso} and the fact that $\dim \hilbert_i \geq 3$, so that $\operatorname{span}_\real \SU(\hilbert_i) \Pi_i = \operatorname{span}_\complex \SU(\hilbert_i) \Pi_i = \hilbert[L](\hilbert_i)$ (see \cref{eq:d3span}).
  \begin{enumerate}

  \item Recall the definition $D_U(X)=U X U^\dag- X$. Choose $U_i \in \SU(\hilbert_i)$ such that $U_i - \Pi_i$ is invertible. Then $(U_i - \Pi_i) A \neq 0$. Then $D_{U_i}(A_{ii}) \in \lie{su}(\hilbert_i)$ and therefore $\tilde{A} \coloneqq D_{U_i}(A) - D_{U_i}(A_{ii}) \in \lie{w}$ is isolated to block $i$ (see also \cref{eq:DUi}). 
Recall 
\begin{equation*}
  \set{B X + X B^\dagger \given B \in \operatorname{span}_\real \mathcal{W}_i \Pi_i} = \operatorname{span}_\real \set{U_i X U_i^\dagger \given U_i \in \mathcal{W}_i} \subseteq \lie{w}. \tag{re \ref{eq:span-iso}}
\end{equation*}
which holds under the assumption that $X$ is anti-Hermitian, and inside $\mathfrak{w}$, and   $\Pi_i X+X\Pi_i=X$, i.e., is isolated.  
Since $\dim \hilbert \geq 3$, $\operatorname{span}_\real \SU(\hilbert_i) \Pi_i = \hilbert[L](\hilbert_i)$, which means   $B$ can be arbitrary operator.  Choosing $X= \tilde{A}$ we conclude that for any $\tilde{B} \in \hilbert[L](\hilbert_i)$,
    \begin{equation}
      \tilde{B} \tilde{A} + \tilde{A} \tilde{B}^\dagger = \tilde{B} (U_i - \Pi_i) A \Pi_\perp + \Pi_\perp A (U_i^\dagger - \Pi_i) \tilde{B}^\dagger \in \lie{w}.
    \end{equation}

    In particular, for any $B \in \hilbert[L](\hilbert_i)$ choose $\tilde{B} = B (U_i - \Pi_i)^{-1}$ to obtain \cref{eq:isoLie}.
  \item \textbf{Integration over uniform Haar measure}: By \cref{eq:mHaar}, $(\mathcal{I} - \mathcal{E}_i)(A) \in \lie{w}$. But note that $\Pi_i (\mathcal{I} - \mathcal{E}_i)(A) \Pi_i = A_{ii} - \mathcal{E}_i(A_{ii}) \in \su(\hilbert_i)$ because it is anti-Hermitian and traceless, so
    \begin{equation}
      \Pi_i A \Pi_\perp + \Pi_\perp A \Pi_i = \Pi_i A + A \Pi_i - 2 \Pi_i A \Pi_i = (\mathcal{I} - \mathcal{E}_i)(A) - (A_{ii} - \mathcal{E}_i(A_{ii})) \in \lie{w}.
    \end{equation}
    Thus \cref{eq:span-iso} proves the claim.
  \item \textbf{Weighting with characteristic function}: By \cref{rem:center}, since $\SU(\hilbert_i)$ has a nontrivial center, $\hilbert_i \otimes \hilbert_i$ does not contain a subrepresentation isomorphic to $\hilbert_i$. Furthermore, $\dim \hilbert_i \geq 3$ implies that $\hilbert_i \otimes \hilbert_i$ does not contain a trivial subrepresentation, since it is not self-dual as a representation of $\SU(\hilbert_i)$. Thus, since $\dim \hilbert_i \geq 3$, $\SU(\hilbert_i)$ satisfies all the assumptions of \cref{lem:character}, \cref{eq:character} applies, and it follows immediately.
  \end{enumerate}
\end{proof}

\end{document}